%% file: main.tex
\documentclass[sigconf]{acmart}
\AtBeginDocument{%
  }

\input{packages}
\input{defines}

\input{glossary}

\newif\ifsubmission
\DeclareRobustCommand{\appref}[1]{%
  \ifsubmission
    the full paper~\cite{deman2026hybridarxiv}%
  \else
    #1%
  \fi
}

\begin{document}
\fancyhead{}
\fancyfoot{}

\input{titleauth}

\begin{abstract}
\input{00_abstract}

\end{abstract}

\renewcommand\footnotetextcopyrightpermission[1]{} 

\maketitle
\thispagestyle{empty}

\input{01_introduction}
\input{02_preliminaries}
\input{03_hybrid_streaming}

\input{04_hybrid_dynamic}

\input{05_implementation}
\input{06_experiments}

\input{07_conclusion}

\begin{acks}
    We gratefully acknowledge support from NSF grants CCF-2247577, CCF-2403235, CNS-2317194, IIS-2551616, and NRT-HDR 2125295. M.A.B. was supported in part by Broadcom, Sandia National Laboratories, and the John L. Hennessy Chaired Professorship.
\end{acks}

\balance
\bibliographystyle{ACM-Reference-Format}
\bibliography{main.bib}
\clearpage

\appendix
\input{AA_balloon_proofs}
\input{AB_hybrid_streaming_proofs}
\input{AC_balloon_dc_proofs}
\input{AD_hybrid_dynamic_proofs}
\input{AE_experiments}

\end{document}

%% file: packages.tex
\usepackage{pgfplots}
\usepackage{pgfplotstable}
\pgfplotsset{compat=1.8}
\usepackage{subcaption}
\usepackage{xspace} 
\usepackage[subtle, mathspacing=normal]{savetrees}

\usepackage{xcolor}
\usepackage{amsmath}
\usepackage{amsfonts}
\usepackage{algorithm}
\usepackage{algpseudocodex}
\algrenewcommand{\algorithmicindent}{1em}
\usepackage{float}
\usepackage{siunitx}
\usepackage{graphicx}
\usepackage{thmtools}
\usepackage{thm-restate}

\usepackage[normalem]{ulem} 

\usepackage{blkarray}
\usepackage{balance}  
\usepackage{booktabs} 

\usepackage{marginnote} 
\usepackage{url}
\usepackage{enumitem}
\graphicspath{{./fig/}}
\usepackage{mathtools}

\usepackage{comment}
\usepackage{color}
\usepackage{cleveref}
\usepackage{tikz}
\usepackage{todonotes} 
\usetikzlibrary{fadings}

\usepackage{dirtytalk}
\usepackage{array}

\usepackage[para,online,flushleft]{threeparttable}


%% file: defines.tex
\newcommand{\punt}[1]{}

\newcommand{\R}{\mathbb{R}}
\newcommand{\Field}{\mathbb{F}}

\newcommand{\bigoh}{\mathcal O}
\newcommand{\softoh}{\tilde{\mathcal{O}}}

\newcommand{\polylog}{\mathrm{polylog}}

\newtheorem{definition}{Definition}[section]

\newtheorem{problem}{Problem}

\newtheorem*{question}{Question}

\makeatletter
\def\@copyrightspace{\relax}
\makeatother

\newcommand{\defn}[1]       {{\textit{\textbf{\boldmath #1}}}\xspace}

\newcommand{\etal}{\text{et al}.\xspace}

\date{}

\newcommand{\revision}[1]{\textcolor{blue}{#1}}
\definecolor{violet}{RGB}{120, 0, 180}
\newcommand{\revisionmoved}[1]{\textcolor{blue}{#1}}
\renewcommand{\revisionmoved}[1]{#1}
\renewcommand{\revision}[1]{#1}

\newcommand{\namedcomment}[3]{}

\newcommand{\mab}[1]{\namedcomment{mab}{red}{#1}}
\newcommand{\mfc}[1]{\namedcomment{mfc}{purple}{#1}}

\newcommand{\ahmed}[1]{\namedcomment{ahmed}{blue}{#1}}
\newcommand{\victor}[1]{\namedcomment{victor}{green}{#1}}
\newcommand{\evan}[1]{\namedcomment{evan}{orange}{#1}}

\newcommand{\gil}[1]{\namedcomment{gil}{red}{#1}}
\newcommand{\quinten}[1]{\namedcomment{Quinten}{red}{#1}}

\renewcommand{\mab}[1]{\todo[size=\tiny,color=green!40]{MAB: #1}}
\renewcommand{\mfc}[1]{\todo[size=\tiny,color=green!40]{MFC: #1}}
\renewcommand{\ahmed}[1]{\todo[size=\tiny,color=yellow]{Ahmed: #1}}
\renewcommand{\victor}[1]{\todo[size=\tiny,color=yellow]{Victor: #1}}
\renewcommand{\evan}[1]{\todo[size=\tiny,color=red!40]{Evan: #1}}

\newcommand{\fixme}[1]{\todo[size=\tiny]{#1}}

\newcommand{\inline}[1]{\todo[inline,color=yellow,size=\tiny]{#1}}

\iffalse
\else

\renewcommand{\mab}[1]{}
\renewcommand{\mfc}[1]{}
\renewcommand{\ahmed}[1]{}
\renewcommand{\victor}[1]{}
\renewcommand{\evan}[1]{}
\renewcommand{\fixme}[1]{}
\renewcommand{\inline}[1]{}

\fi 

\renewcommand{\epsilon}{\varepsilon}

\newcommand{\appref}[1]         {Appendix~\ref{app:#1}}

\renewcommand{\eqref}[1]          {Eq.~\ref{eq:#1}}

\newcolumntype{P}[1]{>{\centering\arraybackslash}p{#1}}

\definecolor{bg}{rgb}{0.95,0.95,0.95}

\newcommand{\myparagraph}[1]{\smallskip\noindent{\bf \boldmath #1.}}

\usepackage{xfp}

\newcommand{\sigfig}[1]{%
    \fpeval{
        #1 >= 1e9 ? round(#1/1e9, 2 - floor(ln(#1/1e9)/ln(10))) : (
        #1 >= 1e6 ? round(#1/1e6, 2 - floor(ln(#1/1e6)/ln(10))) : (
        #1 >= 1e3 ? round(#1/1e3, 2 - floor(ln(#1/1e3)/ln(10))) : 
                    round(#1, 2 - floor(ln(#1)/ln(10)))
        ))
    }%
    \ifnum\fpeval{#1 >= 1e9}=1 B\else
    \ifnum\fpeval{#1 >= 1e6}=1 M\else
    \ifnum\fpeval{#1 >= 1e3}=1 K\fi\fi\fi
}

\newcommand{\AvgDeg}[2]{%
    \fpeval{
        round((2*#2)/#1, 2 - floor(ln((2*#2)/#1)/ln(10)))
    }%
}

%% file: glossary.tex
\newcommand{\graph}{\mathcal{G}}
\newcommand{\nodes}{\mathcal{V}}
\newcommand{\edges}{\mathcal{E}}
\newcommand{\nodesize}{V}
\newcommand{\edgesize}{E}

\newcommand{\graphstream}{S}

\newcommand{\sketch}{\mathcal{S}}
\newcommand{\recoverySketch}{\mathcal{R}}
\newcommand{\recover}{\textsc{Recover}}

\newcommand{\prob}[1]{ \Pr \left [ #1 \right ]}
\newcommand{\charvec}[1]{f_{#1}}
\newcommand{\recoveryVec}{x}
\newcommand{\sketchOfVec}{\Tilde{\recoveryVec}}
\newcommand{\expectation}{\mathbb{E}}
\newcommand{\degree}{d}
\newcommand{\numcolumns}{L}
\newcommand{\numSketchColumns}{L}
\newcommand{\buckhash}{\gamma}

\newcommand{\ufo}{\mathcal{U}}

\newcommand{\losslessname}{lossless\xspace}

\newcommand{\sketchdcname}{sketch-based\xspace}

\newcommand{\streamingname}{streaming connectivity\xspace}

\newcommand{\recoverySize}[0]{r}
\newcommand{\decodeAlgo}{\mathcal{D}}

\newcommand{\integersUntil}[1]{\left[ #1 \right]}
\newcommand{\positivesUpTo}[1]{\left\langle #1 \right\rangle}

\newcommand{\Boruvka}{Bor\r{u}vka\xspace}

\newcommand{\universe}{\mathcal{U}}
\newcommand{\universesize}{n}
\newcommand{\supportset}{Z}
\newcommand{\supportsize}{m}

\newcommand{\cutoff}{\rho}

\newcommand{\hashFamily}{\mathcal{H}}
\newcommand{\hashFunc}{h}
\newcommand{\sketchColumn}[1]{\mathcal{C}_{#1}}

\DeclareMathOperator\supp{supp}

\newcommand{\sketchvec}{\recoveryVec}
\newcommand{\buck}{b}
\newcommand{\buckset}{B}

\newcommand{\sketchmatrix}{\Pi}

\newcommand{\tree}{\mathcal{T}}
\newcommand{\forest}{\mathcal{F}}
\newcommand{\component}{\mathcal{C}}
\newcommand{\cutset}{F}

\newcommand{\insertedge}{\textsc{InsertEdge}\xspace}
\newcommand{\deleteedge}{\textsc{DeleteEdge}\xspace}
\newcommand{\connected}{\textsc{Connected}\xspace}

\newcommand{\update}{\textsc{Update}\xspace}
\newcommand{\myinsert}{\textsc{Insert}\xspace}
\newcommand{\mydelete}{\textsc{Delete}\xspace}
\newcommand{\query}{\textsc{Query}\xspace}

\newcommand{\link}{\textsc{Link}\xspace}
\newcommand{\cut}{\textsc{Cut}\xspace}

\newcommand{\insertvertex}{\textsc{InsertVertex}\xspace}
\newcommand{\deletevertex}{\textsc{DeleteVertex}\xspace}

\newcommand{\incidentedges}{\textsc{IncidentEdges}\xspace}
\newcommand{\promotevertex}{\textsc{PromoteVertex}\xspace}
\newcommand{\demotevertex}{\textsc{DemoteVertex}\xspace}

\newcommand{\densethresh}{\delta}
\newcommand{\densealg}{ALG_{D}}
\newcommand{\sparsealg}{ALG_{S}}
\newcommand{\densespace}{MEM_{D}}
\newcommand{\sparsespace}{MEM_{S}}

\newcommand{\gibb}{Gibb's algorithm\xspace}

\newcommand{\cupcake}{CUPCaKE\xspace}
\newcommand{\clusterforest}{Cluster Forest\xspace}
\newcommand{\sketchname}{\textsc{BalloonSketch}\xspace}
\newcommand{\sketchnames}{\textsc{BalloonSketches}\xspace}

 \newcommand{\sysname}{\textsc{HybridSCALE}\xspace} 
 \newcommand{\densesysname}{\textsc{BalloonDC}\xspace}
\newcommand{\graphzep}{\textsc{GraphZeppelin}\xspace}

\newcommand{\cameosketch}{\textsc{CameoSketch}\xspace}

\newcommand{\cubesketch}{\textsc{CubeSketch}\xspace}

\newcommand{\iblt}{\textsc{IBLT}\xspace}
\newcommand{\iblts}{\textsc{IBLTs}\xspace}

%% file: titleauth.tex


\title{Hybrid Sketching Methods for\\Dynamic Connectivity on Sparse Graphs}


\author{Quinten De Man}
\authornote{Both authors contributed equally to this research.}
\affiliation{
    \institution{University of Maryland}
    \city{College Park}
    \state{MD}
    \country{USA}
}
\email{deman@umd.edu}

\author{Gilvir Gill}
\authornotemark[1]
\affiliation{
    \institution{Stony Brook University}
    \city{Stony Brook}
    \state{NY}
    \country{USA}
}
\email{gigill@cs.stonybrook.edu}

\author{Michael A. Bender}
\affiliation{
    \institution{Stony Brook University}
    \city{Stony Brook}
    \state{NY}
    \country{USA}
}
\email{bender@cs.stonybrook.edu}

\author{Laxman Dhulipala}
\affiliation{
    \institution{University of Maryland}
   \city{College Park}
   \state{MD}
   \country{USA}
}
\email{laxman@umd.edu}

\author{David Tench}
\affiliation{
    \institution{Lawrence Berkeley National Lab}
    \city{Berkeley}
    \state{CA}
    \country{USA}
}
\email{dtench@pm.me}




%% file: 00_abstract.tex
Dynamic connectivity is arguably the most basic and fundamental dynamic graph problem, and recent algorithmic breakthroughs on \emph{dynamic graph sketching} have reshaped what is theoretically possible for the problem: by encoding the graph as per-vertex linear sketches, these algorithms solve dynamic connectivity in only $\Theta(V \log^2 V)$ space, independent of the number of edges, with an asymptotic advantage over lossless $\Theta(V+E)$-space structures that grow as the graph becomes denser.
Prior to this work, no practical dynamic connectivity algorithm has been able to translate these theoretical breakthroughs into space savings on real-world graphs.
The main obstacle is that per-vertex sketches cost thousands of bytes per vertex, so sketching only pays off once the graph becomes extremely dense.

Our starting observation is that real-world graphs that are sparse on average are often not uniformly sparse---these graphs can contain dense cores on a small subset of the vertices that accounts for a large fraction of the edges.
We exploit this structure via \emph{hybrid sketching}: sketch only the dense core, and store the sparse periphery losslessly.

On the theoretical side, we design new hybrid algorithms for both fully-dynamic and semi-streaming connectivity with space $O(\min\{V+E,\; V \log V \log(2+E/V)\})$ w.h.p., simultaneously matching the lossless bound on sparse graphs, the sketching bound on dense graphs, and improving on both in an intermediate regime.
A key technical ingredient is \sketchname, a new $\ell_0$-sampler that reduces per-vertex sketch sizes by up to $8\times$ on real-world graphs.
We give a practical C++ implementation of our new algorithms in \sysname, a modular system that treats the lossless and sketch-based components as subroutines.
To our knowledge, \sysname is the first sketch-based dynamic connectivity system to save space on commonly studied real-world graphs.
Compared to the state-of-the-art lossless baseline, \sysname uses up to 21\% less space on sparse real-world graphs (average degree < 100), up to 92\% less space on intermediate density graphs (average degree $\approx$ 100--1000), and up to 97\% less space on synthetic dense graphs (average degree > 1000).

%% file: 01_introduction.tex
\section{Introduction}

Graph data and graph algorithms play a fundamental role in modern computing, \revision{supporting} applications in web search, transactional systems, and scientific computing, among other areas.
In \revision{many} of these settings, the underlying graph is inherently dynamic: edges are continually inserted and deleted into the graph as the state of the system evolves.
%
Maintaining basic structural queries on a graph as it changes is \revision{thus} a foundational primitive for downstream analytics, and \emph{dynamic connectivity} (answering whether two vertices currently lie in the same connected component as the edge set changes) \revision{turns out to be one of} the most fundamental and well-studied dynamic primitives~\cite{frederickson1985data,eppstein1997sparsification,henzinger1997sampling,henzinger1999randomized,holm2001polylog,thorup2000near,kapron2013dynamic,wulff2013faster,gibb2015dynamic,nanongkai2017dynamic,holm2018dynamic,huang2023fully,chuzhoy2020deterministic}.
Despite a four-decade old theoretical foundation, translating \revision{this line of research} to practical systems has proven  difficult, and dynamic connectivity remains an active frontier for both theoretical and practical research~\cite{deman2024towards,deman2025fast,IDTree}.

%
\revision{
Traditionally, dynamic connectivity algorithms use \emph{lossless}
graph representations where every edge is explicitly stored.
%
Lossless systems use $\Theta(\nodesize+\edgesize)$ space, which is space-efficient for sparse inputs, where $\edgesize$ is not much larger than $\nodesize$.
}

%
\revision{
A more recent approach is to use \defn{graph sketching}, which solves dynamic connectivity in space $\Theta(V \log^2 V)$. 
What is exciting about this bound is that the space scales only with the number of vertices, not the number of edges.
}
\revision{Thus}, the attraction and promise of graph sketching is that it remains space-efficient even as the graph becomes arbitrarily dense, while lossless methods would continue to use space linear in the number of edges.

Sketching techniques have found massive success across computer science, particularly in streaming algorithms where sketches exponentially reduce space requirements~\cite{alonFrequencyMoments1999,cordmodCountMinSketch2005,karpStreamingAlgorithms2003,efraimidisWeightedSampling2006}\revision{.}
However, in \emph{graph} sketching~\cite{woodruffSketchingToolNumerical2014,countsketch,misraGriesSummary1982,Ahn2012}, \revision{every} vertex requires its own sketch, \revision{netting} $\Theta(\mathsf{polylog}(V))$ words \revision{per vertex} \revision{and} 
translating to thousands of bytes per vertex \revision{in practice}.
\revision{As a result,} sketching \revision{empirically} beats storing edges in an adjacency list only when the average degree is at least in the thousands.
\revision{In contrast,} degrees in real-world graphs are typically in the hundreds, well below this break-even point.
\revision{Consequently}, existing sketching systems~\cite{graphzeppelin, landscape} have demonstrated space savings
on synthetic dense inputs but have used orders-of-magnitude more space on sparse inputs compared to lossless systems.
\revision{For instance, the state-of-the-art sketching algorithm is up to $9.6\times$ \emph{more} space efficient on dense graphs than the most compact lossless system, but also up to $152\times$ \emph{less} space efficient on sparse graphs (see Figure ~\ref{fig:memory_experiment_results}). 
An ideal system would perform well on all types on graphs, and would leverage the power of both lossless and sketch-based paradigms.
This leads to our key question:}
%
%
\revision{
\begin{question}
    Can we design a dynamic connectivity system that performs as well as lossless systems on sparse graphs and graph sketching systems on dense graphs?
    Moreover, can graph sketching techniques be leveraged to improve the space efficiency of dynamic connectivity systems on real-world sparse graphs?
\end{question}
}
\myparagraph{\revision{Technical Insights}}
\revision{The key insight that enables our approach is that} \defn{graphs that are sparse on average are typically not uniformly sparse}.
For example, social networks, web graphs, citation graphs, and collaboration networks, among others, are sparse on average but often contain a small, densely interconnected \defn{core} that accounts for a large fraction of the edges~\cite{seidman1983network,faloutsos1999power,borgatti2000models,newman2003structure,alvarezhamelin2005kcore,carmi2007model,mislove2007measurement,kitsak2010identification,batagelj2003om}.
On such graphs, the cost of sketching vertices \revision{would be  worthwhile in the dense core}, but would be prohibitive if blindly applied to every vertex.
This observation suggests a natural strategy: sketch only the dense core, and store the sparse periphery losslessly.
We refer to this approach as \defn{hybrid sketching}.
Although the idea is straightforward to state, realizing it both in theory and in practice requires carefully \revision{combining} the graph sketching system, the lossless structure, and the logic that routes edges between them.
There are three major challenges \revision{for} designing a hybrid sketching method:

\begin{enumerate}[itemsep=0pt, leftmargin=*]

\item
\revision{
The overall challenge of the paper is to design a scheme for dynamically maintaining a space-efficient partition of the graph into a sketched core and a lossless periphery that has good performance both in theory and in practice.
Crucially, the scheme should admit good update throughput and low query latency.
}


\item Since the graph is dynamic, the boundary between the dense core and the sparse periphery may change over time.
Moving edges from the core \revision{back} to the periphery requires recovering \revision{\emph{all}} edges from the sketch data structure.
However, classical connectivity sketches do not guarantee returning all edges, only some of the sketched support.
Hybrid sketching therefore requires a sparse recovery sketch that returns the full support of a sparse vector with high probability.

\item The state-of-the-art implementations of $\ell_0$-samplers, the main building block underlying graph sketching, have a space overhead of thousands of bytes per-vertex.
As a result, these sketches don't save space over a lossless representation until the degree is in the thousands.
\revision{However}, commonly-studied real-world graphs rarely have average degree higher than the low hundreds.

\end{enumerate}
\revision{
We overcome Challenge (2) using Invertible Bloom Lookup Tables (\iblts)~\cite{goodrich2011iblt,belazzougui2024iblt}, which provide a sketching primitive to recover the entire support of a sparse vector with high probability.
}
\revision{
We address Challenge (3) by designing an $\ell_0$-sampler that is especially space-efficient when sketching sparse vectors, while 
still saving a $2$-factor of space when sketching a vertex with $\Theta(V)$ incident edges.
Our overall solution to Challenge (1) is a novel hybrid algorithm and corresponding system (\sysname) that promotes and demotes vertices to/from the graph sketching sub-system using degree-based thresholding.  
Although the high-level idea of our approach is intuitive, making it both theoretically- and practically-efficient involves highly non-trivial techniques that we describe in the remainder of the introduction. 
}

\subsection{Our Contributions}


\myparagraph{\revision{Theoretical Results}} 
Our main theoretical result is a new hybrid sketching algorithm for \revision{dynamic} connectivity with strong space complexity on both sparse and dense graphs.

\begin{restatable}[Hybrid Dynamic Connectivity]{theorem}{hybriddynamic} \label{thm:hybrid_dynamic_connectivity}
There exists a dynamic connectivity algorithm with space complexity
$$
O\!\left(
\min \left\{
\nodesize + \edgesize,\;
\nodesize \log(\nodesize) \log\!\left(2 + \frac{\edgesize}{\nodesize}\right)
\right\}
\right)
\text{ w.h.p.}
$$
The update complexity is amortized $O(\log^4 \nodesize)$. The query complexity is worst-case $O(\log \nodesize / \log \log \nodesize)$.
\end{restatable}

The key property of \revision{our algorithm} is that its space complexity simultaneously matches the asymptotic space of lossless methods on sparse inputs and of sketching-based methods on dense inputs.
\revisionmoved{In fact in a certain density range, our algorithm uses asymptotically less memory usage than all prior algorithms.
Specifically, when the input graph has $\Omega(\nodesize^{1 + \epsilon})$ edges for $0 < \epsilon \leq 1$, our algorithm uses $O(\nodesize \log^2 \nodesize)$ words of space, matching the space of Ahn \etal ~\cite{Ahn2012} (sketching approach)}\revision{, and the lower bound of Nelson~\etal~\cite{nelson2019optimal}}.
\revisionmoved{
When the input graph has $\bigoh(\nodesize \log \nodesize \log \log \nodesize)$ edges, our algorithm uses the same space as an adjacency list representation (lossless approach): $O(\nodesize + \edgesize)$.
When the input graph has $o(\nodesize^{1 + \epsilon})$ edges and $\omega(\nodesize \log \nodesize \log \log \nodesize)$ edges, our algorithm uses asymptotically less space than both Ahn \etal and an adjacency list representation.
}

We additionally give a hybrid algorithm for the \emph{semi-streaming}
model, in which the graph arrives as a stream of edge insertions and
deletions under an $O(\nodesize \polylog(\nodesize))$ memory budget, and connectivity queries are answered from the final sketch rather than continuously during the stream~\cite{semistreaming1,semistreaming2,Ahn2012,AhnGM12b,mcgregor2014graph}. This is the classical setting in which graph sketching was originally developed, and it remains the standard benchmark for evaluating sketch-based space complexity. Our algorithm achieves the same hybrid space bound as in the fully-dynamic setting, with asymptotically faster update complexity in exchange for slower one-shot queries.

\begin{restatable}[Hybrid Streaming Connectivity]{theorem}{hybridstreaming} \label{thm:hybrid_streaming_connectivity}
There exists a semi-streaming algorithm for connectivity on polynomial-length input streams of edge updates with space complexity
$$
O\!\left(
\min \left\{
\nodesize + \edgesize,\;
\nodesize \log(\nodesize) \log\!\left(2 + \frac{\edgesize}{\nodesize}\right)
\right\}
\right)
\text{ w.h.p.}
$$
The update complexity is amortized $\Theta(\log \nodesize)$ w.h.p. The query complexity is worst-case $\Theta(\nodesize \log^2 \nodesize)$.
\end{restatable}



The hybrid bounds above rest on a new sketching primitive, \sketchname, which we believe is of independent interest.
This result improves on the best known $\ell_0$-samplers, which use $\Theta(L \log \universesize)$ space, by replacing the universe-size dependence with a dependence on the true support size $\supportsize$. 
This improved sketch algorithm results in up to a $8\times$ reduction in per-vertex sketch size on sparse real-world graphs. We find that this space reduction is necessary for our hybrid sketching algorithm to be space-efficient on these inputs, \revision{and primarily addresses Challenge (3)}.


\begin{restatable}[\sketchname]{theorem}{balloonsketchthm} \label{thm:balloon_sketch}
There exists an $\ell_0$-sampler, \sketchname, such that given $L$ independent \sketchnames, each sketching the same vector $\recoveryVec \in \Field^{\universesize}$ with $\supportsize = \|\recoveryVec\|_0$:
\begin{enumerate}[itemsep=0pt,leftmargin=15pt]
    \item The total space usage of the sketches is $O(L \log \supportsize + \log \universesize)$ w.h.p.
    \item The total cost of one update to each sketch is $O(L + \log \universesize)$ w.h.p.
\end{enumerate}
\end{restatable}

\begin{figure*}[ht]
    \centering
    \includegraphics[width=\linewidth]{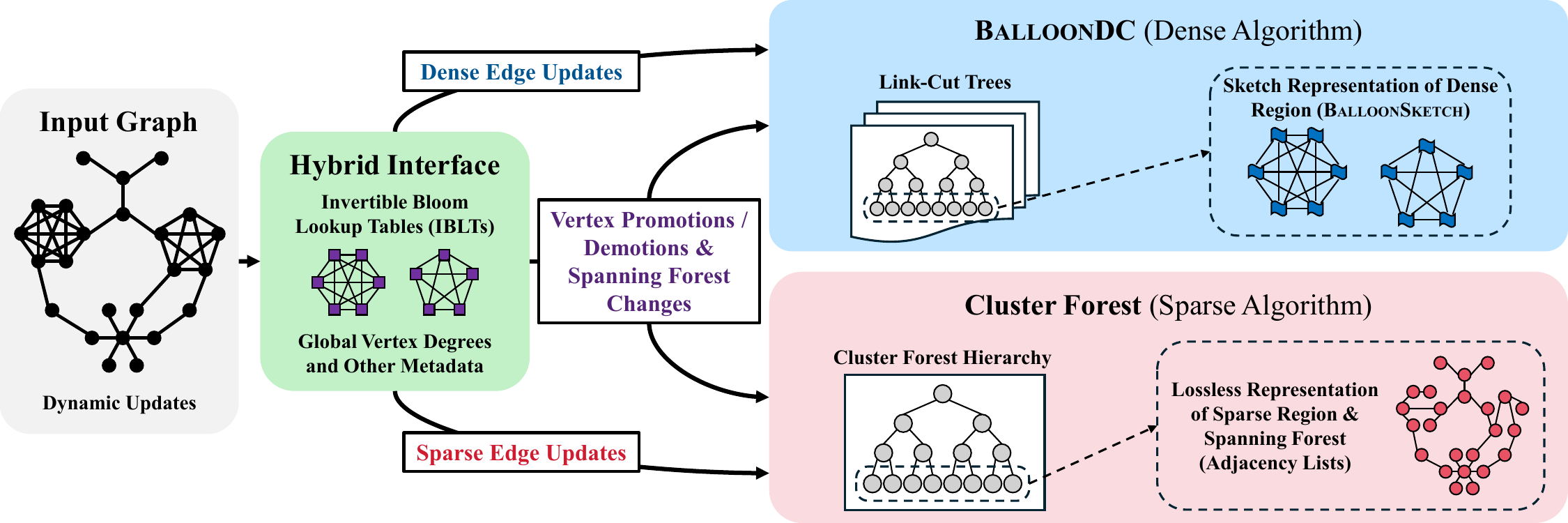}
    \vspace{-0.5em}
    \caption{\revision{
    A top-down view depicting the various components of our \sysname system for dynamic connectivity.
    }}
    \label{fig:hybridscale_overview}
\end{figure*}

\footnotetext[1]{Code is at \url{https://github.com/GraphStreamingProject/HybridDynamicQueriesCC}}

\myparagraph{The \sysname system}
\revision{To solve Challenge (1), we propose and implement \sysname\footnotemark[1], a modular dynamic connectivity system whose design separates a lossless and a sketching-based component and manages communication between the two.}
Any lossless dynamic connectivity structure can be plugged in on one side, and any sketch-based one on the other.
%
Our implementation uses an implementation of \clusterforest~\cite{acar2019parallel,acar2020parallel,deman2024towards} on the lossless side.
For the sketch-based side, we implemented a new system, \densesysname, combining ideas from \gibb~\cite{gibb2015dynamic}, \cupcake~\cite{deman2025fast}, and our new sketching primitive \sketchname.

To our knowledge, \sysname is the first graph sketching system to deliver space savings on commonly studied real-world graphs.
Compared to the state-of-the-art lossless baseline, \sysname uses up to 21\% less space on sparse real-world graphs (average degree < 100), up to 92\% less space on intermediate density graphs (average degree $\approx$ 100--1000), and up to 97\% less space on synthetic dense graphs (average degree > 1000).
%
Update and query throughput remain competitive with both pure baselines. 



%% file: 02_preliminaries.tex
\section{Preliminaries}

\subsection{Semi-Streaming and Dynamic Graph Models}\label{sec:prelims_models}
In the \defn{graph semi-streaming model}~\cite{semistreaming1}, an algorithm is presented with a \defn{stream} of edge insertion or deletion updates that define a graph. The challenge in this model is to compute some property of the graph given a single pass over the stream and memory sublinear in the size of the graph. 
%
Specifically, stream $\graphstream$ defines a graph $\graph = (\nodes,\edges)$ with $\nodesize = |\nodes|$ and $\edgesize = |\edges|$.
Each update has the form $((u,v), \Delta)$ where $u,v \in \nodes, u \neq v$ and $\Delta \in \{-1,1\}$ where $1$ indicates an edge insertion and $-1$ indicates an edge deletion.
In this paper we \revision{first} study the \defn{semi-streaming connected components problem} which returns a spanning forest of the graph at the end of the stream.


In the \defn{dynamic graph model}, the challenge is to be able to efficiently query a specific structural property of a changing graph at any time.
%
Specifically the algorithm is given a sequence $\graphstream$ of updates ($\insertedge(e)$ or $\deleteedge(e)$) which must be processed in order.
Any prefix of the first $i$ updates in $\graphstream$ defines a graph $\graph_i = (\nodes_i,\edges_i)$.
Immediately after processing the $i$-th update , one or more queries may be issued and the algorithm must compute the answers for graph $\graph_i$ before receiving the  $i+1$-th update.
%
In this paper we study the \defn{dynamic connectivity problem}, where the structural property of interest is connectivity, and queries of the form $\connected(u, v)$ return whether vertices $u$ and $v$ are connected in $\graph_i$.


\subsection{Linear Sketching for Streaming Connectivity}\label{sec:streaming_prelims}

Ahn et. al.~\cite{Ahn2012} introduced the first semi-streaming algorithm for connected components using $\bigoh(\nodesize \log^2 \nodesize)$ space.
\revision{Like the majority of sketching algorithms (including those introduced in this paper), their algorithm guarantees that queries are correct \defn{with high probability (w.h.p.)}, that is, with probability $\geq 1-1/\nodesize^c$ for some constant $c \geq 1$.}

\revision{In the general case, this bound cannot be improved upon, as Nelson et. al. ~\cite{nelson2019optimal} shows an $\Omega(\nodesize \log^2 \nodesize)$ lower bound for computing connectivity with constant or better probability.}
The core innovation of Ahn et. al.~\cite{Ahn2012} approach is that the algorithm represents each vertex's adjacency list as a compact \defn{linear sketch} of its characteristic vector.
For a vertex $v$, its \defn{characteristic vector} $\charvec{v} \in \mathbb{R}^{\binom{\nodesize}{2}}$ encodes its incident edges such that for any edge $(u, w)$, the entry $\charvec{v}((u,w))$ is $1$ if $v = u$ and $(u,w) \in \edges$, $-1$ if $v = w$ and $(u,w) \in \edges$, and $0$ otherwise. This specific encoding ensures that for any subset of vertices $S \subseteq \nodes$, the sum of their vectors $\sum_{v \in S} \charvec{v}$ isolates the cut $(S, \nodes \setminus S)$. Every internal edge within $S$ is added twice and cancels out perfectly, leaving only the edges that cross between $S$ and the rest of the graph.
Storing characteristic vectors explicitly would require $\bigoh(\nodesize^3)$ total space. To solve this, the streaming algorithm compresses them using an \defn{$\ell_0$-sampler}.

\begin{definition}
\label{def:l0_sampler}
An algorithm $\sketch$ is an \defn{$\ell_0$-sampler} with failure probability $\delta \in (0, 1)$ if it can process coordinate updates to a vector $\recoveryVec$ and maintain a compact sketch $\sketch(\recoveryVec)$ satisfying:
\begin{enumerate}[itemsep=0pt,leftmargin=*,topsep=0pt]
    \item \textbf{Sampleable}: If $\recoveryVec$ is non-zero, sampling the sketch returns a uniform random non-zero coordinate with probability at least $1 - \delta$.
    \item \textbf{Linear}: For any vectors $\recoveryVec$ and $y$, $\sketch(\recoveryVec) + \lambda \sketch(y) = \sketch(\recoveryVec + \lambda y)$.
\end{enumerate}
\end{definition}

With constant $\delta$, the $\ell_0$-sampler from Ahn~\etal~\cite{Ahn2012} has space, update, and sampling complexity $O(\log \nodesize)$. In Section~\ref{sec:balloonsketch} we provide more details about a more recent $\ell_0$-sampler, \cameosketch~\cite{landscape}, as well as our new $\ell_0$-sampler, \sketchname.

Because $\ell_0$-samplers are linear, the streaming algorithm can compute an $\ell_0$-sampler for an entire component's cut simply by summing the sketches of its individual vertices: $\sum_{v \in S} \sketch(\charvec{v})$.
%
By maintaining $\Theta(\log \nodesize)$ independent sketches per vertex, the algorithm can sum the sketches of any merged component (supernode) to successfully sample an outgoing edge. Repeating this process simulates the $\Theta(\log \nodesize)$ rounds of \Boruvka's algorithm required to recover a spanning forest.

Tench et. al. provided the first practical implementation of the AGM sketch in \graphzep~\cite{graphzeppelin} by introducing \cubesketch, a specialized $\ell_0$-sampler that operates over $\Field_2$ rather than $\mathbb{Z}$.
The current state-of-the-art $\ell_0$-sampler, \cameosketch~\cite{landscape} (see Section~\ref{sec:balloonsketch}), builds on this foundation by further improving \cubesketch's worst-case running time and sampling success probability.

\myparagraph{\revision{\cameosketch Details and Definitions}}
\revisionmoved{
A \defn{\cameosketch} $\sketchColumn{h}$ is a sequence of $\cutoff = \Theta(\log \universesize)$ \defn{sketch buckets} $\buckset=[\buck_0, \dots, \buck_{\cutoff-1}]$. Each coordinate $v \in [\universesize]$ is assigned to exactly two of these buckets: it is always placed in $\buck_0$, designated as the \defn{deterministic bucket}, and it is also placed in a second \defn{random bucket} $\buck_i$ chosen geometrically at random. This second assignment is driven by a $\Theta(1)$-wise independent hash function $h: [\universesize] \rightarrow [2^\cutoff]$ (drawn from a family $\hashFamily$); specifically, $v$ is assigned to bucket $i = \max\{d \in \mathbb{N} \mid h(v) \equiv 0 \text{ (mod $2^d$)}\}$.
This explicitly simulates a \revision{repeated} coin-flipping process, ensuring the probability of landing in $\buck_i$ is $\prob{v \in \buck_i} = 1/2^i$.
}

\revisionmoved{
Each bucket $\buck_i$ maintains a lossy representation of the coordinates assigned to it. It stores two integer values, $\buck_i.\texttt{alpha}$ and $\buck_i.\texttt{gamma}$, both taking $\Theta(\log \universesize)$ bits and initialized to $0$.
When a coordinate $j$ of the input vector $\sketchvec$ is updated, we compute $h(j)$ to find its assigned bucket $\buck_i$, and perform the following bitwise XOR operations:
\begin{itemize}[itemsep=0pt,leftmargin=15pt]
    \item $\buck_i.\texttt{alpha} \gets \buck_i.\texttt{alpha} \oplus j$
    \item $\buck_i.\texttt{gamma} \gets \buck_i.\texttt{gamma} \oplus \buckhash(j)$ (where $\buckhash$ is a $2$-wise independent hash function)
\end{itemize}
Maintaining the buckets this way grants us these two properties:
\begin{enumerate}[itemsep=0pt,leftmargin=15pt]
    \item w.h.p. $\buck_i.\texttt{alpha} = 0$ and $\buck_i.\texttt{gamma} = 0$ if and only if the bucket contains no non-zero coordinates.
    \item w.h.p. $\buck_i.\texttt{gamma} = \buckhash(\buck_i.\texttt{alpha})$ if and only if the bucket contains exactly $1$ non-zero coordinate. In this case, $\buck_i.\texttt{alpha}$ contains the exact index of that coordinate.
\end{enumerate}
%
\revision{We also define the following additional properties of a \cameosketch for vector $\sketchvec \in \Field_2^\universesize$ with support $\supportsize$ (these definitions will later be useful in our analysis of our newly proposed $\ell_0$ sampler, \sketchname):}
\begin{itemize}[itemsep=0pt,leftmargin=15pt]
    \item \textbf{Bucket States:} A bucket $\buck_i$ is \defn{empty} if its support size is $0$. It is \defn{good} if its support size is exactly $1$, and \defn{bad} otherwise.
    \item \textbf{Bucket Depth:} We refer to the index of a bucket as its depth.
    \item \textbf{Sketch Depth:} The smallest index $k \in [\cutoff]$ such that $\buck_i$ is empty for all $i \ge k$, or $\cutoff$ if no such $k$ exists.
    \item \textbf{Residual Depth:} Let $w := \lceil \log_2 (\supportsize) \rceil$ be the index of the bucket where we mathematically expect exactly one non-zero coordinate to land.
    The \defn{residual depth} of a sketch is its actual depth minus this expected depth $w$ (or $0$ if the actual depth is $< w$).
    \item \textbf{Sketch Matrix}: A sequence of $\numSketchColumns$ sketches $(\sketchColumn{h_1}, \dots, \sketchColumn{h_{\numcolumns}})$ \revision{with hash functions with independent randomness}. By utilizing multiple independent sketches, the matrix boosts the success probability. With $L = \Omega(\log \universesize)$, the matrix yields $\Theta(\numcolumns)$ successful samples w.h.p. in $\universesize$.
\end{itemize}
}

\myparagraph{\revision{Invertible Bloom Lookup Tables}}
\revision{An \defn{invertible bloom lookup table (\iblt)}~\cite{goodrich2011iblt} is a space-efficient data structure representing a set of elements under insertions and deletions. Given a target recovery capacity $r$, an \iblt uses $O(r)$ space regardless of how many elements are inserted into it. If the current number of elements in the set is at most $r$, the \iblt guarantees that all stored elements can be decoded losslessly in $O(r)$ time with high probability. In this work, we utilize a tunable variant of \iblt~\cite{belazzougui2024iblt} configured to succeed with high probability in $\nodesize$, yielding $O(\log^2 \nodesize)$ update time.}

\subsection{Dynamic Connectivity} \label{sec:sketch_based_dyn_conn}
\myparagraph{\revision{Dynamic Trees}}
\revision{
The \defn{dynamic trees problem}~\cite{sleator1983data} is to maintain a data structure representing a forest while supporting efficient edge insertions ($\link$), edge deletions ($\cut$), and queries.
All dynamic tree data structures can support vertex connectivity queries, and many also support more advanced queries.
\defn{Path queries} and \defn{subtree queries} compute a commutative and associative function applied over the weights along a path or within a subtree, respectively.
Our new algorithms use UFO trees~\cite{deman2026ufo} in their theoretical analysis.
In \appref{Appendix~\ref{app:balloon_dc}}, we provide details about UFO trees related to our new analysis.
}

\myparagraph{Lossless Dynamic Connectivity}
Traditional dynamic connectivity algorithms all store a data structure containing a lossless representation of the graph. We refer to these as lossless dynamic connectivity algorithms.
The best algorithms, such as \clusterforest~\cite{wulff2013faster}, have $O(\nodesize + \edgesize)$ space usage, $O(\log^2 \nodesize)$ or better (amortized) update time, and $O(\log \nodesize)$ or better query time.
De Man~\etal~\cite{deman2024towards} provided a practical implementation of the \clusterforest algorithm.

\myparagraph{Sketch-Based Dynamic Connectivity}
More recently, a new class of dynamic connectivity algorithms has emerged which combines both the polylogarithmic update and query times of traditional \losslessname dynamic connectivity algorithms and the small space usage of \streamingname algorithms~\cite{kapron2013dynamic, gibb2015dynamic}. We call these \sketchdcname dynamic connectivity algorithms.
The best known sketch-based algorithm, \gibb~\cite{gibb2015dynamic}, uses $O(\nodesize \log^2 \nodesize)$ space, processes updates in $O(\log^4 \nodesize)$ time, and answers queries in $O(\log \nodesize / \log \log \nodesize)$ time. Each query is correct w.h.p.

The core component of \gibb is the \defn{cutset}~\cite{kapron2013dynamic}. For a graph $\graph = (\nodes, \edges)$ and a spanning forest $\forest = (\nodes, \edges_{\forest})$, it can find an edge crossing a component's cut during dynamic updates. Letting $\component(w)$ be the forest component containing $w$, a cutset supports:

\begin{itemize}[itemsep=0pt,leftmargin=15pt]
    \item $\link(u,v)$: Insert edge $(u,v) \notin \edges_{\forest}$ into $\forest$, merging $\component(u)$/$\component(v)$.
    \item $\cut(u,v)$: Remove edge $(u,v) \in \edges_{\forest}$ from $\forest$, splitting $\component(u,v)$.
    \item $\update(e)$: Insert edge $e$ into $\edges$ (if absent) or delete it (if present).
    \item $\query(v)$: Return an edge crossing the cut of $\component(v)$ with constant success probability $p$.
\end{itemize}

Cutsets are implemented via a dynamic tree augmented with an $\ell_0$-sampler $\sketch(\charvec{v})$ for each vertex $v$. The dynamic tree must support subtree aggregate queries using a commutative, associative function (in this case sketch addition).
Due to $\ell_0$-sampler linearity, querying a component's subtree returns the sum of its vertices' sketches, which encodes the cut between the component and the rest of the graph \revision{Since $\ell_0$-samplers return a non-zero entry of a set with some constant probability $p$, the aggregate sketch gives an edge crossing the cut with that same probability}.
Cutsets use $O(\nodesize \log \nodesize)$ space, perform $\link$s, $\cut$s, and $\update$s in $O(\log^2 \nodesize)$ time, and $\query$ in $O(\log \nodesize)$ time.

\gibb maintains cutsets $\cutset_0, \dots, \cutset_{top}$ (where $top = \Theta(\log \nodesize)$) across successive forest levels $\forest_0 \subseteq \dots \subseteq \forest_{top}$. Higher levels track increasingly larger components, ultimately yielding the true connected components. It also maintains a separate dynamic tree $\tree$ supporting maximum-weight path queries, where tree edges match $\edges_{\forest_{top}}$ and an edge's weight is the lowest level $i$ that the edge appears in $\forest_i$.
The algorithm ensures three invariants:
\begin{enumerate}[itemsep=0pt,leftmargin=15pt]
    \item \label{inv:gibb1} $\forest_0 = (\nodes, \emptyset)$.
    \item \label{inv:gibb2} $\forest_i \subseteq \forest_{i+1}$ for $0 \leq i < top$.
    \item \label{inv:gibb3} For $0 \leq i < top$, if a $\query$ on $\cutset_i$ for a component $\component \in \forest_i$ succeeds, $\component$ is a proper subset of some $\component' \in \forest_{i+1}$.
\end{enumerate}
Gibb~\etal show that if these invariants hold then $\forest_{top}$ is a spanning forest of $\graph$ w.h.p.
They also provide an update algorithm which maintains these invariants w.h.p., and takes $O(\log^4 \nodesize)$ time.
Consequently, queries are correct w.h.p. across polynomially many updates.

%% file: 03_hybrid_streaming.tex
\section{Hybrid Streaming Connectivity}\label{sec:hybrid_streaming}

As a first step towards our result of a hybrid dynamic connectivity algorithm, we present a new hybrid semi-streaming algorithm for connected components.
Theorem~\ref{thm:hybrid_streaming_connectivity} summarizes our result.

\hybridstreaming*

\myparagraph{\revision{Algorithm Structure}}
The main idea behind this result is to losslessly store the incidence lists of low-degree vertices as an explicit set, while storing high-degree vertices in sketch form.
\revision{We say a vertex whose incidence list is stored as an explicit set (e.g. using a standard balanced binary search tree) is in \defn{explicit form}, and a vertex stored as a sketch is in \defn{sketch form}.}
We set a fixed \defn{density threshold} $\densethresh$.
\revision{Initially, each vertex starts empty in explicit form. If its degree increases to $> \densethresh$, it is \defn{promoted} to sketch form. Conversely, if edge deletions cause its degree to fall back to $\le \densethresh/2$, it is \defn{demoted} to explicit form. 
As we will see, this gap between the promotion and demotion thresholds allows for more efficient amortized update complexity.}

\myparagraph{\revision{Sketching Primitives}}
\revision{To make this algorithm work, we require two sketching primitives for vertices stored in sketch form.}

First, \revision{when a vertex demotes from sketch form back to explicit form, its edges must be recovered losslessly. We must therefore be able to \defn{recover} all $\le \densethresh/2$ edges from the sketch, but} existing $\ell_0$ sketching algorithms only guarantee recovery of $\Theta(\log\nodesize)$ edges with high probability.
To address this problem, we maintain an \defn{invertible bloom lookup table (\iblt)} \revision{(see Section~\ref{sec:streaming_prelims})} over the set of neighbors for each \revision{vertex in sketch form}, \revisionmoved{configured with a recovery capacity threshold of $r = \densethresh/2$}.
\revision{Because an \iblt always uses space proportional to its capacity $r$ rather than the size of the represented set, this requires only $O(\densethresh)$ space per vertex while ensuring all $\le \densethresh/2$ edges can be decoded losslessly in $O(\densethresh)$ time upon demotion w.h.p.~\cite{belazzougui2024iblt}.}

Second, we present (in Section~\ref{sec:balloonsketch}) a more compact $\ell_0$-sampling primitive which we call \defn{\sketchname}. \revisionmoved{Our hybrid streaming algorithm maintains the adjacency lists of vertices in sketch form as a \sketchname matrix with $\numcolumns = \Theta(\log \nodesize)$ (see Section~\ref{sec:streaming_prelims} for sketch-related definitions).}
When the vertex being sketched has high degree, \sketchname matches the $O(\log \nodesize)$ space of existing $\ell_0$-samplers with constant success probability. But unlike existing $\ell_0$-sampling algorithms, the size of \sketchname is $O(\log \degree)$ in expectation where $\degree$ is the degree.
The key insight is that when the degree is low, in expectation much of the traditional $\ell_0$-sampler data structure is empty. \sketchname exploits this fact to minimize space via aggressive dynamic resizing.
\revision{We prove that a \sketchname matrix with $\numcolumns = \Theta(\log \nodesize)$ uses $\bigoh(\log \nodesize \log \degree)$ space w.h.p.}

\myparagraph{\revision{Processing Updates}}
\revision{To ensure the correctness of our connected components query algorithm (described in \appref{Appendix~\ref{app:hybrid_streaming}}), we maintain the following invariant: any edge with at least one endpoint in explicit form is stored uniquely in the explicit set of exactly one of its explicit-form endpoints, and is omitted from any sketch. In other words, only edges between two sketch-form vertices are stored symmetrically by both endpoints.}

\revisionmoved{Under this invariant, edge updates are processed as follows. If we insert or delete an edge $(u,v)$ and both endpoints are in sketch form, we update the \sketchname matrix and \iblt for both endpoints. If either endpoint is in explicit form, the edge should only be stored in the adjacency list of one of the endpoints that is in explicit form. For insertions, we arbitrarily pick one of these endpoints and insert the edge into its set. For deletions, we check both endpoints and delete the edge if the set contains it.}

\revisionmoved{When an edge insertion causes a vertex $u$ in explicit form to exceed the density threshold ($\degree(u) > \densethresh$), $u$ is promoted to sketch form. At this point, $u$ converts its set into sketch form by initiating a new \sketchname and an \iblt. To maintain the invariant, $u$ iterates through the edges currently in its set and redistributes them as follows.}
\revisionmoved{For neighbors in sketch form, these edges now exist between two vertices in sketch form. Vertex $u$ adds the edge to its new sketch and issues an update to the neighbor $v$ instructing it to also add the edge to its sketch.}
\revisionmoved{For neighbors in explicit form, these edges now have exactly one endpoint in explicit form (the neighbor). Since the vertex $u$ is now in sketch form, it can no longer store these edges. It therefore pushes these edges to the respective neighbors' sets and does not add them to its own sketch.}

\revisionmoved{Conversely, if edge deletions cause a sketch-form vertex $u$ to drop to $\degree(u) \leq \densethresh/2$, it demotes to explicit form. Because the invariant ensures $u$'s \iblt only tracks edges where both endpoints were in sketch form, $u$ uses its \iblt to perfectly recover these edges. It moves these recovered edges into its new explicit set and issues updates to those neighbors (which are still in sketch form) to delete the edges from their sketches. Crucially, any edges $u$ shares with other explicit form vertices were already uniquely stored in the neighbors' sets.}

\myparagraph{\revision{Analysis}}
\revisionmoved{For a single vertex of degree $\degree$, an explicit set representation requires $\Theta(\degree)$ space, whereas the \sketchname matrix representation consumes $\bigoh(\log \nodesize \log \degree)$ space w.h.p.} 
\revision{The optimal density threshold $\densethresh$ occurs where these two space complexities are equal. Solving $\degree = \Theta(\log \nodesize \log \degree)$ yields $\densethresh = \Theta(\log \nodesize \log \log \nodesize)$. Importantly, the space usage of an \iblt for a vertex in sketch form is $\bigoh(\densethresh) = \bigoh(\log \nodesize \log \log \nodesize)$, which is always asymptotically dominated by our bound on the size of the \sketchname matrix for vertices in sketch form, which is $\bigoh(\log \nodesize \log (\log \nodesize \log \log \nodesize))$ w.h.p.}

\revisionmoved{With $\densethresh = \Theta(\log \nodesize \log \log \nodesize)$, we prove in \appref{Appendix~\ref{app:hybrid_streaming}} that the total space complexity across all vertices satisfies both $O(\nodesize + \edgesize)$ and $O(\nodesize \log \nodesize \log (2+\edgesize/\nodesize))$ w.h.p., establishing the space bounds for Theorem~\ref{thm:hybrid_streaming_connectivity}.} 
\revision{Also in \appref{Appendix~\ref{app:hybrid_streaming}}, we prove that updates have amortized cost $O(\log \nodesize)$ w.h.p., and we describe and analyze the connected components query algorithm taking worst-case $O(\nodesize \log^2 \nodesize)$ time.}

\input{03.1_balloon_sketch}

%% file: 03.1_balloon_sketch.tex
\subsection{\sketchname}
\label{sec:balloonsketch}

This section introduces \sketchname, a new sketching primitive for $\ell_0$-sampling that uses less space on sparse vectors than prior methods.
\revision{We refer the reader to Section~\ref{sec:streaming_prelims} for details and definitions related to the state-of-the-art $\ell_0$-sampler, \cameosketch. Many of these definitions are used in our description of \sketchname.}
The key benefit of \sketchname is that it has an asymptotic space advantage over \cameosketch and all existing $\ell_0$-samplers on sparse vectors. On dense vectors, it matches the asymptotic space complexity of \cameosketch.
In a sketch matrix with large enough size, updates to \sketchname are as efficient as updates to \cameosketch.
The sampling success probability and the cost of sampling for a \sketchname are identical to \cameosketch.
Definition~\ref{defn:skinny_sketch} defines a \sketchname:

\begin{definition}[\sketchname] \label{defn:skinny_sketch}
    A \defn{\sketchname} is a sequence of sketch buckets stored in a dynamically allocated container. The size of the container is always resized to the sketch depth.
\end{definition}

Unlike \cameosketch which statically allocates $\Theta(\log \universesize)$ buckets, \sketchname avoids storing the empty buckets deeper than its sketch depth.
Its space usage is therefore proportional to its depth.
The procedures for assigning coordinates to buckets and updating the \sketchname are almost identical to those of \cameosketch. 
However, if an update changes the depth of the sketch, then the bucket list must shrink or grow.
The pseudo-codes for the \sketchname algorithms are given in Algorithm~\ref{alg:skinny_sketch_column}.

\begin{algorithm}
    \caption{\sketchname Algorithms}\label{alg:skinny_sketch_column}
    \begin{algorithmic}[1]
    
    \Function{initialize}{$\universesize$}
        \State $\buckset \gets []$, $\cutoff \gets \lg(\universesize) + \Theta(1)$
        \State Let $h: \positivesUpTo{\universesize} \rightarrow \integersUntil{2^\cutoff}$ be drawn from a family of $\Theta(1)$-wise independent hash functions.
    \EndFunction
    \Statex \vspace{-0.5em}
    
    \Function{sample}{}
        \If {$\buckset == []$} \Return ``empty'' \EndIf
        \For{$i= \buckset$.size to $0$}
            \If{$\buckset[i]$ is good}
                \Return (``good'', $\buckset[i]$.value)
            \EndIf
        \EndFor
        \State \Return ``fail''
    \EndFunction
    \Statex \vspace{-0.5em}
    
    \Function{reallocate}{$\ell$}
        \State Initialize new bucket list $\buckset^\prime$ with $\ell$ empty buckets.
        \State Copy all non-empty buckets from $\buckset$ to $\buckset^\prime$.
        \State $\buckset \gets \buckset^\prime$
    \EndFunction
    \Statex \vspace{-0.5em}
    
    \Function{update}{$u$}
        \State $i \gets$ index of random bucket for $u$ (found from $h(u)$)
        \If {$i \geq \buckset.size()$} \Call{reallocate}{$i + 1$} \EndIf
        \State Update bucket $\buckset[i]$ with $u$.
        \State Update bucket $\buckset[0]$ with $u$.
        \If {$i == B.size()-1$ and $\buckset[i]$ is empty}
            \State $d \gets$ index of deepest non-empty bucket in $\buckset$.
            \If {$d < \buckset.size()-1$}
                \State \Call{reallocate}{$d + 1$}
            \EndIf
        \EndIf
    \EndFunction
\end{algorithmic}
\end{algorithm}

The update function first computes the random bucket index for $u$, and if necessary reallocates the bucket list so that there is a bucket allocated at that index.
Then both the random bucket and the deterministic bucket are updated.
Finally, if the random bucket was the last bucket in the list, and it became empty from this update, the bucket list is reallocated to remove all empty buckets at the tail.
This only has to occur if the random bucket was the last bucket, because otherwise there must be a later bucket that is still non-empty.

Because of the possible need to reallocate the bucket list during an update and the cost of possibly searching for the deepest non-empty bucket, the time for an update to coordinate $u$ can be proportional to the depth of the random bucket for $u$.
Since this bucket is chosen randomly between index $0$ and $\cutoff-1$, the worst-case update time for a \revision{single} \sketchname is $\Theta(\log(\universesize))$.
This is notably worse than \cameosketch's worst-case update time of $\Theta(1)$.
For similar reasons, the only worst-case space guarantee for \sketchname is $\Theta(\log (\universesize))$ which matches \cameosketch but is no better for sparse vectors.
However, there are three key observations that \revision{enable} asymptotic space savings and update efficiency for a \revision{sequence of columns}:
\begin{enumerate}[itemsep=0pt,leftmargin=*]
    \item The index of the random bucket a coordinate is assigned to is a geometric random variable, and thus has $\Theta(1)$ expected value.
    \item The residual depth of a sketch is bounded by a geometric random variable and thus has $\Theta(1)$ expected value.
    \item In a sketch matrix, these random variables for each sketch are independent, and sums of independent random variables concentrate around the sum of their expectations.    
\end{enumerate}
%
These observations guide our proof of Theorem~\ref{thm:balloon_sketch} which provides high probability bounds on the space usage and update time of a \sketchname matrix with $L$ sketches.
When $L=\Theta(\log\universesize)$ (as is the case for applications like streaming connectivity), The update time for a \sketchname matrix is $O(\log\universesize)$ w.h.p., whereas the update time for a \cameosketch matrix is $O(\log\universesize)$ worst-case.
In this case the \sketchname uses $O(\log \universesize \log \supportsize)$ space w.h.p. (where $\supportsize$ is the number of non-zero elements in the vector), whereas \cameosketch uses $\Theta(\log^2\universesize)$ worst-case space.
This indicates that on sparse vectors (low $\supportsize$) \sketchname can use asymptotically less space than \cameosketch.
The full proof is provided in \appref{Appendix~\ref{appendix:balloon_proofs}}.

\balloonsketchthm*


%% file: 04_hybrid_dynamic.tex
\section{Hybrid Fully-Dynamic Connectivity}\label{sec:hybrid_dynamic}

\revision{
In this section, we present our hybrid fully dynamic graph connectivity algorithm, proving Theorem~\ref{thm:hybrid_dynamic_connectivity}.
}

\hybriddynamic*

\input{04.1_hybrid_dc}

\input{04.2_balloon_dc}

%% file: 04.1_hybrid_dc.tex
\subsection{Hybrid Dynamic Connectivity Framework} \label{sec:hybrid_framework}

\revision{
Unlike in the streaming setting, it is not as straightforward to develop a hybrid algorithm for dynamic connectivity, because the data structures used in dynamic connectivity algorithms aggregate information over several subsets of vertices. Thus we can not simply use a hybrid approach on the granularity of each vertex.
}

\revisionmoved{
Instead, we introduce a novel dynamic connectivity framework which uses different data structures to store the {\bf sparse regions} and {\bf dense regions} of the input graph.
Employing a hybrid approach over {\em regions} of the graph allows us to overcome the difficulty of doing it on the granularity of vertices.
Our framework is general, and can use any two dynamic connectivity algorithms (assuming they support a certain reasonable set of operations) to store the sparse and dense regions of the graph.
Our framework is intended to use a {\bf lossless algorithm} to store the sparse regions, and a {\bf sketch-based algorithm} to store the dense regions.
}

\revision{
First, we describe the framework abstractly using two dynamic connectivity algorithms as black boxes.
Later, we will prove Theorem~\ref{thm:hybrid_dynamic_connectivity} by analyzing the theoretical guarantees of our framework using two specific dynamic connectivity algorithms as the subroutines.
We will use a new \sketchname-based dynamic connectivity algorithm (described in Section~\ref{sec:skinny_dynamic}) as the dense algorithm, and the \clusterforest algorithm~\cite{wulff2013faster, deman2024towards}, as the sparse algorithm.
}

\begin{figure*}[ht]
    \centering
    \includegraphics[width=\linewidth]{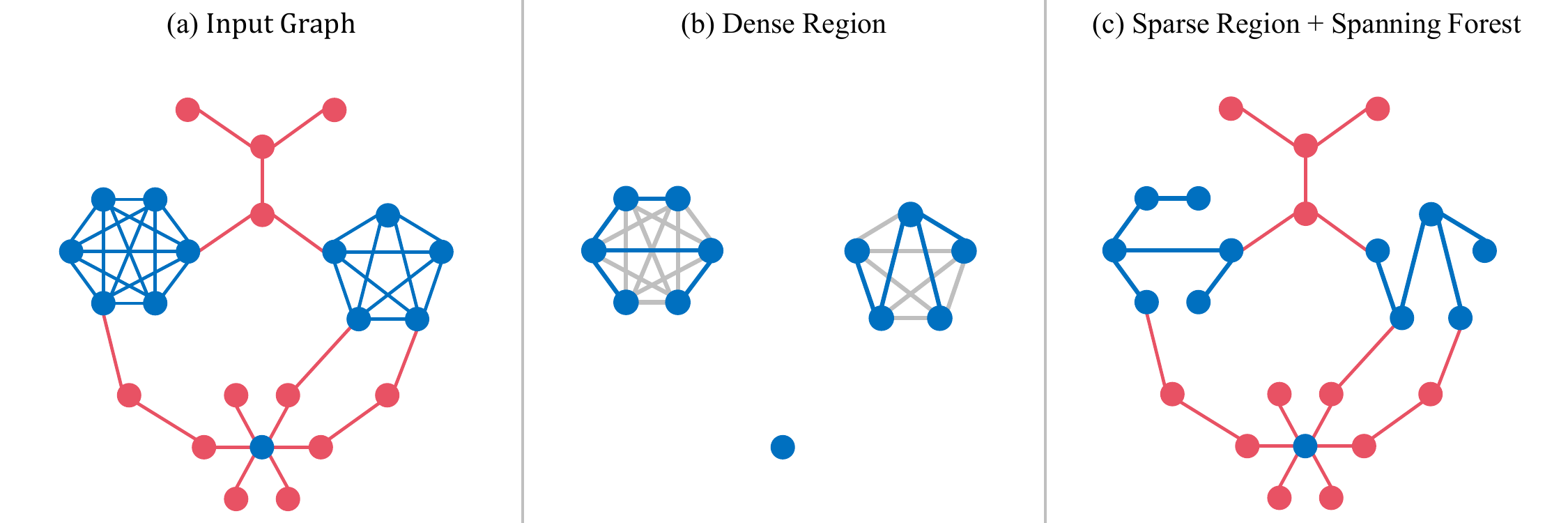}
    \vspace{-1em}
    \caption{
    An example of the partitioning of a graph by our hybrid framework with density threshold $\densethresh = 3$.
    Subfigure (a) shows a graph $\graph$. The blue vertices and edges in (a) represent heavy vertices and dense edges, and the red vertices and edges represent light vertices and sparse edges.
    Subfigure (b) shows the dense region of $\graph$ which is stored by $\densealg$, indicating the spanning forest with blue edges.
    Subfigure (c) shows the union of the sparse region of $\graph$ and the spanning forest from $\densealg$, which is stored by $\sparsealg$.
    }
    \label{fig:hybrid_dynamic_example}
\end{figure*}

\revision{Now we describe the key definitions in our framework.}
As before, our framework is parametrized by some fixed density threshold $\densethresh$.
\begin{itemize}[leftmargin=15pt,itemsep=0pt]
    \item A \defn{heavy vertex} is defined as a vertex $v$ such that $\degree(v) \geq \densethresh$.
    \item A \defn{light vertex} is defined as a vertex $v$ such that $\degree(v) < \densethresh$.
    \item A \defn{dense edge} is defined as an edge $e=(u,v)$ such that $u$ and $v$ are both heavy vertices.
    \item A \defn{sparse edge} is defined as any edge that is not a dense edge.
    \item The \defn{dense region} of $\graph$ is the subgraph induced by the heavy vertices. That is all the heavy vertices, and all the dense edges (between two heavy vertices).
    \item The \defn{sparse region} of $\graph$ contains all vertices in $\nodes$ and all the sparse edges.
\end{itemize}
Note that by these definitions, the edges of the graph are partitioned between the dense region and sparse region, while there may be overlap in the sets of vertices in each region (the sparse region always has all vertices and the dense region has a subset).
Our framework uses two separate dynamic connectivity algorithms to store the dense and sparse regions of the graph:

\begin{itemize}[leftmargin=15pt,itemsep=0pt]
    \item \defn{Dense Algorithm ($\densealg$)}: stores only the dense region of $\graph$.
    \item \defn{Sparse Algorithm ($\sparsealg$)}: stores the union of the sparse region of $\graph$ and a spanning forest of the dense region.
\end{itemize}
Specifically, we require that $\densealg$ explicitly stores a spanning forest, and that $\sparsealg$ contains the same spanning forest of the dense region that is maintained by $\densealg$.
This means that all edges are partitioned between $\densealg$ and $\sparsealg$ except for these \defn{dense spanning forest edges} which are stored by both algorithms.
Figure~\ref{fig:hybrid_dynamic_example} shows an example of a graph $\graph$, its dense and sparse regions, and the subgraphs maintained by $\densealg$ and $\sparsealg$.

Our framework only requires that (in addition to supporting all the operations required by the dynamic connectivity problem) $\densealg$ and $\sparsealg$ have the following properties, which are standard or easily supported by most existing dynamic connectivity algorithms:

\smallskip
\noindent \underline{Requirements for $\densealg$}:
\begin{itemize}[leftmargin=15pt,itemsep=0pt]
    \item Explicitly maintains a \defn{spanning forest $\densealg.\forest$} over its graph. Note that some dynamic connectivity algorithms can answer connectivity queries, but do not maintain a spanning forest.
    \item $\insertedge$ and $\deleteedge$ operations \defn{return a list of the changes to $\densealg.\forest$} induced by the update.
    \item \defn{$\insertvertex(v) / \deletevertex(v)$}: Given a vertex $v$ with $\degree(v)=0$, inserts or deletes $v$ in the graph. This is necessary since the vertex set in the dense region of the graph may change.
\end{itemize}
\noindent \underline{Requirements for $\sparsealg$}:
\begin{itemize}[leftmargin=15pt,itemsep=0pt]
    \item \defn{$\incidentedges(v)$}: Return all edges incident to vertex $v$. This is not supported by existing sketch-based methods, but is simple for most lossless methods.
    \item \defn{$e \in \sparsealg.\edges$}: Returns whether edge $e$ is in the graph maintained by $\sparsealg$. This is also not possible for existing sketch-based methods, but easy for lossless methods.
\end{itemize}
Most notably, $\densealg$ require each update to return a list of changes to its spanning forest. 
We require this because different dynamic connectivity algorithms can vary greatly in how they update their spanning forest in response to edge insertions or deletions.
For example, any insertions or deletions in sketch-based algorithms like \gibb can induce up to $\Theta(\log \nodesize)$ spanning forest changes.


In addition to $\densealg$ and $\sparsealg$, our framework stores a few extra data structures to interface between the two.
To handle vertices that change from heavy to light, our framework maintains an \iblt with recovery size $\densethresh/2$, for each vertex in $\densealg$.
We also store an array of $\nodesize$ boolean values indicating whether each vertex is in $\densealg$.
The last piece of additional data stored by our framework is the \defn{global degree} $\degree(v)$ for each vertex $v$ (its degree in the input graph $\graph$). This is necessary because the edges of $\graph$ may be divided between $\densealg$ and $\sparsealg$ and some edges may be in both, so it is hard to discern the total degree using just these data structures.


\myparagraph{Query Algorithm}
Since $\sparsealg$ stores the union of the sparse region of the graph and a spanning forest of the dense region, all vertices in $\graph$ are in $\sparsealg$, and the connected components in $\sparsealg$ are exactly the same as the connected components of $\graph$.
Thus, a query $\connected(u,v)$ in our framework can be answered at any time by calling $\sparsealg.\connected(u,v)$.

\begin{algorithm}[ht]
\caption{Edge insertion algorithms for the hybrid framework.}
\label{alg:hybrid_insert}
\begin{algorithmic}[1]
    \Function{\insertedge}{$e=\{u, v\}$}
        \For{$x \in \{u,v\}$} \label{line:promote_start}
        \State $\degree(x) \gets \degree(x)+1$
        \If{$x \notin \densealg.\nodes$ and $\degree(x) > \densethresh$}
            \State $\promotevertex(x)$
        \EndIf
        \EndFor \label{line:promote_end}
        \If{$u \in \densealg.\nodes$ and $v \in \densealg.\nodes$} \label{line:ins_dense_start}
            \State $\densealg.\insertedge(e)$
            \State $\iblt_u.\myinsert(v)$, $\iblt_v.\myinsert(u)$
            \State Update $\sparsealg$ with all changes to $\densealg.\forest$. \label{line:ins_dense_end}
        \Else
            \State $\sparsealg.\insertedge(e)$ \label{line:ins_sparse}
        \EndIf
    \EndFunction
    \Statex

    \Function{\promotevertex}{$v$}
        \State $\densealg.\insertvertex(v)$ \label{line:promote_ins_vertex}
        \State Instantiate $\iblt_v$ \label{line:promote_create_recovery}
        \For{$e = \{v,u\} \in \sparsealg.\incidentedges(v)$} \label{line:promote_ins_edges_start}
            \If{$u \in \densealg.\nodes$}
                \State $\densealg.\insertedge(e)$
                \State $\iblt_u.\myinsert(v)$, $\iblt_v.\myinsert(u)$
                \State $\sparsealg.\deleteedge(e)$
            \EndIf
        \EndFor \label{line:promote_ins_edges_end}
        \State Update $\sparsealg$ with all changes to $\densealg.\forest$.
    \EndFunction
\end{algorithmic}
\end{algorithm}

\myparagraph{Insertion Algorithm}
Algorithm~\ref{alg:hybrid_insert} shows the pseudo-code for inserting an edge $e=\{u,v\}$ in our hybrid framework.
First, the algorithm updates the global degree of $u$ and $v$, and if their new degree is $> \densethresh$, calls \promotevertex on them (lines~\ref{line:promote_start}--\ref{line:promote_end}).
$\promotevertex(v)$ first inserts $v$ into $\densealg$ and instantiates an \iblt for $v$ (lines~\ref{line:promote_ins_vertex}--\ref{line:promote_create_recovery}).
Next, for all the edges incident to $v$ in $\sparsealg$, if the other endpoint is also in $\densealg$ they are inserted into $\densealg$ and the \iblts, and deleted from $\sparsealg$ (lines~\ref{line:promote_ins_edges_start}--\ref{line:promote_ins_edges_end}). This ensures that $\densealg$ includes all edges that became dense edges when $v$ became a heavy vertex.
Finally, we must update $\sparsealg$ to contain the spanning forest maintained by $\densealg$.
As mentioned previously, our framework requires that edge insertion operations in $\densealg$ report a list of changes to the spanning forest caused by these insertions.
Our algorithm uses this list of changes to insert or delete any necessary edges in $\sparsealg$ so that it always stores the exact spanning forest maintained by $\densealg$.

Now that the algorithm has dealt with vertices becoming heavy, the next step is to actually insert the new edge $e$.
If both endpoints of $e$ are in the dense region, we insert $e$ into $\densealg$, update the \iblts for both endpoints, and again update $\sparsealg$ with all the spanning forest changes in $\densealg$ (lines~\ref{line:ins_dense_start}--\ref{line:ins_dense_end}).
Else we simply insert $e$ into $\sparsealg$, since the edge is in the sparse region of the graph (line~\ref{line:ins_sparse}).

\begin{algorithm}[ht]
\caption{Edge deletion algorithms for the hybrid framework.}
\label{alg:hybrid_delete}
\begin{algorithmic}[1]
    \Function{\deleteedge}{$e=\{u, v\}$}
        \For{$x \in \{u,v\}$} \label{line:demote_start}
        \State $\degree(x) \gets \degree(x)-1$
        \If{$x \in \densealg.\nodes$ and $\degree(x) \leq \densethresh/2$}
            \State $\demotevertex(x)$
        \EndIf
        \EndFor \label{line:demote_end}
        \If{$u \in \densealg.\nodes$ and $v \in \densealg.\nodes$} \label{line:del_dense_start}
            \State $\densealg.\deleteedge(e)$
            \State $\iblt_u.\mydelete(v)$, $\iblt_v.\mydelete(u)$
            \State Update $\sparsealg$ with all changes to $\densealg.\forest$. \label{line:del_dense_end}
        \Else
            \State $\sparsealg.\deleteedge(e)$ \label{line:del_sparse}
        \EndIf
    \EndFunction
    \Statex

    \Function{\demotevertex}{$v$}
        \For{$e=(u,v) \in \iblt_v.\recover()$} \label{line:demote_recover}
            \State $\densealg.\deleteedge(e)$ \label{line:demote_del_start}
            \State $\iblt_u.\mydelete(v)$
            \If{$e \notin \sparsealg.\edges$}
                \State $\sparsealg.\insertedge(e)$
            \EndIf
        \EndFor \label{line:demote_del_end}
        \State Update $\sparsealg$ with all changes to $\densealg.\forest$.
        \State $\densealg.\deletevertex(v)$ \label{line:demote_del_vertex}
        \State Delete $\iblt_v$ \label{line:demote:free_recovery}
    \EndFunction
\end{algorithmic}
\end{algorithm}

\myparagraph{Deletion Algorithm}
Algorithm~\ref{alg:hybrid_delete} shows the pseudo-code for deleting an edge $e=\{u,v\}$ in our hybrid framework.
First, the algorithm updates the global degree of $u$ and $v$, and if their new degree is $\leq \densethresh/2$, calls \demotevertex on them (lines~\ref{line:demote_start}--\ref{line:demote_end}).
Similar to our hybrid streaming algorithm (Section~\ref{sec:hybrid_streaming}), we only demote vertices when their degree drops down to $\densethresh/2$ to prevent thrashing and enable efficient amortized update complexity.
$\demotevertex(v)$ uses the \iblt of $v$ to extract the at most $\densethresh/2$ edges incident to $v$ in $\densealg$ (line~\ref{line:demote_recover}).
For each of these edges, they are deleted from $\densealg$, deleted from the \iblt of the other endpoint, and inserted into $\sparsealg$ if they are not already present there (lines~\ref{line:demote_del_start}--\ref{line:demote_del_end}).
Once again, any changes caused to the spanning forest maintained by $\densealg$ by these deletions must be updated in $\sparsealg$.
In particular, this step must avoid deleting the edges just inserted to $\sparsealg$ in the previous loop.
Finally, $v$ is deleted from $\densealg$, and $\iblt_v$ is freed (lines~\ref{line:demote_del_vertex}--\ref{line:demote:free_recovery}).

Now the algorithm deletes edge $e$.
If both endpoints of $e$ are in the dense region, delete $e$ from $\densealg$, update the \iblts for both endpoints, and again update $\sparsealg$ with all the spanning forest changes in $\densealg$ (lines~\ref{line:ins_dense_start}--\ref{line:ins_dense_end}).
This last step is particularly important, because if a spanning forest edge was deleted in $\densealg$, this edge would currently be present in $\sparsealg$, and this step also deletes it there (disconnecting the graph represented by $\sparsealg$). If a replacement edge was found in $\densealg$, it becomes part of the spanning forest and is also inserted into $\sparsealg$.
If both endpoints of $e$ are not in the dense region, we simply delete $e$ from $\sparsealg$, since the edge is only in the sparse region of the graph (line~\ref{line:del_sparse}).


\myparagraph{\revision{Analysis}}
We prove Theorem~\ref{thm:hybrid_dynamic_connectivity} by analyzing the performance of our hybrid framework using a new \sketchname-based algorithm (described in Section~\ref{sec:skinny_dynamic}) as $\densealg$, and the \clusterforest algorithm~\cite{wulff2013faster} as $\sparsealg$.
\revision{The \sketchname-based algorithm has space complexity $O(\nodesize \log \nodesize \log(2+ \edgesize / \nodesize))$ w.h.p.}
It is easily verifiable that our \sketchname-based algorithm and \clusterforest meet the requirements for $\densealg$ and $\sparsealg$ respectively.

We configure the density threshold to be $\densethresh = \log \nodesize \log \log \nodesize$.
Intuitively, this is the value of $\densethresh$ where if every vertex has degree exactly $\densethresh$, then the space usage of a lossless representation is asymptotically equal to the space usage required in our \sketchname-based algorithm: $O(\nodesize+\nodesize \densethresh)$ versus $O(\nodesize \log\nodesize \log(2+\densethresh\nodesize/\nodesize))$.

We prove that the space complexity of our \revision{hybrid} algorithm is $O(\nodesize + \edgesize)$ w.h.p. and $O(\nodesize \log \nodesize \log (2+\edgesize / \nodesize))$ w.h.p in \appref{Appendix~\ref{app:hybrid_dynamic}}.
The key observation for the first bound is that the space used by $\densealg$ is $O(\edgesize)$ since the existence of $\nodesize_D$ vertices in $\densealg$ implies the number of $\edgesize$ is at least $\nodesize_D \densethresh$.
The key observation for the second bound is that $\sparsealg$ uses $O(\nodesize \log \nodesize \log(2+\edgesize/\nodesize))$ space since every light vertex has maximum degree $\densethresh$, so the number of sparse edges is $O(\nodesize\densethresh)$.
%
We also prove in \appref{Appendix~\ref{app:hybrid_dynamic}} that the update and query complexities are amortized $O(\log^4 \nodesize)$ and $O(\log \nodesize / \log \log \nodesize)$.

%% file: 04.2_balloon_dc.tex
\subsection{\sketchname-Based Dynamic Algorithm} \label{sec:skinny_dynamic}

The starting point for our new sketch-based algorithm is \gibb~\cite{gibb2015dynamic} (summarized in Section~\ref{sec:sketch_based_dyn_conn}).
Rather than using the $\ell_0$-sketch algorithm of Cormode~\etal~\cite{cormode2014unifying}, we will instead use our new $\ell_0$-sketch from Section~\ref{sec:balloonsketch}, \sketchname.
Recall that the data structures maintained by \gibb are $\Theta(\log \nodesize)$ levels of cutset data structures (see Section~\ref{sec:sketch_based_dyn_conn}).
Using the sketch of Cormode~\etal, the space used by each cutset is $\bigoh(\nodesize \log \nodesize)$ w.h.p., resulting in $\bigoh(\nodesize \log^2 \nodesize)$ space w.h.p. overall.

We prove that by using \sketchname instead of the sketch of Cormode~\etal in the cutset data structures, the aggregate space usage of each cutset data structure can be reduced to $\bigoh(\nodesize \log(\edgesize/\nodesize))$ in expectation, and the total space usage for the entire algorithm can be reduced to $\bigoh(\nodesize \log(\nodesize)\log(\edgesize/\nodesize))$ w.h.p.

There is inherent flexibility in the specific implementation of the cutset data structure, as any dynamic tree supporting subtree aggregate queries can be used.
Previous works such as \gibb and \cupcake use Euler-tour trees (ETTs) for this purpose, but the asymptotic space usage is the same regardless of the choice of dynamic tree. In this work, we \revision{analyze} the cutset data structures using UFO trees~\cite{deman2026ufo} augmented with \sketchname sketches.
\revision{Although we use more space-efficient LCTs in practice (see Section~\ref{sec:system}),} we chose to use UFO trees \revision{in the analysis} because they simplify the \revision{aggregation} of space across all \sketchnames in our data structures.

Our overall dynamic connectivity data structure and algorithms remain equivalent to \gibb.
Since UFO trees provide the same update complexity as ETTs, the overall update and query complexity of our dynamic connectivity algorithm remain the same: $O(\log^4 \nodesize)$ and $O(\log \nodesize / \log \log \nodesize)$ respectively.
Lemma~\ref{lem:balloon_dc} summarizes the results of our algorithm.

\begin{restatable}{lemma}{balloondcspace} \label{lem:balloon_dc}
    There exists a dynamic connectivity algorithm with space complexity $O(\nodesize \log \nodesize \log (2+\edgesize/\nodesize))$ w.h.p.
    The update complexity is $O(\log^4 \nodesize)$. The query complexity is $O(\log \nodesize / \log \log \nodesize)$.
\end{restatable}

\myparagraph{Space Usage}
Here we provide an overview of the space complexity analysis for our new algorithm.
We provide the full proof of Lemma~\ref{lem:balloon_dc} in \appref{Appendix~\ref{app:balloon_dc}}.
In our prior proof of the space usage of our streaming algorithm ( Theorem~\ref{thm:hybrid_streaming_connectivity}, Section~\ref{sec:hybrid_streaming}), we established that the total space usage of the \sketchname matrix for each vertex $v$ is $\Theta(\log \nodesize \log \degree_v)$ w.h.p.
%
\revision{However,} proving the space bounds for our dynamic algorithm introduces two main challenges.

First, we must account for space in internal UFO tree nodes, not just the leaves representing the vertices of the graph.
While leaf nodes asymptotically dominate the total node count, higher-level internal nodes may sketch an increasing average number of edges.
We must rigorously prove that the sum of these potentially larger internal \sketchnames does not asymptotically exceed the leaf-level space.

Second, the $\Theta(\log \nodesize)$ sketches for a vertex are distributed across independent UFO trees.
Naively bounding the space of each sketch individually yields a suboptimal $O(\nodesize \log^2 \nodesize)$ total space w.h.p.
Instead, we must leverage the independent random seeds across tiers to aggregate sketch residual depths and apply concentration bounds, achieving the tighter $O(\nodesize \log \nodesize \log (2+\edgesize/\nodesize))$ bound w.h.p.

\myparagraph{Vertex Updates}
For use in our hybrid dynamic connectivity framework, we describe how to modify the \sketchname-based algorithm to support a changing set of vertices.
Specifically we support the operations $\insertvertex(v)$ and $\deletevertex(v)$ which insert or delete a degree $0$ vertex in the graph.

Let $\nodesize_D$ be the current number of vertices in the graph.
We assume that there are at most $\nodesize$ total vertices at any time, that is $\nodesize_D \leq \nodesize$.
To maintain that our algorithm is correct w.h.p. in $\nodesize$ we still maintain $\Theta(\log \nodesize)$ tiers at all times, even if the actual number of vertices $\nodesize_D$ is actually much less than $\nodesize$.
Consequently, the space usage of the data structure is $O(\nodesize_D \log \nodesize \log (2+\edgesize/\nodesize_D))$ with high probability in $\nodesize$ (note that the first logarithmic term is in terms of $\nodesize$, not $\nodesize_D$).

To support vertex updates, we maintain a hash map that maps each vertex $v$ to the UFO tree leaf cluster representing them in each tier.
Vertex insertions allocate a new UFO cluster with a \sketchname in each tier, and vertex deletions free them all.
Both take $O(\log \nodesize)$ time, since an empty \sketchname has constant size.

%% file: 05_implementation.tex
\section{System Design \& Implementation}\label{sec:system}

In this section, we describe \defn{\sysname} (\textbf{Hybrid} \textbf{S}ketching \textbf{C}onnectivity \textbf{A}lgorithm with \textbf{L}ossless \textbf{E}dges), our hybrid sketching system for dynamic connectivity based on the algorithm from Section~\ref{sec:hybrid_dynamic}.
As per the name, our system can seamlessly \emph{scale} existing dynamic connectivity systems to larger and denser real-world graphs.
\sysname consists of three major components:
\begin{enumerate}[itemsep=0pt,leftmargin=15pt]
    \item $\sparsealg$: the \defn{Cluster Forest implementation} of De Man~\etal~\cite{deman2024towards}.
    \item $\densealg$: \defn{\densesysname}, an implementation of our new \sketchname-based dynamic connectivity algorithm from Section~\ref{sec:skinny_dynamic}.
    \item A \defn{hybrid interface} that manages communication and movement of edges and vertices between the first two components.
\end{enumerate}
We chose to use this implementation of the \clusterforest because it provides good guarantees on performance and provides robust state-of-art performance regardless of the structure of the input graph.
However, the first two components can be easily interchanged with any other existing or future dynamic connectivity systems that meet the basic requirements for $\sparsealg$ and $\densealg$. All of our implementations are written in C++ and compiled with -03 optimization.

\myparagraph{\revision{Promotion Threshold}}
A key parameter in our hybrid algorithm is the density threshold $\delta$.
\revision{In \appref{Appendix~\ref{app:results}}, we provide a thorough experimental evaluation of the memory usage of \sysname using different values of $\densethresh$. Based on these results, we chose to set} $\densethresh = 25 \cdot \lceil \log_2 \nodesize \rceil$. \revision{Although the best $\densethresh$ value differs slightly for each graph, we found that this common value provides close to optimal performance across a wide variety of graphs.}

\myparagraph{\revision{Demotion Threshold}}
In \sysname, we demote a vertex when its degree drops to $\densethresh/8$ (as opposed to $\densethresh/2$ in our theoretical analysis).
This choice manages a tradeoff between the aggressiveness of vertex reclamation and global memory overhead.
While a higher threshold would move vertices back to the space-efficient \clusterforest sooner, it would also require every vertex in the dense subsystem to maintain a larger \iblt. Provisioning  each \iblt for $\densethresh/2$ would cause these structures to occupy a significantly larger portion of the system's total memory. By lowering the threshold to $\densethresh/8$, we decrease this memory overhead, at the cost of some vertices remaining in the dense subsystem a bit longer.

\myparagraph{\revision{Error Rate}}
\revision{
\sysname can produce errors in two ways.
The first way is that a 64-bit checksum collision within a sketch bucket could produce an undetectable error. We treat this risk as negligible as relying on 64-bit hash collision resistance is standard practice across large-scale systems~\cite{fan2014cuckoo,cormode2005improved}.
}

\revision{
The second type of error is when the number of sketch tiers used was not enough to produce a full spanning forest in \densesysname.
These errors can be detected by checking whether any root sketch of a top tier component is non-empty.
The rate of these errors is determined by} the number of tiers used in \densesysname, which must be $\Theta(\log \nodesize)$ in theory \revision{to achieve a polynomially small error rate}.

We use $\max\{15, \lceil \log_2 \nodesize \rceil \}$ tiers, which is slightly more aggressive than prior work on graph-sketching methods~\cite{landscape} and below the number of tiers needed according to a standard analysis of the semi-streaming algorithm~\cite{Ahn2012}.
\revision{We chose to impose a minimum of 15 tiers, since for graphs with extremely small amounts of vertices, the weaker correctness guarantee leads to higher error rates.}

\revision{In \appref{Appendix~\ref{app:results}}, we evaluate the empirical correctness of \sysname across 500 independent trials on different graphs, checking the connected components every 100,000 updates. In these experiments, \sysname produced zero errors.
Furthermore, we found that} \sysname actually uses only 64\% of these $\lceil\log_2\nodesize\rceil$ tiers (geometric mean across all tested datasets), and consistently uses fewer tiers than \cupcake.
This is intuitively because our design choice to sketch only the dense regions of the graph makes finding connectivity via Bor\r{u}vka take significantly fewer rounds.
Additionally, while we provision enough tiers for the maximum number of vertices, the actual number of vertices present in the dense region is often much lower, naturally reducing the required number of tiers.
Together, these factors indicate that there may be an opportunity for more aggressive space savings by reducing the number of tiers stored.

\myparagraph{\revision{Vertex Updates}}
\revision{
\sysname can be easily modified to support dynamic vertex insertions and deletions.
In the \clusterforest, dynamic membership can be supported via a standard level of indirection: a hash map mapping vertex IDs to \clusterforest leaves.
In Section~\ref{sec:skinny_dynamic}, we described how vertex updates can be implemented in the \sketchname-based dynamic connectivity algorithm.
Because the number of tiers cannot be increased at runtime, maintaining a low error rate under vertex insertions requires accommodating future vertices in the dense subsystem upfront.
Fortunately, the number of tiers needed to maintain a low error rate scales logarithmically with vertex count. Therefore, over-provisioning by just a few tiers prevents error-rate degradation while allowing orders of magnitude more dense vertices in the system.
}

\subsection{\densesysname}\label{sec:dense_subsystem}

\densesysname uses the implementation of \cupcake~\cite{deman2025fast} as a starting point, but makes many modifications. The dynamic tree and $\ell_0$-sketching primitives we use are both different than that of \cupcake,  we handle insertion/deletion of vertices which \cupcake does not support, and we improve the buffering techniques of \cupcake.

\myparagraph{Cutset Implementation}
Systems like \cupcake and \densesysname rely on multiple tiers of cutset data structures (see Section~\ref{sec:sketch_based_dyn_conn} for details).
In \densesysname, we implement the cutset data structure using link-cut trees~\cite{sleator1983data} based on splay trees~\cite{sleator1985self}.
Note that our theoretical analysis uses UFO trees~\cite{deman2026ufo} to implement the cutset data structure for simplicity of analysis, but in practice we use link-cut trees for their lower empirical memory footprint.
Each vertex in the link-cut tree is augmented with a \sketchname, which we develop a new implementation for.
Link-cut trees typically do not support subtree queries, but they work well for subtree queries in our use-case since sketch addition is an invertible function, and they only need to store one sketch per vertex (avoiding internal aggregate sketches).
Additionally, we optimize sketch update operations in the link-cut trees by only splaying an accessed node if its depth exceeds some constant times $\log \nodesize$.
\revision{In \appref{Appendix~\ref{app:results}}, we experimentally evaluate the performance of \sysname using \sketchname versus using \cameosketch as the $\ell_0$-sketching primitive.}

\myparagraph{Improved Batching}
The \cupcake system was optimized to rapidly process batches of updates that induce no structural changes in its cutsets.
To take advantage of this, they speculatively buffer batches of updates in hopes that the entire batch will induce no structural changes.
This has two major drawbacks.
The first is that if the batch of updates does induce a structural update, the system must undo some of the work it did and process each update in the batch one-by-one without the batch optimizations.
The second issue is that the update buffer must be flushed whenever a query occurs, so this optimization can be entirely ineffective in streams with many queries, and it can significantly increase query latencies.

\densesysname improves the buffering techniques of \cupcake so that any type of updates can be processed in a batch, including ones that induce structural changes. This overcomes the first drawback of \cupcake's buffering strategy.
Regarding the second drawback, our \sysname uses \clusterforest to answer queries, so flushing the update buffer immediately in \densesysname is not required for queries.
The only thing that requires flushing the buffer of updates is a deletion of a current spanning forest edge in \densesysname, which happens infrequently, and does not impact query latency.

\subsection{Hybrid Interface}

The hybrid interface manages communication and movement of edges and vertices between the \clusterforest and \densesysname systems, implementing the abstract logic of the algorithm from Section~\ref{sec:hybrid_dynamic}.
The set of currently dense vertices is stored in a hash table, and the total degree of each vertex is stored as an integer.
The hybrid interface also stores an \iblt for each vertex $v$ that allows for recovering all of the edges incident to a vertex being demoted from the dense region (see next paragraph for implementation details).
To optimize communication, edges that should be moved to the dense subsystem are collected in an update buffer rather than being sent immediately. This buffer is flushed to the dense subsystem in batches either when it reaches capacity or when an incoming update deletes an active edge from the dense spanning forest.
Finally, the hybrid interface stores a list of modifications made to the spanning forest during \densesysname updates. We appropriately append to this list while processing updates in \densesysname, and apply all these changes to the \clusterforest after finishing the updates.

\myparagraph{Recovery \iblts}
Our \iblt implementation is based on the original \iblt algorithm~\cite{goodrich2011iblt} and uses \cameosketch buckets as a building block.
Our specific implementation utilizes two tiers of \iblts with $k=3$ hash locations and $1.3\delta/8$ buckets and $\max\{\log_2 \nodesize, 0.2\delta/8\}$ buckets respectively for the two tiers, using xxHash~\cite{xxhash} with iterated seeds to ensure unique hash mapping.
Although the original \iblt algorithm's recovery is only correct w.h.p. in $\densethresh$ (not in $\nodesize)$, we found that the \iblt recovery succeeded in more than 99.9\% of cases in practice.
To ensure correctness of our overall system, any failed recovery operation simply aborts the demotion, retaining the vertex in the dense subsystem until a later demotion attempt. While this approach may keep light vertices in the dense subsystem longer, the exceptionally low error rate ensures that any resulting overhead in space or time executing demotions remains negligible.




%% file: 06_experiments.tex
\section{Experiments}
Each experiment reported in this section consists of running \sysname on an input stream comprised of edge insertions and deletions generated from some static graph (by inserting all the edges then deleting the set of inserted edges). 
All experiments were run on a 48-core AMD EPYC 7643 CPU with 256GB of RAM.

\begin{table}[ht]
    \small
    \centering
    \begin{tabular}{llrrrrc}
        \toprule
        Graph & Type & $\nodesize$ & $\edgesize$ & $2\edgesize/\nodesize$ & \begin{tabular}{@{}c@{}}Max. \\ $k$-Core\end{tabular} & Cite \\
        \midrule
        Google+          & social & \sigfig{107614}      & \sigfig{10803117}      & \AvgDeg{107614}{10803117}       & 676    & \cite{mcauley2012learning} \\
        Friendster       & social & \sigfig{65608366}    & \sigfig{1806067135}    & \AvgDeg{65608366}{1806067135}   & 304    & \cite{yang2012communities} \\
        Twitter          & social & \sigfig{41652231}    & \sigfig{1202513046}    & \AvgDeg{41652231}{1202513046}   & 2488   & \cite{twittergraph} \\
        Orkut            & social & \sigfig{3072627}     & \sigfig{117185083}     & \AvgDeg{3072627}{117185083}     & 253    & \cite{mislove2007measurement} \\
        \midrule
        ENWiki           & web    & \sigfig{4206289}     & \sigfig{91939727}      & \AvgDeg{4206289}{91939727}      & 145    & \cite{boldi2004webgraph} \\
        Youtube          & web    & \sigfig{1157828}     & \sigfig{2987624}       & \AvgDeg{1157828}{2987624}       & 51     & \cite{mislove2007measurement} \\
        \midrule
        RoadUSA          & road   & \sigfig{23947348}    & \sigfig{28854312}      & \AvgDeg{23947348}{28854312}     & 3      & \cite{roadgraphs} \\
        RoadGermany      & road   & \sigfig{12277375}    & \sigfig{16115881}      & \AvgDeg{12277375}{16115881}     & 3      & \cite{roadgraphs} \\
        \midrule
        SIFT-RS50K       & knn/rs & \sigfig{1000000}     & \sigfig{226803335}     & \AvgDeg{1000000}{226803335}     & 4301   & \cite{simhadri2021results} \\
        SIFT-KNN500      & knn/rs & \sigfig{1000000}     & \sigfig{410509248}     & \AvgDeg{1000000}{410509248}     & 499    & \cite{simhadri2021results} \\
        MSSP-RS10K      & knn/rs & \sigfig{1000000}     & \sigfig{547322605}     & \AvgDeg{1000000}{547322605}     & 18778  & \cite{simhadri2021results} \\
        MSSP-KNN500     & knn/rs & \sigfig{1000000}     & \sigfig{379809818}     & \AvgDeg{1000000}{379809818}     & 499    & \cite{simhadri2021results} \\
        \midrule
        Kron-13          & kron   & \sigfig{8192}        & \sigfig{16666534}      & \AvgDeg{8192}{16666534}         & 2827   & \cite{krongraphs} \\
        Kron-15          & kron   & \sigfig{32768}       & \sigfig{267289854}     & \AvgDeg{32768}{267289854}       & 11306  & \cite{krongraphs} \\
        Kron-16          & kron   & \sigfig{65536}       & \sigfig{1070763653}    & \AvgDeg{65536}{1070763653}      & 22623  & \cite{krongraphs} \\
        \bottomrule
    \end{tabular}
    \caption{Summary of the graph datasets in our experiments.}
    \label{tab:datasets}
\end{table}

\myparagraph{Datasets}
Table~\ref{tab:datasets} summarizes the datasets used in our experiments.
The social, web, and road graphs are sourced from SuiteSparse~\cite{suitesparse} and NetworkRepository~\cite{netrepo}, and are widely used in academic research.
The kron graphs are synthetic Kronecker random graphs generated according to the Graph500 specification~\cite{krongraphs}. They are studied in all prior graph sketching experimental research because they are dense and exhibit structural properties of real-world networks.

The $k$nn/rs graphs are symmetrized $k$-nearest neighbors ($k$nn) and range search (rs) graphs over a $1$M-point subset of the SIFT1B and MSSPACEV1B datasets from BIGANN~\cite{simhadri2021results}.
We generate these graphs using ParlayANN~\cite{parlayann2023} to compute the ground truth $k$nn or rs.
For $k$nn graphs, we use $k=500$. For rs graphs we used a squared euclidean distance radius of $50$K and $10$K respectively for SIFT1B and MSSPACEV1B.
The choice of $k$ or the search radius when constructing $k$nn/rs graphs directly controls graph density and is a central modeling question in large-scale graph construction pipelines~\cite{nndescent2011,stars2022,parlayann2023}.
Larger values of $k$ or radius yield denser graphs with richer connectivity, which typically improves the quality of downstream tasks such as graph-based clustering, semi-supervised learning, and retrieval, but substantially increases storage and processing costs. Smaller values yield sparser graphs that are cheaper to process but can under-connect the underlying data~\cite{stars2022}.
Because of this trade-off, practitioners building $k$nn/rs graphs must target a range of densities depending on the application, which motivates evaluating \sysname across moderately large $k$ and radius regimes.

\myparagraph{
\revision{Impact of Dense Cores}
}
Importantly, our list of datasets includes graphs with similar average densities but very different dense cores (see $k$-core column in Table~\ref{tab:datasets} which reports the largest $k$ such that the graph has a non-empty $k$-core).
The $k$-core of a graph is a popular metric for identifying well-connected communities within a graph; in particular, all vertices in a $k$-core have induced degree (i.e., degree restricted only within this subgraph) at least $k$.
%
%
Thus, the presence of a large $k$-core in a graph implies the existence of a subgraph where all vertices have high degree (in particular, the average degree is at least $k$).
The social network and road graphs both have relatively low average densities, but the social networks have much larger, denser cores than the road networks.
Similarly, for the rs and $k$nn graphs, average degrees are similar but the rs graphs have larger, denser cores.
This allows us to measure the impact of core size and density on our system performance. 

\myparagraph{Comparison Systems}
In these experiments we evaluate the performance of \clusterforest and \cupcake, since they are state-of-the-art implementations of dynamic connectivity systems optimized for sparse and dense graphs respectively. 
Our focus on these specific systems is intended to demonstrate the efficacy of the hybrid framework. Because the framework is modular, \sysname is designed so that any other sparse or dense optimized implementations developed in the future can be seamlessly integrated.

\begin{figure*}
    \centering
    \includegraphics[width=0.97\linewidth]{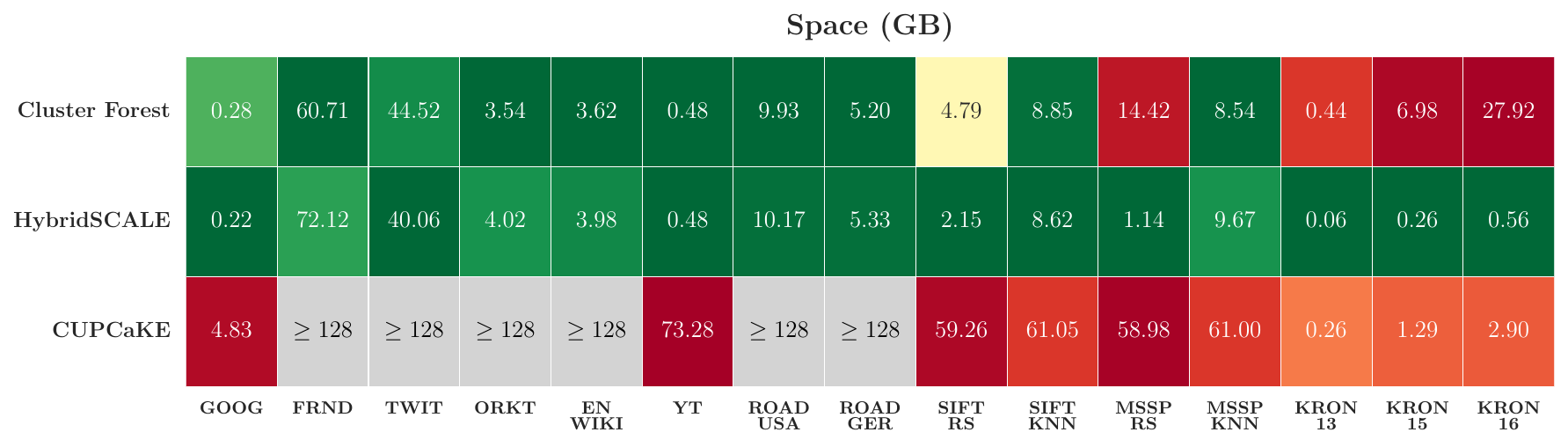}
    \vspace{-1em}
    \caption{\revision{Peak memory usage of \sysname, \clusterforest, and \cupcake on various input graph streams.
    The heat map uses per-column normalization (see \appref{Appendix ~\ref{app:results}}). Cell annotations indicate the raw values.}
    }
    \label{fig:memory_experiment_results}
\end{figure*}

\begin{figure*}[ht]
    \centering
    \includegraphics[width=0.97\linewidth]{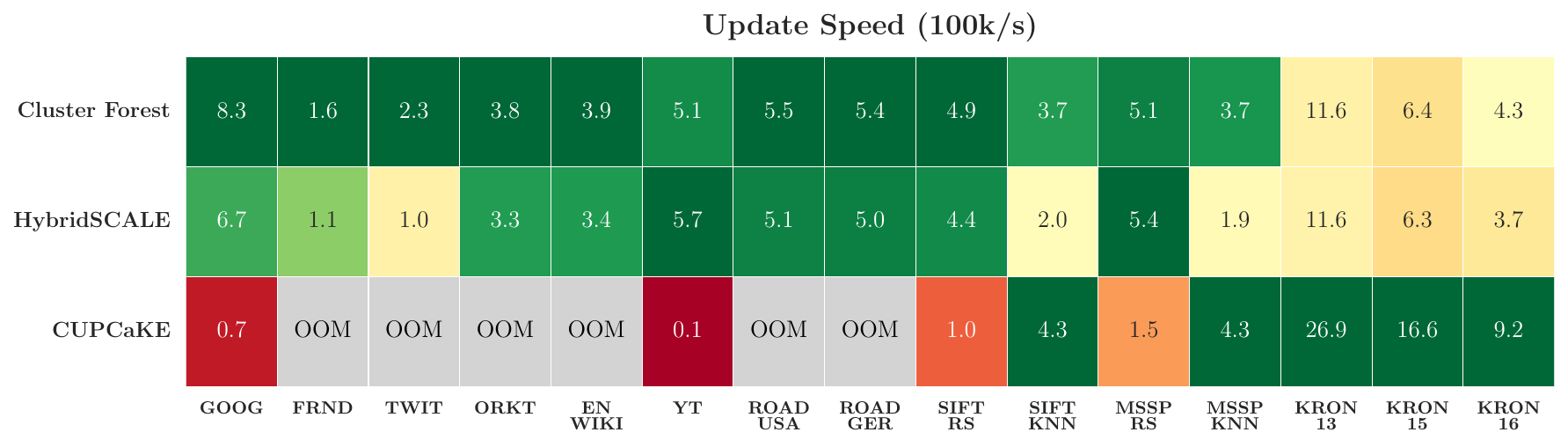}
    \vspace{-1em}
    \caption{
    \revision{Average update throughput for \sysname, \clusterforest, and \cupcake on various input graph streams.
    The heat map uses per-column normalization (see \appref{Appendix ~\ref{app:results}}). Cell annotations indicate the raw values.}
    }
    \label{fig:speed_experiment_results}
\end{figure*}

\begin{figure*}[ht]
    \centering
    \includegraphics[width=0.97\linewidth]{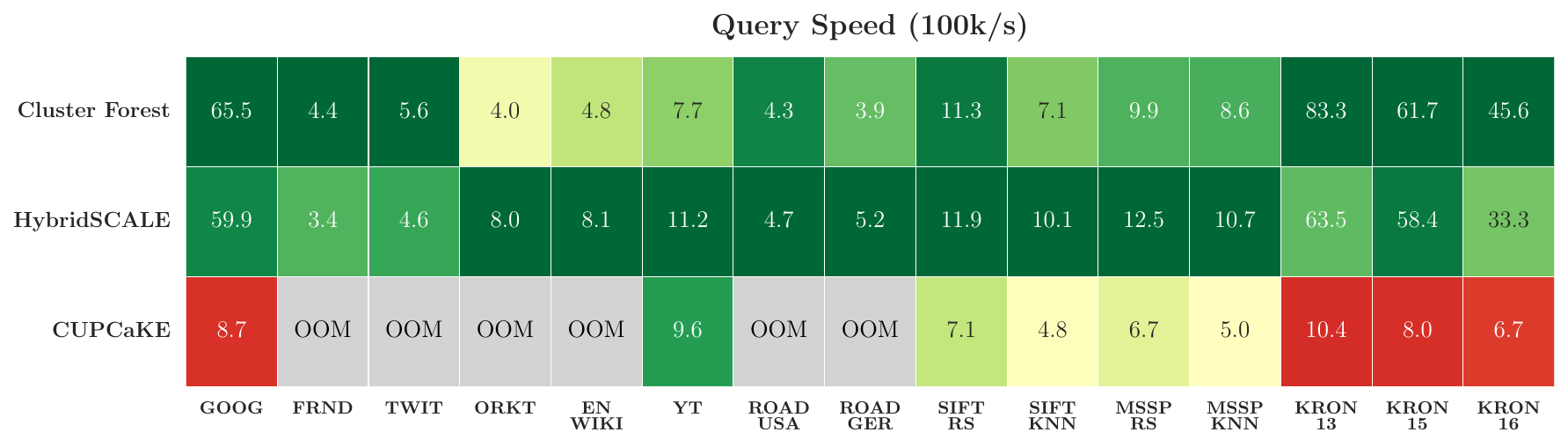}
    \vspace{-1em}
    \caption{
    \revision{Average query throughput for \sysname, \clusterforest, and \cupcake on various input graph streams.
    The heat map uses per-column normalization (see \appref{Appendix ~\ref{app:results}}). Cell annotations indicate the raw values.}
    }
    \label{fig:query_experiment_results}
\end{figure*}

\begin{figure}[t]
    \centering
    \includegraphics[width=\linewidth]{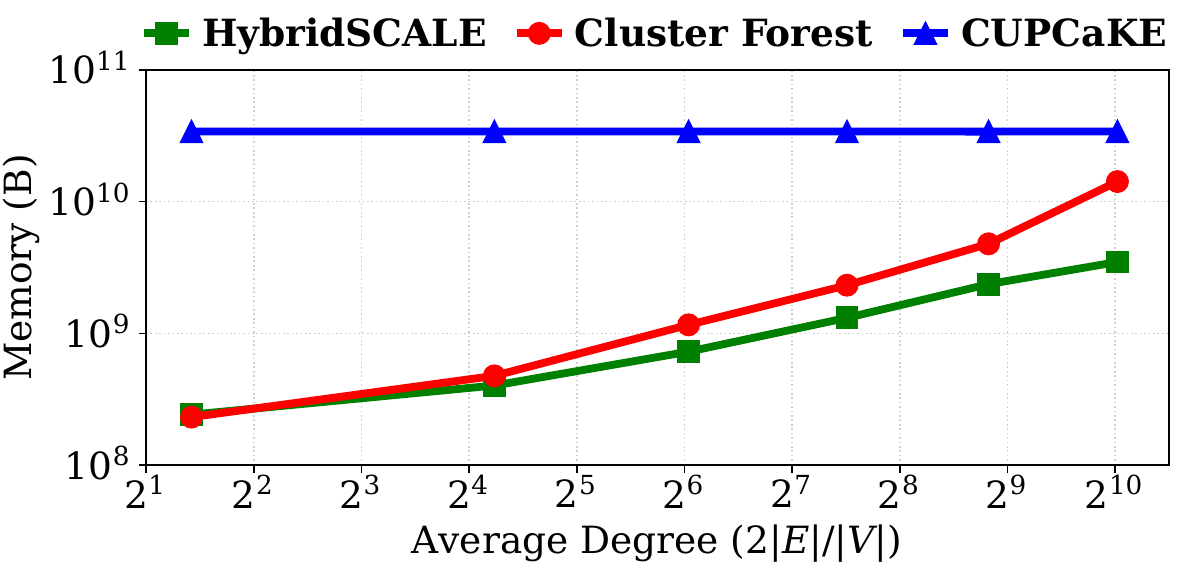}
    \vspace{-1em}
    \caption{Memory usage \sysname, \clusterforest, and \cupcake across various range search graphs constructed from a 1M-point subset of SIFT1B. The search radii are square Euclidean distances 10K, 20K, 30K, 40K, 50K, and 60K.}
    \label{fig:placeholder}
\end{figure}

\subsection{Memory Usage}

For each dataset, we measure the maximum space used by each system at any point during stream ingestion.

The results of this experiment are summarized in Figure~\ref{fig:memory_experiment_results}.
\revision{In \appref{Appendix~\ref{app:results}}, we also provide memory breakdowns by the various components of \sysname, memory usage with different choices of $\densethresh$, and memory usage with $\cameosketch$ instead of $\sketchname$.}

For each of the (sparse) social, web, and road graphs, the space used by \sysname is comparable to that of the \clusterforest, and \cupcake uses one to two orders of magnitude more space. \sysname is $1.27\times$ more compact than \clusterforest on the Google+ graph and $1.11\times$ more compact on the Twitter graph, because these graphs have high max $k$-core, indicating relatively large, dense cores.
The other social, web, and road graphs have smaller dense cores, leaving fewer opportunities to reduce the space cost by sketching.
Consequently \sysname uses slightly more space than \clusterforest due to the minor space overhead of our hybrid interface.

On the (dense) kron graphs, \sysname uses $1$--$2$ orders of magnitude less space than \clusterforest, because these graphs have very large average degree, making a lossless representation very large.
Note that \sysname is also several times smaller than \cupcake, primarily due to the space improvements of \sketchname, and the usage of link-cut trees instead of Euler-tour trees.

On the SIFT1B and MSSPACEV1B datasets, we can further see the importance of core size and density to \sysname's space performance.
On these graphs with $\nodesize = 10^6$, $\delta$, the \sysname density threshold is $25 \cdot \lceil \log_2(10^6) \rceil = 500$.
In the $k$nn graphs with $k = 500$ each vertex has degree at least $500$, therefore every vertex is included in the dense subsystem.
However due to the properties of $k$nn graphs, the max $k$-core is only $499$.
Thus these $k$nn graphs represent a sort of worst-case input for the hybrid framework where every vertex is just dense enough to be sketched, but the dense region does not have a high enough average degree to achieve meaningful space savings through sketching.
In contrast, the rs graphs have no limit on the degree of any vertex and tend to have larger, denser cores than their $k$nn counterparts with similar average degree.

This core size discrepancy directly affects the space performance of \sysname. For the rs graphs, \sysname is significantly more space-efficient than \clusterforest. Specifically, for SIFT1B \sysname uses a third of the space of \clusterforest, and for MSSPACEV1B it uses an order of magnitude less space. 
For the $k$nn graphs we observe that \sysname is moderately larger than \clusterforest. Since it cannot save much space by sketching the nonexistent core, its overheads for the sketching and hybrid driver data structures make it a bit larger than \clusterforest.

\myparagraph{Density Sweep}
Figure~\ref{fig:placeholder} further demonstrates the impact of core size, measuring the space usage of each system across rs graphs over the SIFT1B dataset with increasing squared euclidean distance radii. 
At low radius values, the dense core is too small to sketch efficiently and therefore \sysname and the \clusterforest are essentially the same size.
As the radius (and the size of the dense core) increases, \sysname becomes more compact than \clusterforest.
Note also that regardless of core size, \sysname is significantly smaller than \cupcake due to the reduced space complexity of \sketchname.
Due to this fact, \sysname can effectively leverage sketching to yield space savings at much lower average degrees than \cupcake (for \cupcake it must be in the low thousands).

\subsection{\revision{Update and Query Throughput}}

\myparagraph{\revision{Update Speed}}
For each dataset, we measured the total time of each system to \revision{insert all edges in the graph and then delete all edges, both in a random order}. The results are summarized in Figure~\ref{fig:speed_experiment_results}.
Our goal is to show that our hybrid framework incurs no significant degradation over the running time of its two subsystems.
For example, on sparser graphs where few edges are sketched (e.g. the road graphs) the performance of \sysname{} is essentially the same as that of cluster forest.
On the other hand, on graphs where most of the edges are sketched, (e.g. the $k$nn and kron graphs), the update performance of \sysname{} is closer to that of \cupcake{}.
%
%
%
%
%
For sparse graphs with reasonably dense cores (social and rs graphs), both subsystems of \sysname process a significant fraction of the edges, thus the running time is between \clusterforest and \cupcake.

\myparagraph{\revision{Query Speed}}
\revision{
To study query performance, we measure the total time of each system to perform connectivity queries on 20 million uniform-randomly generated vertex pairs.
Figure ~\ref{fig:query_experiment_results} shows the results.
As was the case for updates, we see that \sysname does not incur significant degradation in performance over the pure \clusterforest. In fact, in some cases the query performance is noticeably better.
This difference is surprising because \sysname relies directly on its \clusterforest subsystem to serve queries. However, the graph maintained by this subsystem differs from the one in pure \clusterforest: it is often substantially sparser, leading to a smaller memory footprint and improved cache locality.
The performance of \cupcake is always worse than \sysname.
}

\myparagraph{\revision{Mixed Workloads}}
\revision{
In \appref{Appendix~\ref{app:results}} we provide experimental results studying the performance of these systems under mixed workloads that interleave updates and queries.
This allows us to evaluate whether frequent queries degrade update throughput, or vice versa.
In summary, the results show that \sysname experiences little degradation in performance in mixed workloads.
On the other hand, there is a known issue with \cupcake~\cite{deman2025fast} losing performance on streams with queries interleaved between updates.
This degradation occurs in \cupcake because it relies on buffering updates for parallel processing. Queries trigger buffer flushes, which drastically diminishes the effectiveness of this optimization. As described in Section~\ref{sec:dense_subsystem}, \sysname adopts an improved version of this buffering strategy, but avoids this problem by answering queries directly from \clusterforest.}






%% file: 07_conclusion.tex
\section{Conclusion}


This work advances the story of practical graph sketching in several important ways.
Sketch-based dynamic graph algorithms have a very different space usage profile from lossless ones: their space scales with the number of vertices rather than the number of edges, which is extremely attractive on dense graphs but imposes a significant per-vertex overhead that makes them very costly on sparse real-world graphs.
This mismatch between where sketching shines in theory and where graphs actually sit in practice has long stood in the way of the broader adoption of graph sketching.
Our central message of this paper is that sketching does not need to win on every graph in order to be useful.
In particular, by applying sketching selectively to the dense cores of graphs and storing the rest losslessly, hybrid sketching can leverage sketching exactly where it helps most.
Hybrid sketching is thus a simple ``self-tuning'' data structure that works well on graphs with densities ranging from very sparse graphs all the way toward very dense regimes where sketching offers massive advantages.

Hybrid sketching is a general approach, not a single algorithm.
In particular, our system \sysname{} is agnostic to the specific lossless and sketch-based subroutines it wraps, so future work can freely substitute new components as the underlying algorithms improve.
For instance, we could replace the lossless subsystem with a heuristic structure such as the ID-Tree~\cite{IDTree}, which trades theoretical guarantees for fast empirical performance on low-diameter graphs, or plug in future improved sketching primitives.
Hybrid sketching should also be naturally applicable to other graph streaming problems such as cut and spectral sparsifiers, spanners, triangle counting, and other sketch-based dynamic graph data structures.
The broader takeaway of our work is that when used carefully, graph sketching already provides meaningful space savings for real-world graph datasets.

%% file: AA_balloon_proofs.tex
\section{\sketchname Analysis}
\label{appendix:balloon_proofs}

\subsection{Probability Preliminaries}

\subsubsection{Stochastic Dominance}

\begin{definition}[Stochastic Dominance]
\label{def:stochasticOrdering}
Let $X$ and $Y$ be random variables over $\mathbb{R}$. We use $Y \succeq X$ to denote that $Y$ \defn{stochastically dominates} $X$, meaning:

\begin{align*}
    \prob{Y \geq t} \geq \prob{X \geq t},
\end{align*}

for all $t$ in $\mathbb{R}$.

\end{definition}

This definition is equivalent to the statement that $F_X(t) \geq F_Y(t)$ for all $t$ in $\mathbb{R}$ (note the flipped direction), where $F_X(t)$ and $F_Y(t)$ are the cumulative distribution functions of $X$ and $Y$, respectively.
Stochastically ordered variables have a number of useful closure properties, namely that they are closed under summation of independent variables:

\begin{lemma}
\label{lem:stochasticOrderingSum}
Let $X_1, X_2, \dots X_n$ be a set of independent random variables, and let $Y_1, Y_2, \dots, Y_n$ be another set of independent random variables such that $X_i \preceq Y_i$ for all $i \in \positivesUpTo{n}$. Then:

\begin{align*}
    \sum_{i=1}^n X_i \preceq \sum_{i=1}^n Y_i.
\end{align*}
\end{lemma}
(Adapted from Theorem 1.A.3. in Shaked \& Shanthikumar \cite{shakedStochasticOrders2007})

\subsubsection{Concentration Bounds}

We will make use of the following tail bound for geometric random variables:

\begin{lemma}
\label{lemma:sumGeometric}
Let $\{Y_1, \dots, Y_n\}$ be the sum of $n$ independent geometric random variables with success probability $p$, and $k$ be some fixed constant. Define $Y := \sum_{i=1}^n Y_i$. Then $Y$ is a negative binomially distributed random variable such that:

\begin{align*}
    \prob{Y \geq k \expectation [Y]} \leq e^{\frac{-kn(1-1/k)^2}{2}}
\end{align*}
\end{lemma}

A simple proof of this lemma can be found in \cite{brownGeometricConcentration}.

\subsection{Bounding the Residual Depth of \\ \sketchname}

Regardless of the level of independence of the hash function $h$, the residual depth of a geometric sampling sketch column $\sketchColumn{h}$ is itself bounded by a geometric random variable. 
This follows from union bounding the probability that a bucket $\buck_i$, or any entry below it, is non-empty, and then applying Markov's inequality to bound the probability that any bucket of depth $i$ or greater is non-empty. 

The formal statement and proof are given below in Lemma \ref{lemma:residualdepth}.

\begin{lemma}
\label{lemma:residualdepth}
Let $M$ be the random variable corresponding to the residual depth of a sketch column $\sketchColumn{h}(\sketchvec)$ with input $\sketchvec \in \Field^\universesize$. Then $M$ is stochastically dominated by a geometric random variable $Y$ with $p=1/2$ and support $\{1, 2, \dots\}$.
\end{lemma}

\begin{proof}
    Let $\supportset := \supp(\sketchvec)$, and $\supportsize := ||\sketchvec||_0 = |\supportset|$.
    We note that for all $z \in \supportset$, it holds, by the definition of a geometric sketch column, that 
\begin{align*}
    \Pr \left[ z \in (Z \cap \buckset_i)\right] = \Pr \left[z \in \buckset_i \right] \leq \frac{1}{2^i}
\end{align*}

Thus, the expected number of non-zeros in $\buck_i$ is:

\begin{align*}
    \expectation \left[ \sum_{z \in \supportset} \mathbf{1}\{z \in \buckset_i \} \right] &= \sum_{z \in \supportset} \expectation \left[ \mathbf{1}\{ z \in \buckset_i\}\right] &\text{(linearity of expectation)}\\
    &= \sum_{z \in Z} \Pr[z \in B_i] \\
    &\leq \sum_{z \in Z} \frac{1}{2^i} = \frac{m}{2^i} \\
\end{align*}

We bound the probability that $\buck_i$ is non-empty, by Markov's inequality: 

\begin{align*}
    \Pr\left[|\buckset_i| > 0\right] &= \Pr \left[ \left(\sum_{z \in \supportset} \mathbf{1} \{ z \in \buckset_i \} \right) \geq 1\right] \\
    &\leq {\expectation \left[\sum_{z \in \supportset} \mathbf{1}\{z \in B_i \} \right]} \\
    &\leq \frac{m}{2^i}
\end{align*}

Let $w := \lceil \log_2 \supportsize \rceil$. Then for all $i \geq 0$:

\begin{align*}
    \Pr[|\buckset_{w + i}| > 0] &\leq \frac{\supportsize}{2^{w+i}} \\
    &\leq \frac{2^{\log_2 \supportsize}}{2^{\lceil \log_2 \supportsize \rceil}} \cdot \frac{1}{2^i} \\
    &\leq \frac{1}{2^i} \\
\end{align*}


Let $(G_1, G_2, \dots)$ be the sequence of random variables such that $G_k$ is equal to $1$ if there exists a non-empty bucket with depth greater than or equal to $w+k$, and $0$ otherwise.
Each $G_k$ is a Bernoulli random variable with $p = 2/2^{w+k}$, and is $1$ if and only if a non-empty bucket exists at or below $\buck_{w+k}$:

\begin{align*}
    \Pr\left[ G_k = 1 \right] := \Pr \left[\bigvee_{j = k}^{\infty} \left( \left| \buckset_{w+j + 1} \right| > 0 \right)\right]
    &\leq \frac{1}{2^{k}}\sum_{j=0}^{\infty} \frac{1}{2^j} \\
    &\leq \frac{2}{2^{k}} \\
    \Pr[G_0 =1 ] &\leq 1\\
\end{align*}

Let $X$ be a random variable equal to the smallest non-negative integer $j$ such that $b_{w+j}$ and every bucket in $\mathcal{C}_h$ below $|b_{w+j}|$ are empty.

\begin{align*}
    X := \min \{ j \in \mathbb{Z}^+ \mid G_j = 0\} 
\end{align*}

This definition implies that $w + j$ is the depth of $\mathcal{C}_h$, and thus such a $j$ would be the residual depth of $\mathcal{C}_h$, unless the residual depth is equal to $0$, in which case $j$ is one more than the residual depth. 

Accordingly, $X$ is at most the residual depth of $\mathcal{C}_h$, plus one.
We can thus say that, if $M$ is the R.V. for the residual depth, $M \leq X$ statewise (this is a stronger condition than $M \preceq X$, as it implies that $\Pr[X \geq k \mid M \geq k] =1$).


Since $G_b = 1 \implies G_a = 1$ and $G_a = 0 \implies G_b = 0$ for all $a \leq b$:

\begin{align*}
    M \geq k \implies X \geq k 
    &\iff \left(\bigwedge_{j = 0}^{k} \big(G_j = 1\big)\right) \iff G_k = 1 \\
\end{align*}
thus:

\begin{align*}
    \Pr[M \geq k] \leq \Pr[X \geq k] 
    &= \Pr[ G_k = 1] 
    \leq \frac{2}{2^{k}} \\
    \Pr[M \geq 0] &= 1\\
\end{align*}

Let $Y$ be a geometric random variable with $p=1/2$ and support $\mathbb{Z}^+$.
Then:

\begin{align*}
    \Pr[Y \geq k] &= \sum_{j=k}^{\infty} \frac{1}{2^j} \\ 
    &= \frac{1}{2^k} \sum_{j=0}^{\infty} \frac{1}{2^j} \\
    &= \frac{2}{2^k} \\
    &\geq \Pr[M \geq k] \\
\end{align*}

for $k \in \mathbb{Z}^+$. We deal with the case that of $k=0$ separetely (noting that, given how $Y$ is defined, $0$ is not typically in it's support):
\begin{align*}
    \Pr[Y \geq 0] \geq \Pr[Y \geq 1] &= 1 \\ 
    &\geq \Pr[M \geq 0]\\ 
\end{align*}

Therefore, it holds for all $t \in \mathbb{N}$ that $\Pr[Y \geq t] \geq \Pr[M \geq t]$ (thus, $Y \succeq M$).
\end{proof}






\begin{corollary}[Expected Residual Depth]
\label{cor:expectedresidualdepth}
Given a geometric sampling sketch column $\mathcal{C}_h$ with input $a \in \R^\universesize$, the expected residual depth is $\Theta(1)$.
\end{corollary}
\begin{proof}
Let $Y$ be a geometric random variable with $p=1/2$ and support $\mathbb{Z}^+$. 
Since $Y \succeq M$ by Lemma \ref{lemma:residualdepth}, it follows that $\expectation [M] \leq \expectation [Y] \leq 2$.
\end{proof}

\subsection{Aggregating Residual Depths Across Multiple Sketch Columns}
\label{sec:columnstacks}
Often, sketching algorithms use multiple sketch columns with independent randomness. 
The connectivity algorithms of AGM \cite{Ahn2012} (and it's derivatives such as \graphzep \cite{graphzeppelin}), for example, require $\Theta (\log n)$ sketches with constant success probability.  

This independent randomness, combined with the properties of sums of independent geometrically distributed random variables, allows one to bound the sums of the residual depths for geometric sampling sketch columns both in expectation \textbf{and} with high probability, which in turn is also an upper bound on the total number of deep buckets.  

Notably, as long as the hash functions $h$ drawn for each $\mathcal{C}_h$ have independent randomness, the value of the sketched vector $a$ does not matter. It could be identical for all columns, as is the case during streaming ingestion for the AGM sketch and it's variants.

\begin{lemma}
\label{lemma:columnstackexpected}
Given a set of sketch columns $\{\mathcal{C}_{h_1}, \dots, \mathcal{C}_{h_\numcolumns }\}$, the expected sum of residual depths is $\Theta(\numcolumns)$. 
\end{lemma}
\begin{proof}
    Let $M_i$ be the random variable for the residual depth of $\mathcal{C}_{h_i}$. By linearity of expectation,
\begin{align*}
    \expectation\left[\sum_{i=1}^\numcolumns M_i\right] &= \sum_{i=1}^\numcolumns \expectation[M_i] \\ 
    &\leq \sum_{i=1}^\numcolumns 2  \leq 2 \numcolumns 
    &\text{(Lemma \ref{lemma:residualdepth})}
\end{align*}

\end{proof}

\begin{theorem}
\label{thm:columnstackwhpresidual}
Given a set of sketch columns $\{\mathcal{C}_{h_1}, \dots, \mathcal{C}_{h_\numcolumns}\}$ with independent randomness, the sums of residual depths across all $\mathcal{C}_{h_i}$ is at most $\Theta(\numcolumns + \log(\universesize))$ with probability at least $1-1/\universesize^c$.  
\end{theorem}

\begin{proof}
Let $M_i$ be the random variable for the residual depth of $\mathcal{C}_{h_i}$. 
Let $M$ be $\sum_{i=1}^\numcolumns M_i$, the sum of residual depths across all columns.
Let $\{Y_1, \dots, Y_\numcolumns\}$ be a set of independent geometric random vaiables with $p=1/2$ and support $\mathbb{Z}^+$.
Let $Y := \sum_{i=1}^\numcolumns Y_i$. 
By Lemma \ref{lemma:sumGeometric}, $Y$ is a negative binomially distributed random variable such that:

\begin{align*}
    \Pr \left[Y \geq k \expectation[Y] \right] &\leq e^{\frac{-k\numcolumns(1-1/k)^2}{2}} \\
\end{align*}

Assume $\numcolumns \geq b \ln \universesize$ for some constant $b \in \R^+$. Choosing $k = 4bc$, for $c \geq 1$, and noting (by linearity of expectation) that $\expectation[Y] = 2L$:

\begin{align*}
    \Pr [Y \geq 8bc\numcolumns] &\leq e^{-2bc \numcolumns \cdot \frac{9}{16}} \\
     &\leq \left(e^{b\numcolumns}\right)^c \\
    &\leq \frac{1}{\universesize^c}
\end{align*}

Therefore, the probability that $Y$ exceeds $8bc\numcolumns$ is at most $1/\universesize^c$.

By Lemma \ref{lemma:residualdepth}, $M_i \preceq Y_i$ for all $i$.
Since $M = \sum M_i$ and $Y = \sum Y_i$, it follows from Lemma \ref{lem:stochasticOrderingSum} that $M \preceq Y$. Therefore,

\begin{align*}
    \Pr [M \geq 8bc\numcolumns] &\leq \Pr[Y \geq 8bc\numcolumns] &\text{($M \preceq Y$)} \\
    &\leq 1/\universesize^c 
\end{align*}

Note that in the case that $L < b \ln (n)$, we can simply suppose we add an additional $b \ln (n)$ geometric random variables to $Y$, and call this new variable $\hat{Y}$. From the tail bounds above it follows that $\hat{Y}$ is $\Theta(L + b \ln (n))$ with high probability, and since $Y \leq \hat{Y}$ statewise, we conclude that $Y$ is $\Theta(L + b \ln (n))$ with at least that probability. The application to $M$ via stochastic dominance follows exactly as above.
\end{proof}

\subsection{\sketchname Space and Time Complexity}

\balloonsketchthm*

\begin{proof}
    The total space for number of initialized non-deterministic buckets for a \sketchname is equal to
    $\lceil \log_2 \supportsize \rceil$ plus the residual depth.
    Theorem~\ref{thm:columnstackwhpresidual} bounds the total residual depth as $\Theta(\numcolumns + \log(\universesize))$ w.h.p..
    
    The update cost for a single coordinate $x$ is upper bounded by the sums of the depths derived from
    $h_i(x)$ for the $i$-th \sketchname; if the depth of the coordinate happens to be greater or equal to that of depth of the sketch, we may have to reallocate the \sketchname to shrink.
    Let $Y_i$ be the depth of a single update $x$ in column $i$, this is stochastically dominated by some geometric random variable $G_i$ by construction, and thus $\sum_{i=1}^L Y_i = \bigoh(L + \log\universesize)$ w.h.p. by an identical argument to the one found in Theorem ~\ref{thm:columnstackwhpresidual}. 
    Therefore, the update cost is $\bigoh(L + \log\universesize)$ w.h.p.  
    
    The correctness and sampling probability follow from the proof for \cameosketch found in ~\cite{landscape}.
\end{proof}





%% file: AB_hybrid_streaming_proofs.tex
\section{Hybrid Streaming Connectivity Proofs} \label{app:hybrid_streaming}

\begin{restatable}{lemma}{hybridstreamingspace} \label{lem:space_bounds}
    The total space usage of the hybrid streaming algorithm is $O(\nodesize + \edgesize)$ and $O(\nodesize \log \nodesize \log (2+\edgesize/\nodesize))$ w.h.p.
\end{restatable}

\begin{proof}
    First we prove the $O(\nodesize + \edgesize)$ w.h.p. space bound.
    We start by arguing that for each vertex $v$ with degree $\degree_v$, it uses $O(1+\degree_v)$ space w.h.p.
    If $v$ is in explicit form it requires $\Theta(1 + \degree_v)$ space.
    If $v$'s adjacency list is stored as a \sketchname matrix, then $\degree_v > \densethresh/2$, meaning $\degree_v = \Omega(\log \nodesize \log \log \nodesize))$. Per Theorem~\ref{thm:balloon_sketch}, and using $\numcolumns = \Theta(\log \nodesize)$, the space usage is $O(\numcolumns \log \degree_v + \log \nodesize) = O(\log \nodesize \log \degree_v)$ w.h.p.
    Since $\degree_v$ is at most $\nodesize$, the space used by a \sketchname matrix for $v$ is also $O(1+\degree_v)$ w.h.p.
    Summing over all vertices, the total space usage is:
    \[ \sum_{v \in \nodes} O(1 + \degree_v) = O(\nodesize) + O\left(\sum_{v \in \nodes} \degree_v \right) = O(\nodesize + \edgesize) \text{ w.h.p.} \]

    Next we prove the $O(\nodesize \log \nodesize \log (2+\edgesize/\nodesize))$ w.h.p. space bound.
    In this case, we will argue that each vertex $v$ with degree $\degree_v$ uses $O(\log \nodesize \log (2 + \degree_v))$ space w.h.p.
    For a vertex in explicit form, we have $\degree_v \leq \densethresh = \Theta(\log \nodesize \log \log \nodesize)$. At the threshold ($\degree_v = \densethresh$), the set's space usage (BBST) is $\Theta(\log \nodesize \log \log \nodesize)$, which matches the desired bound of $O(\log \nodesize \log(2+ \degree_v)) = O(\log \nodesize \log (2+\log \nodesize \log \log \nodesize)) = O(\log \nodesize \log \log \nodesize)$. For all degrees below the threshold ($\degree_v < \densethresh$), the bound $O(\log \nodesize \log(2+\degree_v))$ remains valid because the ratio, $R(d) = d / \log(2+d)$, is monotonically increasing for $d \geq 1$.

    For a vertex in sketch form, we have $\degree_v = \Omega(\log \nodesize \log \log \nodesize)$.
    Per Theorem~\ref{thm:balloon_sketch}, the space of the \sketchname matrix with $\numcolumns = \Theta(\log \nodesize)$ is $O(\log \nodesize \log (2 + \degree_v))$ w.h.p. (since the number of nonzero elements in the sketched vector is $m = \degree_v$).
    By a union bound over all $\nodesize$ vertices, the total space usage is:
    \[ \sum_{v \in \nodes} \log \nodesize \log (2+\degree_v) = O(\nodesize) + O(\log \nodesize) \sum_{v \in \nodes} \log(2 + \degree_v) \text{ w.h.p.} \]
    Since $\log(2+x)$ is concave, applying Jensen's inequality:
    \[ \frac{1}{\nodesize} \sum_{v \in V} \log(2 + \degree_v) \leq \log \left( 2 + \frac{1}{\nodesize} \sum_{v \in V} \degree_v \right) = \log \left( 2 + \frac{2\edgesize}{\nodesize} \right). \]
    Substituting this back into the total space expression:
    \[ O(\nodesize) + \nodesize \cdot O(\log \nodesize) \cdot \log \left( 2 + \frac{2\edgesize}{\nodesize} \right) = O\left(\nodesize \log \nodesize \log \left( 2 + \frac{\edgesize}{\nodesize} \right) \right) \text{ w.h.p.} \]

    Finally, we note that the \iblt for a vertex $v$ in sketch form uses $\Theta(\log \nodesize \log \log \nodesize)$ space which is always dominated by the sketch space since $\degree_v = \Omega(\log \nodesize \log \log \nodesize)$, so the space used by the sketch is $\Omega(\log \nodesize \log \log(2+ \log \nodesize \log \log \nodesize)) = \Omega(\log \nodesize \log \log \nodesize)$ w.h.p.
\end{proof}

\begin{restatable}{lemma}{hybridstreamingupdate} \label{lem:amortized_update}
    The hybrid streaming algorithm processes each edge insertion and deletion in $O(\log \nodesize)$ amortized time with high probability.
\end{restatable}

\begin{proof}
    First, we bound the base cost of an update that does not trigger a vertex transition. As established, explicit sets are implemented as BBSTs, so insertion, deletion, and search times are all loosely upper bounded by $O(\log \nodesize)$. If the edge connects two sketch-form vertices, inserting or deleting the edge requires updating the \sketchname matrix and \iblt of both endpoints, which takes $O(\log \nodesize)$ time w.h.p. Thus, the base cost of any single stream update is bounded by $O(\log \nodesize)$ w.h.p.

    Next, we bound the computational cost of vertex promotions and demotions. A vertex $u$ is promoted from explicit form to sketch form when its degree reaches $\densethresh$. Vertex $u$ initializes a new sketch and \iblt, and redistributes its $\densethresh$ explicitly stored edges. For each edge, $u$ either adds it to its new sketch and issues an update to a sketch-form neighbor (costing $O(\log \nodesize)$ w.h.p.), or pushes it to an explicit-form neighbor's set (costing $O(\log \nodesize)$). Summing over all $\densethresh$ edges, the total cost of a promotion is $O(\densethresh \log \nodesize)$ w.h.p.

    Similarly, a vertex $u$ demotes from sketch form to explicit form when its degree drops to $\densethresh / 2$. Vertex $u$ decodes its \iblt to recover at most $\densethresh / 2$ edges (taking $O(\densethresh)$ time), builds a new explicit set (taking $O(\densethresh \log \nodesize)$ time), and issues deletions to its sketch-form neighbors (taking $O(\densethresh \log \nodesize)$ time w.h.p.). The total cost of a demotion operation is $O(\densethresh \log \nodesize)$ w.h.p.

    Once a vertex is promoted or demoted, it must experience an absolute net change of at least $\densethresh/2$ edge insertions or deletions before it can transition again. We can therefore distribute the $O(\densethresh \log \nodesize)$ transition cost evenly across these $\Omega(\densethresh)$ required stream updates, yielding an amortized cost of $O(\log \nodesize)$ per update.
\end{proof}

\myparagraph{Connected Components Queries}
Our connected component query algorithm uses ideas from that of Ahn~\etal (described in Section~\ref{sec:streaming_prelims}). However, their algorithm does not work directly because we have some vertices store their adjacency list in explicit form. 
Recall that the $\ell_0$ sketching math requires every internal edge to be stored symmetrically by both endpoints so they cancel out via XOR summation. If an edge were stored explicitly by an array on one side and sketched on the other, it would fail to cancel, resulting in the query failing. Our asymmetric update rule was designed specifically to prevent this: by keeping mixed edges out of the sketches entirely and only sketching edges when both endpoints are in sketch form, we preserve the required mathematical parity. This allows us to safely execute the query using a multi-phase algorithm.

First, we run a standard spanning forest algorithm (e.g. BFS) over all edges stored in explicit form across the entire graph.
This effectively collapses the graph into a set of disjoint super-components, such that any edges between them are stored as sketch-to-sketch edges in our hybrid algorithm.
For each resulting super-component $C$, we generate its aggregate \sketchname matrix by computing the bitwise XOR sum of the sketches of all sketch-form vertices $v \in C$. Vertices in explicit-form contribute nothing to this sum.
We treat the aggregated super-components as individual super-vertices and execute the standard Bor\r{u}vka-style query algorithm of Ahn~\etal on their aggregated \sketchname structures to find the remaining cross-component edges and finalize the spanning forest.
We prove in Appendix~\ref{app:hybrid_streaming} that queries are correct w.h.p. and take $O(\nodesize \log^2 \nodesize)$ time.

\begin{restatable}{lemma}{hybridstreamingquery} \label{lem:query_correctness_time}
    The hybrid streaming connected components query algorithm runs in $O(V \log^2 V)$ time and is correct w.h.p.
\end{restatable}

\begin{proof}
    We bound the time complexity by analyzing the multi-phase execution. The first phase runs a standard spanning forest algorithm (e.g., BFS) over the explicitly stored edges, which takes $O(V + E_{\text{exp}})$ time. Because vertices only store explicit edges up to the density threshold $\delta = \Theta(\log V \log \log V)$, the total number of explicit edges is $E_{\text{exp}} \le V \delta$. This BFS phase therefore strictly takes $O(V \log V \log \log V)$ time. In the second phase, aggregating the sketches and executing the standard AGM-style Bor\r{u}vka query takes at most $O(V \log^2 V)$ time. This second phase dominates the runtime, giving an overall time complexity of $O(V \log^2 V)$.

    For correctness, the first phase is deterministically correct: the BFS evaluates explicitly stored edges and perfectly groups the graph into super-components. The algorithm then relies on the standard AGM framework to resolve the remaining sketch-to-sketch edges. Because the $\ell_0$ samplers guarantee the discovery of cross-component edges with high probability, the overall algorithm correctly outputs the final spanning forest with high probability.
\end{proof}

%% file: AC_balloon_dc_proofs.tex
\section{\sketchname Dynamic Connectivity Analysis}
\label{app:balloon_dc}

\subsection{\revision{UFO Tree Details}}

\revision{
UFO trees maintain a balanced, hierarchical clustering of the vertices in a dynamic forest. 
For a given tree in the forest, the UFO tree represents it as a multi-level structure that facilitates bottom-up aggregation.
At the base level (level $0$), each vertex of the underlying tree corresponds to a single cluster. 
At each subsequent level $\ell > 0$, adjacent clusters from level $\ell-1$ are merged together based on a certain set of rules, creating a new set of clusters.
}

\revision{
A crucial property of UFO trees is that the number of clusters decreases by a constant fraction at each level. Specifically, the clustering strategy guarantees that the number of nodes at level $\ell$ is at most $5/6$ of the number of nodes at level $\ell-1$. 
Because of this geometric decay, the height of the UFO tree is bounded by $\bigoh(\log \nodesize)$.
}

\revision{
Furthermore, at any given level $\ell$, the nodes of the UFO tree collectively form a disjoint partitioning of the vertices in the original tree.
UFO trees support augmenting each vertex with a value in domain $D$.
Each internal node in the UFO tree can store an aggregate value representing the value of some commutative and associative function $f: D^2 \rightarrow D$ applied over the augmented values for each vertex in its cluster.
In our \sketchname-based dynamic connectivity algorithm, this augmented value is a \sketchname, and the function is sketch addition.
}

\revision{
When edges are inserted or deleted, the UFO tree algorithm dynamically re-clusters the affected regions and updates these aggregates. The structural guarantees ensure that each dynamic operation modifies only a logarithmic number of nodes, allowing us to bound the update cost.
As we will see, the strict upper bounds on the number of nodes per level directly allow us to bound the total space consumed by the aggregate sketches across the entire structure.
}

\subsection{\revision{\sketchname-Based DC Proofs}}

\balloondcspace*

\begin{proof}
    We prove the space bound.
    Our data structure consists of $T = \Theta(\log(\nodesize))$ tiers, matching the number needed in \cupcake, each with a UFO tree $\ufo^t$ maintaining aggregates for a subtree query (where the aggregate is a single \sketchname). 
    For each UFO tree $\ufo^t$, there are levels $\ufo^{0, t}, \ufo^{1, t}, \ldots, \ufo^{\ell, t}, \ldots, \ufo^{\bigoh(\log(\nodesize)), t}$.
    Level $\ell$ of a UFO tree has at most $5/6$ as many nodes as level $\ell-1$. 
    Each level $\ell$ consists of $N_{\ell, t} \leq (5/6)^\ell \nodesize$ nodes, each of which with an aggregate \sketchname.
    Each node represents a subset of the vertices, and the nodes at each level $\ell$ collectively form a partitioning of the vertices.

    Let $U^{\ell,t}_i$ be the $i$th node at level $\ell$ of the tier $t$ UFO tree for $1 \leq i \leq N_{\ell, t}$ 
    (assigning an arbitrary ordering to the nodes at each level). 
    Let $d^{\ell, t}_i$ be the number of edges in the cut corresponding to $(\nodes, \nodes \setminus U^{\ell,t}_i)$.
    Intuitively $d^{\ell, t}_i$ is the ``degree'' (number of edges coming out) of the $i$-th cluster at level $\ell$ in the UFO tree for tier $t$.
    The depth of the \sketchname for $U^{\ell, t}_i$ is:

    $$
        \lg(d^{\ell, t}_i) + Y^{\ell, t}_i
    $$

    where $Y^{\ell, t}_i$ is a random variable (determined by the randomness of the sketch seed) that is stochastically dominated by a geometric random variable.
    The sum of depths for all internal UFO tree nodes across all $T$ tiers is bounded by:

    \begin{align*}
        \sum_{t=0}^{T-1} \sum_{\ell=0}^{\bigoh(\log(\nodesize))} \sum_{i=1}^{N_{\ell, t}} \left(\lg\left(d^{\ell, t}_i\right) + Y^{\ell, t}_i \right)
    \end{align*}
    
    First, we show that the contribution from the $\lg(d^{\ell, t}_i)$ terms is $\bigoh(\nodesize \log\nodesize \log(\edgesize/\nodesize))$.
    For a fixed tier $t$ and fixed level $\ell$, we know that $N_{\ell, t} \leq (5/6)^{\ell} \nodesize$. 
    Let $E_{\ell} \leq E$ be the number of \defn{external} edges among our level $\ell$ partition, i.e., the edges that exist between two different nodes at level $\ell$ but not the edges internal to a node.
    Applying Jensen's inequality:

    \begin{align*}
        D(\ell, t) := \sum_{i=1}^{N_{\ell, t}} \lg\left(d^{\ell, t}_i\right) &\leq N_{\ell, t} \lg\left(\frac{\sum_{i=1}^{N_{\ell, t}} d^{\ell, t}_i}{N_{\ell, t}}\right) \\
        &\leq N_{\ell, t} \lg\left(\frac{2 E_{\ell}}{N_{\ell, t}}\right) \\
        &\leq \alpha \left(\frac{5}{6}\right)^{\ell} \nodesize \lg\left(\frac{2\edgesize}{\alpha \left(\frac{5}{6}\right)^{\ell} \nodesize}\right) \\
        &\leq \alpha \left(\frac{5}{6}\right)^{\ell} \left(1 - \lg(\alpha) + \ell \lg\left(\frac{6}{5}\right)\right) \nodesize \lg\left(\frac{2\edgesize}{\nodesize}\right)
    \end{align*}
    
    Where $\alpha \in (0,1)$ is an unknown real number such that $N_{\ell, t} = \alpha (5/6)^{\ell} \nodesize$. Upper bounding the r.h.s. with respect to  $\alpha$ gives 
    
    \begin{align*}
        D(\ell, t)
        &\leq \alpha \left(\frac{5}{6}\right)^{\ell} \left(1 - \lg(\alpha) + \ell \lg\left(\frac{6}{5}\right)\right) \nodesize \lg\left(\frac{2\edgesize}{\nodesize}\right) \\
        &\leq \left(\frac{5}{6}\right)^\ell \cdot \left(\alpha \left(1 - \lg(\alpha)\right) + \alpha \ell \lg\left(\frac{6}{5}\right)\right) \nodesize \lg\left(\frac{2\edgesize}{\nodesize}\right) \\
        &\leq \left(\frac{5}{6}\right)^\ell \cdot \left(1.07 + \ell \lg\left(\frac{6}{5}\right)\right) \nodesize \lg\left(\frac{2\edgesize}{\nodesize}\right) \\
    \end{align*}
    
    The upper bounding of $\alpha(1 - \lg(\alpha))$ was found analytically (the maximum is achieved at $\alpha = 2/e$).

    Across all levels $\ell$ of a fixed tier $t$, the sum of depths of stored sketch aggregates is:

    \begin{align*}
        \sum_{\ell=0}^{\bigoh(\log \nodesize)} D(\ell, t) &\leq \sum_{\ell=0}^{\bigoh(\log \nodesize)} \left(\frac{5}{6}\right)^\ell \cdot \left(1.07 + \ell \lg\left(\frac{6}{5}\right)\right) \nodesize \lg\left(\frac{2\edgesize}{\nodesize}\right) \\
        &\leq \nodesize \lg\left(\frac{2\edgesize}{\nodesize}\right) \cdot \sum_{\ell=0}^{\bigoh(\log \nodesize)} \left(\frac{5}{6}\right)^\ell \cdot \left(1.07 + \ell \lg\left(\frac{6}{5}\right)\right) \\
        &\leq \nodesize \lg\left(\frac{2\edgesize}{\nodesize}\right) \cdot \left(1.07 \sum_{\ell=0}^\infty \left(\frac{5}{6}\right)^\ell + \lg\left(\frac{6}{5}\right) \sum_{\ell=0}^\infty \ell\left(\frac{5}{6}\right)^\ell \right) \\
        &\leq \nodesize \log\left(\frac{2\edgesize}{\nodesize}\right) \cdot \left(1.07 \cdot 6 + \lg\left(\frac{6}{5}\right) \cdot 30 \right) \\
        &\leq 15 \nodesize \log\left(\frac{2\edgesize}{\nodesize}\right) \\
    \end{align*}
    
    Summing across all $T$ tiers:

    \begin{align*}
        \sum_{t=0}^{T-1} \sum_{\ell=0}^{\bigoh(\log(\nodesize))} \sum_{i=1}^{N_{\ell, t}} \lg\left(d^{\ell, t}_i\right) &=
        \sum_{t=0}^{T-1} \sum_{\ell=0}^{\bigoh(\log(\nodesize))} D(\ell, t) \\
        &\leq \sum_{t=0}^{T-1} 15 \nodesize \log\left(\frac{2\edgesize}{\nodesize}\right) \\
        &= \bigoh\left(\nodesize \log\nodesize \log\left(\frac{2\edgesize}{\nodesize}\right)\right)
    \end{align*}
    
    Meanwhile, the sums of residual depths across all sketches is:

    \begin{align*}
        \sum_{t=0}^{T-1} \sum_{\ell=0}^{\bigoh(\log(\nodesize))} \sum_{i=1}^{N_{\ell, t}} Y^{\ell, t}_i \\
    \end{align*}
    Let $L$ be the maximum level across all $T$ tiers and note that $L = \bigoh(\log(V))$. Let

    \begin{align*}
        \hat{Y}^{\ell, t}_i &= \begin{cases}
            Y^{\ell, t}_i &\text{if UFO node $U^{\ell, t}_i$ exists} \\
            0 &\text{otherwise}
        \end{cases}
    \end{align*}
    
    and note:
    \begin{align*}
        \sum_{t=0}^{T-1} \sum_{\ell=0}^{L} \sum_{i=1}^{N_{\ell, t}} Y^{\ell, t}_i &= \sum_{t=0}^{T-1} \sum_{\ell=0}^{L} \sum_{i=1}^{(5/6)^\ell} \hat{Y}^{\ell, t}_i \\
        &= \sum_{\ell=0}^{L} \sum_{i=1}^{(5/6)^\ell} \left( \sum_{t=0}^{T-1} \hat{Y}^{\ell, t}_i  \right)\\
    \end{align*}

    Focusing on a fixed $\ell$ and $i$ at the sum over the tiers, define new geometric random variables $G_{0}, \dots G_{T-1}$. By the definition of $Y_{i}^{\ell, t}$ it is clear that $G_{i} \succeq \hat{Y}_{i}^{\ell, t}$. Furthermore, since the randomness of $\hat{Y}^{\ell, t}_i$ is determined by the tier $t$ and each R.V. in the set $\{\hat{Y}^{\ell, t}_i\ \mid 0 \leq t \leq T-1 \}$ has a unique $t$, the set of $\hat{Y}^{\ell, t}_i$ are mutually independent. Lemma \ref{lem:stochasticOrderingSum} then implies:

    \begin{align*}
        \sum_{t=0}^{T-1} \hat{Y}^{\ell, t}_i &\preceq \sum_{t=0}^{T-1} G_t
    \end{align*}
    
    and therefore:

    \begin{align*}
        \sum_{t=0}^{T-1} \sum _{\ell=0}^L \sum_{i=1}^{N_{\ell, t}} Y^{\ell, t}_i &\preceq \sum_{\ell=0}^L \sum_{i=1}^{(5/6)^\ell} \left( \sum_{t=0}^{T-1} G_t  \right)\\
    \end{align*}
    
    By Lemma \ref{lemma:sumGeometric}, w.h.p., $\sum_{t=0}^{T-1} G_t = \bigoh(T + \log(\nodesize)) = \bigoh(\log(\nodesize))$ for a fixed $\ell$ and $i$. Union bounding over all $\bigoh(\nodesize)$ choices of $(\ell, i)$:
    
    $$
    \sum_{t=0}^{T-1} \sum _{\ell=0}^L \sum_{i=1}^{N_{\ell, t}} Y^{\ell, t}_i = \bigoh(\nodesize \log(\nodesize))
    $$ 
    with high probability, which is dominated by the $\bigoh(\nodesize \log \nodesize \log(\edgesize/\nodesize))$ minimum depth implied by the $\lg(d^{\ell, t}_i)$ terms. Therefore, the total space usage is $\bigoh(\nodesize \log\nodesize \log(\edgesize/\nodesize))$ with high probability.

    The update bound and the query bound follow from the fact that the costs of \link and \cut operations in UFO tree is asymptotically equivalent to other dynamic tree data structures, and we do not modify the overall update and query algorithms of \gibb.

    For the sake of clean arithmetic we assumed WLOG that each vertex (and each cluster) has a degree of at least $2$; if a vertex has degree zero, it still uses $\Theta(1)$ words of space per each $\Theta(\log\nodesize)$ tier,
    and thus the space is always at least $\Omega(\nodesize \log \nodesize)$. Our final space usage thus has a plus two added to the second logarithm to account for this,
    yielding $\bigoh(\nodesize \log \nodesize \log (2+ \edgesize / \nodesize))$.
\end{proof}

%% file: AD_hybrid_dynamic_proofs.tex
\section{Hybrid Dynamic Connectivity Proofs} \label{app:hybrid_dynamic}

\begin{restatable}{lemma}{hybriddynamicspaceone} \label{lem:hybrid_dynamic_space_one}
    The space complexity of our hybrid dynamic connectivity algorithm is $O(\nodesize + \edgesize)$ w.h.p.
\end{restatable}

\begin{proof}
    Let $\nodesize_D$ and $\edgesize_D$ be the number of vertices and edges in $\densealg$ respectively. Let $\densespace$ be the random variable for the space usage of $\densealg$, and $\sparsespace$ be the space usage of $\sparsealg$.
    Due to Lemma~\ref{lem:balloon_dc}, we know that $\densespace = O(\nodesize_D \log\nodesize \log(2 + \edgesize_D / \nodesize_D))$ with high probability in $\nodesize$.
    If $\edgesize_D \geq \nodesize_D \log\nodesize \log(2 + \edgesize_D / \nodesize_D)$, then this space is $O(\edgesize)$ with high probability since $\edgesize \geq \edgesize_D$.
    Otherwise,
    $$\edgesize_D \leq \nodesize_D \log\nodesize \log(2 + \edgesize_D / \nodesize_D) \leq 2\nodesize_D \log^2 \nodesize.$$
    \begin{align*}
    \densespace &= O(\nodesize_D \log\nodesize \log(2 + \edgesize_D / \nodesize_D)) \\
    &= O(\nodesize_D \log\nodesize \log(2 + 2\log^2 \nodesize)) \\
    &= O(\nodesize_D \log\nodesize \log\log\nodesize) \text{ w.h.p.}
    \end{align*}
    
    Every vertex in $\densealg$ has degree at least $\densethresh/2 = \log\nodesize \log\log\nodesize / 2$, therefore the total number of edges is $\edgesize \geq \nodesize_D \log\nodesize \log\log\nodesize / 4$.
    Thus, in this case as well, the space in $\densealg$ is $O(\edgesize)$ with high probability.

    $\sparsespace$ is $O(\nodesize + \edgesize)$ since the graph stored by $\sparsealg$ is a subgraph of $\graph$, and the \clusterforest algorithm uses linear space.
    The space usage in our degree storage is clearly $O(\nodesize)$.
    The space usage of the $\nodesize_D$ $\delta$-sparse recovery sketches is $O(\nodesize_D \log\nodesize \log\log\nodesize)$, which again is $O(\edgesize)$ since $\edgesize \geq \nodesize_D \log\nodesize \log\log\nodesize / 4$.
    Therefore, the total space complexity $\densespace + \sparsespace$ along with auxiliary structures is $O(\nodesize + \edgesize)$ with high probability.
    
    The space usage in \iblt's is $O(\nodesize_D \log \nodesize \log \log \nodesize)$ which is dominated by $\densespace$. The boolean and global degree for each vertex only take $O(\nodesize)$ total space.
\end{proof}

\begin{restatable}{lemma}{hybriddynamicspacetwo} \label{lem:hybrid_dynamic_space_two}
    The space complexity of our hybrid dynamic connectivity algorithm is $O(\nodesize \log \nodesize \log (2+\edgesize / \nodesize))$ w.h.p.
\end{restatable}

\begin{proof}

    Let $\nodesize_D$ and $\edgesize_D$ be the number of vertices and edges in $\densealg$ respectively. Let $\densespace$ be the random variable for the space usage of $\densealg$, and $\sparsespace$ be the space usage of $\sparsealg$.

    We know that $\densespace = O(\nodesize_D \log\nodesize \log(2 + \edgesize_D / \nodesize_D))$ with high probability in $\nodesize$.
    Note that $f(x) = x \log (2 + y/x)$ is a monotonically increasing function for positive $x$ and a fixed positive $y$.
    Thus, since $\nodesize_D \leq \nodesize$ and $\edgesize_D \leq \edgesize$, we have $\densespace = O(\nodesize \log\nodesize \log(2 + \edgesize / \nodesize))$ with high probability in $\nodesize$.
    Because of the linear space complexity of the \clusterforest algorithm, $\sparsespace = O(\nodesize + \edgesize_S)$, where $\edgesize_S$ is the number of edges in $\sparsealg$.
    The spanning forest of $\densealg$ contributes at most $\nodesize - 1$ edges to $\edgesize_S$. The remaining edges must be sparse edges (the set of sparse edges may overcount the remaining edges in $\sparsealg$ because some sparse edges may be in $\densealg$ since we only demote vertices with degree $< \delta/2$, but no dense edges other than the spanning forest can be in $\sparsealg$).
    
    To upper bound the number of sparse edges, we can sum the degrees of each sparse vertex. This results in edges between two sparse vertices being counted twice, and edges between a dense and a sparse vertex being counted once.
    Since every sparse vertex has degree $< \delta$, the number of sparse edges is $< \nodesize \delta$, so $\edgesize_S < \nodesize \log \nodesize \log\log\nodesize$.
    
    If $\edgesize < \nodesize \log\nodesize \log(2+\edgesize/\nodesize)$ then $\sparsespace = O(\nodesize \log\nodesize \log(2+\edgesize/\nodesize))$ is trivially true.
    Else $\edgesize \geq \nodesize \log\nodesize \log(2+\edgesize/\nodesize)$ and,
    
    \begin{align*}
        \nodesize \log\nodesize \log\log\nodesize &\leq \nodesize \log\nodesize \log\left(2+\frac{\nodesize\log\nodesize}{\nodesize}\right) \\
        &\leq \nodesize \log\nodesize \log\left(2+\frac{\nodesize \log\nodesize \log(2+\edgesize/\nodesize)}{\nodesize}\right) \\
        &\leq \nodesize \log\nodesize \log\left(2+\frac{\edgesize}{\nodesize}\right).
    \end{align*}
    Therefore, $\edgesize_S = O(\nodesize \log\nodesize \log(2+\edgesize/\nodesize))$, and $\sparsespace = O(\nodesize + \edgesize_S) = O(\nodesize \log\nodesize \log(2+\edgesize/\nodesize))$.
    Combining both bounds, the total space complexity $\densespace + \sparsespace$ is $O(\nodesize \log\nodesize \log(2 + \edgesize / \nodesize))$ w.h.p.
\end{proof}

\begin{restatable}{lemma}{hybriddynamicupdate} \label{lem:hybrid_dynamic_update}
    The amortized update complexity of our hybrid dynamic connectivity algorithm is $O(\log^4 \nodesize)$.
\end{restatable}

\begin{proof}
    An edge insertion or deletion in our hybrid framework consists of up to three main components: (1) updating the corresponding base algorithm ($\densealg$ or $\sparsealg$) and auxiliary structures, (2) applying any resulting spanning forest changes from $\densealg$ to $\sparsealg$, and (3) handling potential vertex promotions or demotions.

    First, consider the cost of a standard edge update that does not trigger a promotion or demotion.
    If the edge is in the sparse region, we update $\sparsealg$. The \clusterforest algorithm handles updates in amortized $O(\log^2 \nodesize)$ time~\cite{wulff2013faster}.
    If the edge is in the dense region, we update $\densealg$. Our \sketchname-based algorithm (Section~\ref{sec:skinny_dynamic}) supports updates in worst-case $O(\log^4 \nodesize)$ time. Additionally, updating the \iblt's for the endpoints takes $O(1)$ time.
    In the dense case, $\densealg$ may return a list of changes to its spanning forest. Sketch-based dynamic connectivity algorithms produce at most $O(\log \nodesize)$ spanning forest changes per update. We must apply each of these changes to $\sparsealg$, which takes $O(\log \nodesize) \cdot O(\log^2 \nodesize) = O(\log^3 \nodesize)$ amortized time.
    Thus, the base cost of any single edge update, excluding vertex promotions/demotions, is bounded by $O(\log^4 \nodesize)$ amortized.

    Next, we analyze the cost of vertex promotions and demotions.
    When a vertex $v$ is promoted (i.e., its degree reaches above $\densethresh$), we iterate over its $\densethresh$ incident edges in $\sparsealg$. For any edge connecting to another heavy vertex, we insert it into $\densealg$ and delete it from $\sparsealg$.
    This requires at most $\densethresh$ operations in $\densealg$ and $\sparsealg$, along with tracking the subsequent spanning forest changes.
    Because each standard edge move costs bounded by $O(\log^4 \nodesize)$ amortized, the total cost of promoting a vertex is $O(\densethresh \log^4 \nodesize)$ amortized.
    Similarly, when a vertex $v$ is demoted (its degree falls down to $\densethresh / 2$), we recover its at most $\densethresh / 2$ incident edges from $\iblt_v$ in $O(\densethresh)$ time, delete them from $\densealg$, and insert them into $\sparsealg$. This process also involves at most $\densethresh/2$ edge moves, bounding the demotion cost by $O(\densethresh \log^4 \nodesize)$ amortized.

    To bound the amortized cost, we use a standard potential argument.
    A vertex is promoted only when its degree increases above $\densethresh$, and demoted only when its degree decreases to $\densethresh/2$.
    Therefore, between any two consecutive threshold crossings (a promotion followed by a demotion, or vice versa), the vertex must have undergone at least $\densethresh/2$ individual edge insertions or deletions incident to it.
    We charge the $O(\densethresh \log^4 \nodesize)$ cost of the promotion or demotion evenly across these $\Omega(\densethresh)$ updates.
    This adds an amortized cost of $O(\densethresh \log^4 \nodesize / \densethresh) = O(\log^4 \nodesize)$ to each update.

    Combining the base update cost and the amortized promotion and demotion cost, the overall amortized update complexity of our hybrid framework is $O(\log^4 \nodesize)$.
\end{proof}

\begin{restatable}{lemma}{hybriddynamicquery} \label{lem:hybrid_dynamic_query}
    The worst-case query complexity of our hybrid dynamic connectivity algorithm is $O(\log\nodesize / \log\log\nodesize)$.
\end{restatable}

\begin{proof}
    As described in our query algorithm (Section~\ref{sec:hybrid_framework}), any connectivity query $\connected(u,v)$ in our framework is directly answered by calling $\sparsealg.\connected(u,v)$.
    Because we use the \clusterforest algorithm~\cite{wulff2013faster} as $\sparsealg$, which supports worst-case $O(\log\nodesize / \log\log\nodesize)$ query times, the query complexity of our hybrid framework is exactly $O(\log\nodesize / \log\log\nodesize)$ in the worst case.
\end{proof}

%% file: AE_experiments.tex
\section{Additional Experiment Details}
\label{app:results}

\myparagraph{\revision{Correctness \& Number of Tiers}}
Table~\ref{tab:num_tiers_max_link_tier} indicates the number of cutset tiers that were actually required for example runs of \sysname and \cupcake. The results show that the number of tiers is always lower than the number we allocated ($\max\{15,\lceil\log_2\nodesize\rceil\}$), indicating that the Monte Carlo random algorithms never induced an error in our experiments.
Also \sysname can require significantly fewer tiers than \cupcake.

\revision{
 As described in Section~\ref{sec:system}, there are two possible failure conditions for \sysname; undetectable failures due to 64-bit hash collisions, and detectable failures due to an insufficient number of sketch tiers in \densesysname.
In order to validate the unlikelihood of either type of error, we computed the set of connected components every 100,000 updates on the Google+ and kron-13 graphs, across hundreds of runs. 
This meant computing the \emph{entire} set of connected components more than 100,000 times for Google+ and more than 160,000 times for kron-13 across 500 independent runs.  
With $\min\{15, \lceil \log_2 \nodesize \rceil\}$ tiers and using 64-bit hash checksums, we never observed \sysname reporting a single connected component incorrectly.
}

\begin{table}[ht]
    \centering
    \begin{tabular}{lccc}
        \toprule
        Graph & $\max\{15,\lceil\log_2\nodesize\rceil\}$ & \textsc{Hybrid} & CUP \\
         &  & \textsc{SCALE} & CaKE \\
        \midrule
        Google+         & 17 & 11 & 14 \\
        Friendster      & 26 & 14 & \multicolumn{1}{c}{---} \\
        Twitter         & 26 & 15 & \multicolumn{1}{c}{---} \\
        Orkut           & 22 & 11 & 18 \\
        ENWiki          & 23 & 12 & \multicolumn{1}{c}{---} \\
        Youtube         & 21 & 9  & 15 \\
        RoadUSA         & 25 & 0  & \multicolumn{1}{c}{---} \\
        RoadGER         & 24 & 0  & \multicolumn{1}{c}{---} \\
        SIFT-RS-50K     & 20 & 13 & 17 \\
        SIFT-KNN-500    & 20 & 14 & 18 \\
        MSSPACE-RS-10K  & 20 & 13 & 16 \\
        MSSPACE-KNN-500 & 20 & 15 & 18 \\
        kron-13         & 15 & 11 & 12 \\
        kron-15         & 15 & 12 & 14 \\
        kron-16         & 16 & 14 & 14 \\
        \bottomrule
    \end{tabular}
    \caption{\small Max number of tiers used. Entries indicated as --- ran out of memory (OOM). A $0$ means \sysname sketched nothing.}
    \label{tab:num_tiers_max_link_tier}
\end{table}

\myparagraph{\revision{Choice of Hybrid Threshold \& Sketching Primitive}}
\revision{Figures~\ref{fig:threshold_enwiki}--\ref{fig:threshold_sift_rs50k} show the space usage of \sysname on different graphs using various choices of $\densethresh$.
We include pure \clusterforest, \densesysname, and \cupcake for comparison.
We additionally visualize the space breakdown of the various components of \sysname in these plots, and show both \sysname using \sketchname as the sketching primitive, as well as \cameosketch instead.
Table~\ref{tab:hybridscale_percent_breakdowns} provides the raw data for the breakdown of the space usage in \sysname, and the percentage of edges and vertices that are sketched.
}

\revision{Based on these results, we chose to set $\densethresh = 25 \cdot \lceil \log_2 \nodesize \rceil$. Although the best $\densethresh$ value differs slightly for each graph, we found that this common value provides close to optimal performance across a wide variety of graphs.}
\revision{We also found that \sketchname was often needed to save any significant amount of space with our hybrid algorithm. 
Figure~\ref{fig:threshold_gplus} shows, for instance, that values of $c$ from $10$ to $25$ offer reasonable space savings with \sketchname, while no value of $c$ shows any space savings for the Google+ graph.}

\myparagraph{\revision{Mixed Workloads}}
\revision{Table~\ref{tab:hybridscale_interleaved} provides experimental results studying the performance of \sysname under mixed workloads that interleave updates and queries.
This allows us to evaluate whether frequent queries degrade update throughput, or vice versa.
The results show that \sysname experiences little degradation in performance in mixed workloads.
On the other hand, there is a known issue with \cupcake~\cite{deman2025fast} losing performance on streams with queries interleaved between updates.
This degradation occurs in \cupcake because it relies on buffering updates for parallel processing. Queries trigger buffer flushes, which drastically diminishes the effectiveness of this optimization. As described in Section~\ref{sec:dense_subsystem}, \sysname adopts an improved version of this buffering strategy, but avoids this problem by answering queries directly from \clusterforest.}

\begin{table}[ht]
    \centering
    \begin{tabular}{lccc}
    \toprule
    Dataset & 0 q/upt. & 0.1 q/upt. & 1 q/upt. \\
    \midrule
    Google+ & 0:38 & 0:50 & 0:56 \\
    Friendster & 12:04:39 & 11:43:35 & 15:34:22 \\
    Twitter & 6:32:51 & 7:17:52 & 10:07:10 \\
    Orkut & 13:40 & 14:54 & 20:20 \\
    ENWiki & 10:38 & 11:30 & 15:46 \\
    Youtube & 0:17 & 0:18 & 0:23 \\
    RoadUSA & 4:29 & 4:42 & 5:48 \\
    RoadGER & 2:47 & 2:53 & 3:36 \\
    SIFT-RS-50K & 29:20 & 32:09 & 36:02 \\
    SIFT-KNN-500 & 2:14:09 & 2:11:50 & 2:19:40 \\
    MSSPACE-RS-10K & 1:05:37 & 1:08:20 & 1:13:40 \\
    MSSPACE-KNN-500 & 1:52:14 & 2:00:01 & 2:04:10 \\
    kron-13 & 1:08 & 1:09 & 1:10 \\
    kron-15 & 27:47 & 28:20 & 28:28 \\
    kron-16 & 2:42:30 & 2:50:37 & 2:52:00 \\
    \bottomrule
    \end{tabular}
    \caption{\small\revision{Time to ingest graph stream with varying numbers of queries per update. The columns specify the number of queries interleaved between each update in expectation (generated by repeatedly sampling from a Poisson distribution).}}
    \label{tab:hybridscale_interleaved}
\end{table}

\begin{table*}[ht]
\centering
    \begin{tabular}{lcccccc}
    \toprule
    & \multicolumn{4}{c}{Peak memory usage (GB)} & \multicolumn{2}{c}{\% Sketched} \\
    Dataset & Total & CF & Sketching & Recovery & Vertices & Edges \\
    \midrule
    Google+ & 0.22 & 0.11 & 0.09 & 0.02 & 17.4\% & 67.5\% \\
    Friendster & 72.12 & 58.37 & 10.34 & 2.73 & 2.4\% & 13.8\% \\
    Twitter & 40.06 & 32.85 & 5.38 & 1.40 & 2.0\% & 26.1\% \\
    Orkut & 4.02 & 3.44 & 0.42 & 0.12 & 2.6\% & 9.5\% \\
    ENWiki & 3.98 & 3.58 & 0.27 & 0.08 & 1.2\% & 4.3\% \\
    Youtube & 0.48 & 0.46 & 0.01 & 0.00 & 0.1\% & 0.9\% \\
    RoadUSA & 10.17 & 9.93 & 0.00 & 0.00 & 0.0\% & 0.0\% \\
    RoadGER & 5.33 & 5.20 & 0.00 & 0.00 & 0.0\% & 0.0\% \\
    SIFT-RS-50K & 2.15 & 0.84 & 1.04 & 0.26 & 17.8\% & 91.6\% \\
    SIFT-KNN-500 & 8.62 & 1.45 & 5.69 & 1.45 & 100.0\% & 99.9\% \\
    MSSPACE-RS-10K & 1.14 & 0.50 & 0.52 & 0.12 & 8.0\% & 98.6\% \\
    MSSPACE-KNN-500 & 9.67 & 2.53 & 5.66 & 1.45 & 100.0\% & 99.9\% \\
    kron-13 & 0.06 & 0.01 & 0.04 & 0.01 & 99.7\% & 99.9\% \\
    kron-15 & 0.26 & 0.05 & 0.17 & 0.04 & 99.8\% & 99.9\% \\
    kron-16 & 0.56 & 0.09 & 0.38 & 0.08 & 99.9\% & 99.9\% \\
        \bottomrule
    \end{tabular}
    \caption{\revision{Breakdown of space usage in final experimental runs with $25\lceil \log_2 V \rceil$ tiers. Note that some small differences from the threshold sweep experiments are expected, e.g. from stream randomness. The total is higher than the sum of the three constituents parts due to the overhead of the hybrid driver itself. }}
    \label{tab:hybridscale_percent_breakdowns}
\end{table*}

\myparagraph{\revision{Heat Map Normalization}}
\revision{
To normalize the heatmaps we first, for each column, normalize the space to the \emph{most} space-efficient system in each column, and normalize the amount of \emph{seconds per operation} to the \emph{fastest} system for updates and queries. These values are then further nonlinearly transformed by the formula $f(x) = \tanh(\log_2(x)/2)$.  
This scaling accommodates the wide range in the results, ensuring meaningful color differentiation among closely performing systems without allowing multi-order-of-magnitude outliers to dominate the heatmap.
}

\clearpage

\begin{figure*}[ht]
    \centering
    \includegraphics[width=\linewidth]{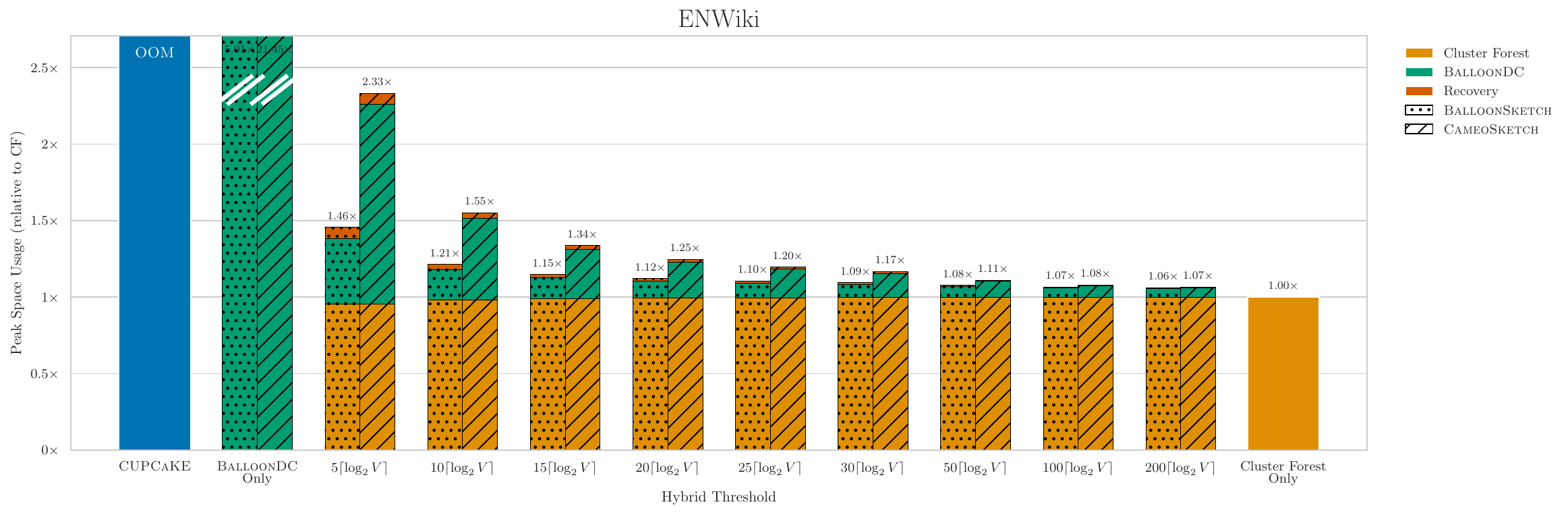}
    \vspace{-1em}
    \caption{\revision{Space usage of \sysname at various $\delta$ on ENWiki with \cupcake, \densesysname, and \clusterforest.}}
    \label{fig:threshold_enwiki}
\end{figure*}

\begin{figure*}[ht]
    \centering
    \includegraphics[width=\linewidth]{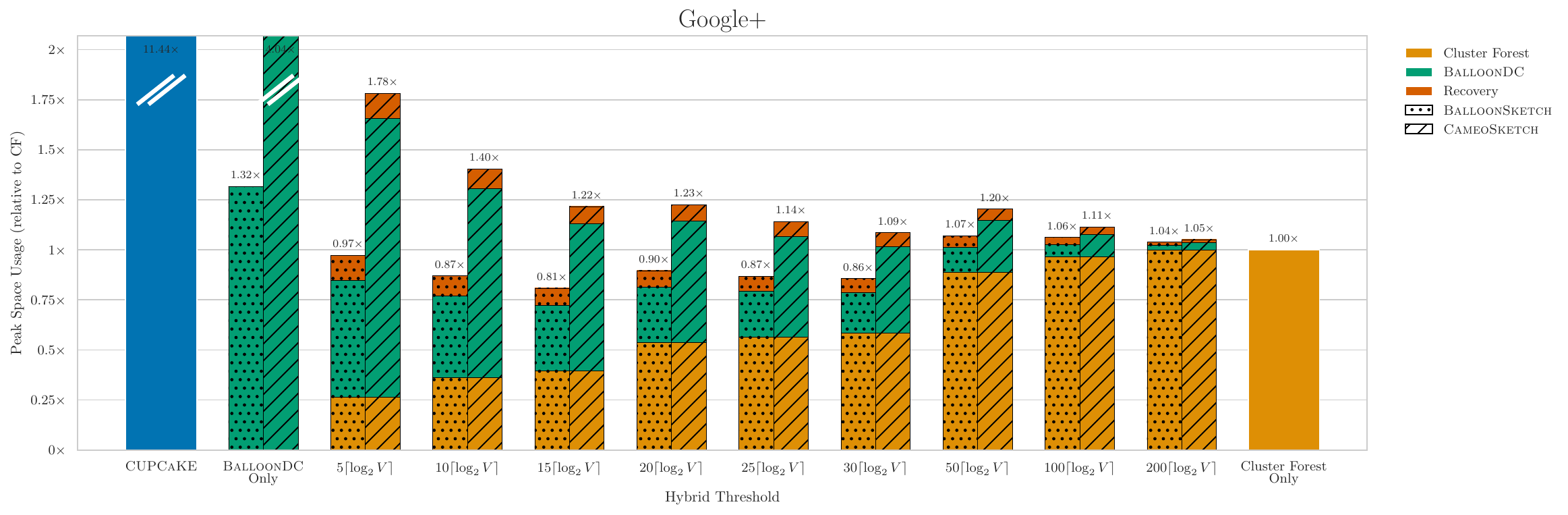}
    \vspace{-1em}
    \caption{\revision{Space usage of \sysname at various $\delta$ on Google+ with \cupcake, \densesysname, and \clusterforest.}}
    \label{fig:threshold_gplus}
\end{figure*}

\begin{figure*}[ht]
    \centering
    \includegraphics[width=\linewidth]{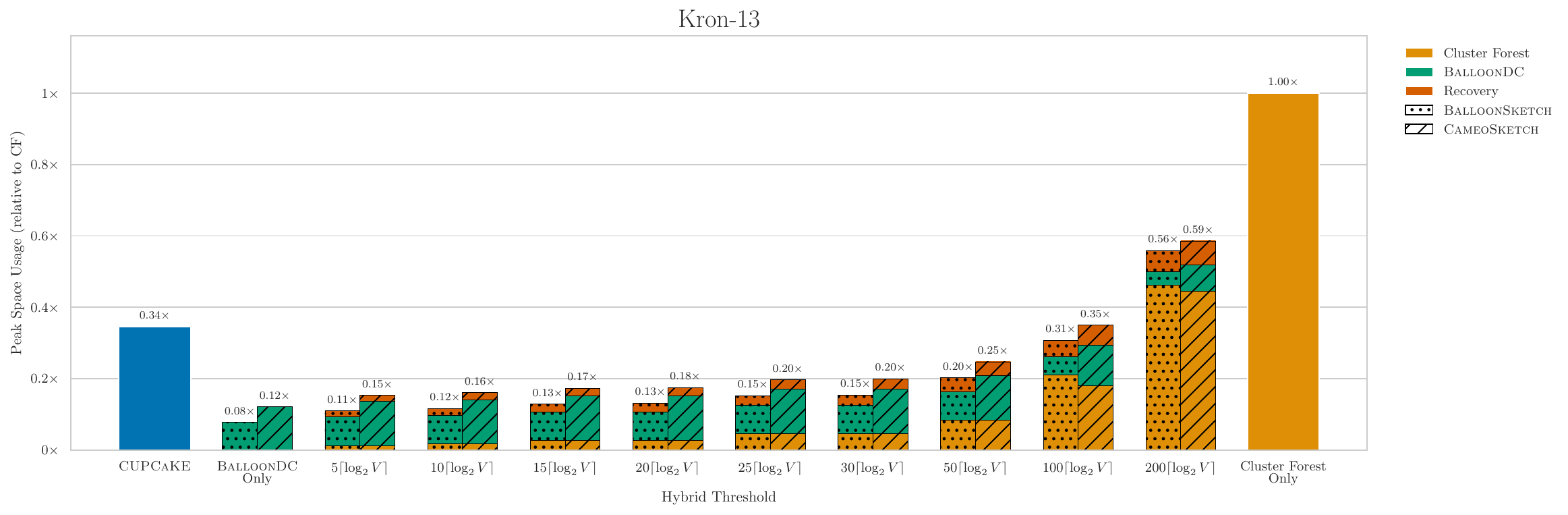}
    \vspace{-1em}
    \caption{\revision{Space usage of \sysname at various $\delta$ on Kron-13 with \cupcake, \densesysname, and \clusterforest.}}
    \label{fig:threshold_kron13}
\end{figure*}

\begin{figure*}[ht]
    \centering
    \includegraphics[width=\linewidth]{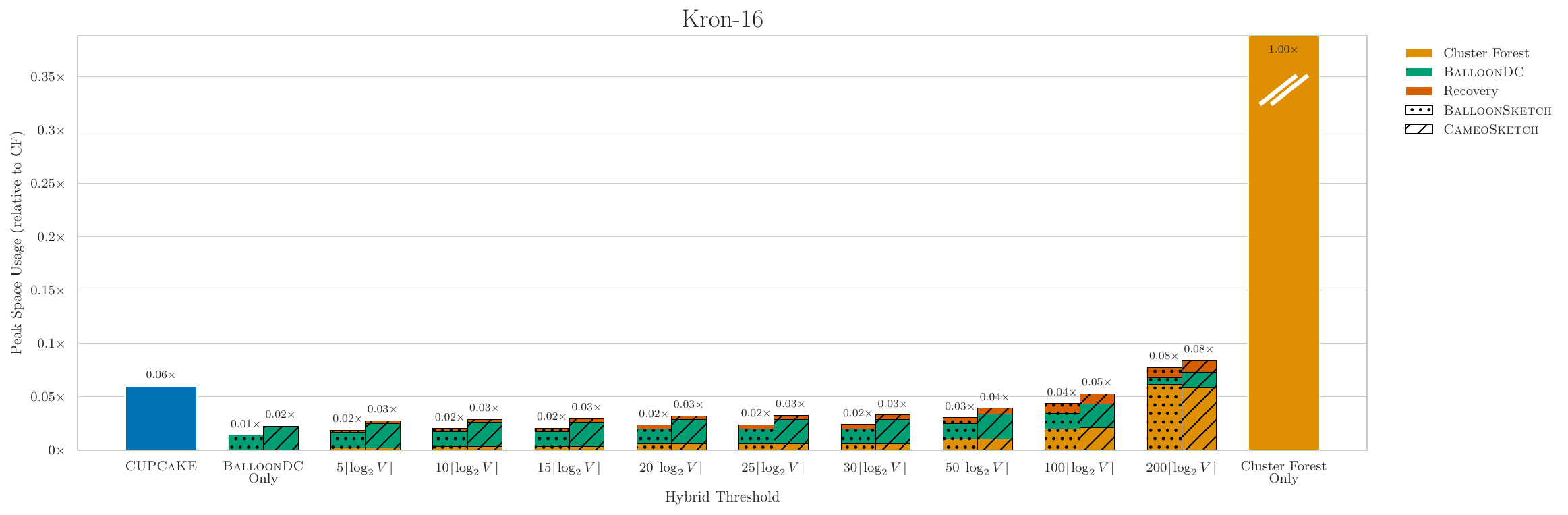}
    \vspace{-1em}
    \caption{\revision{Space usage of \sysname at various $\delta$ on Kron-16 with \cupcake, \densesysname, and \clusterforest.}}
    \label{fig:threshold_kron16}
\end{figure*}

\begin{figure*}[ht]
    \centering
    \includegraphics[width=\linewidth]{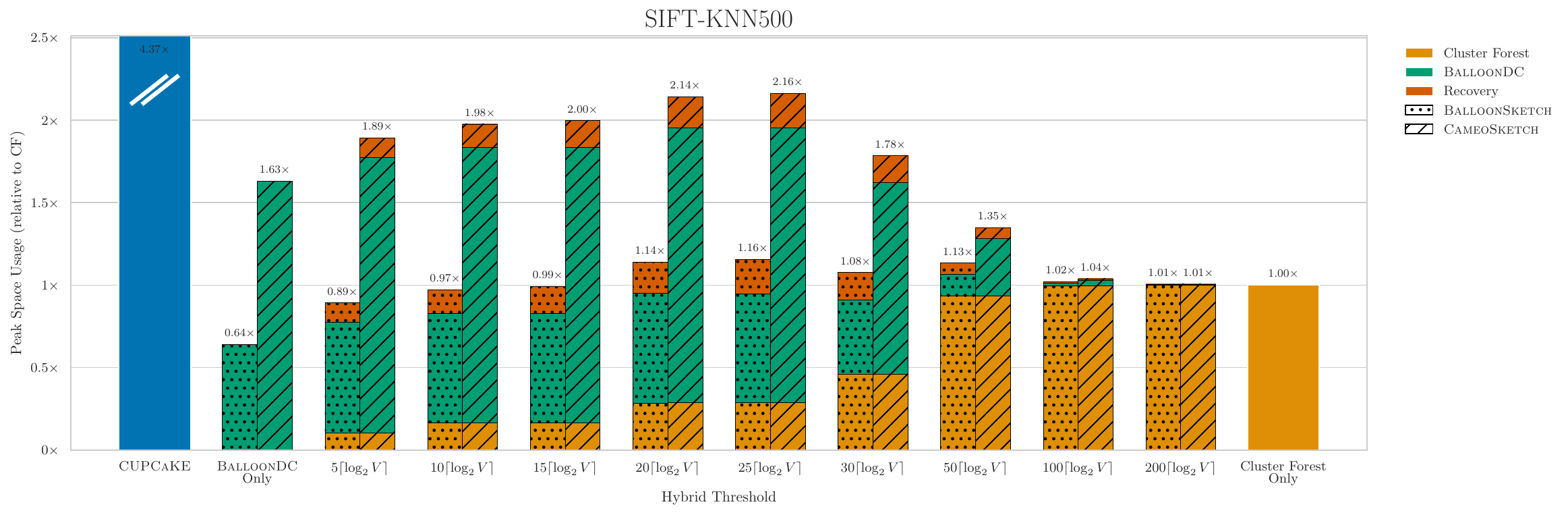}
    \vspace{-1em}
    \caption{\revision{Space usage of \sysname at various $\delta$ on SIFT-KNN-500 with \cupcake, \densesysname, and \clusterforest.}}
    \label{fig:threshold_sift_knn500}
\end{figure*}

\begin{figure*}[ht]
    \centering
    \includegraphics[width=\linewidth]{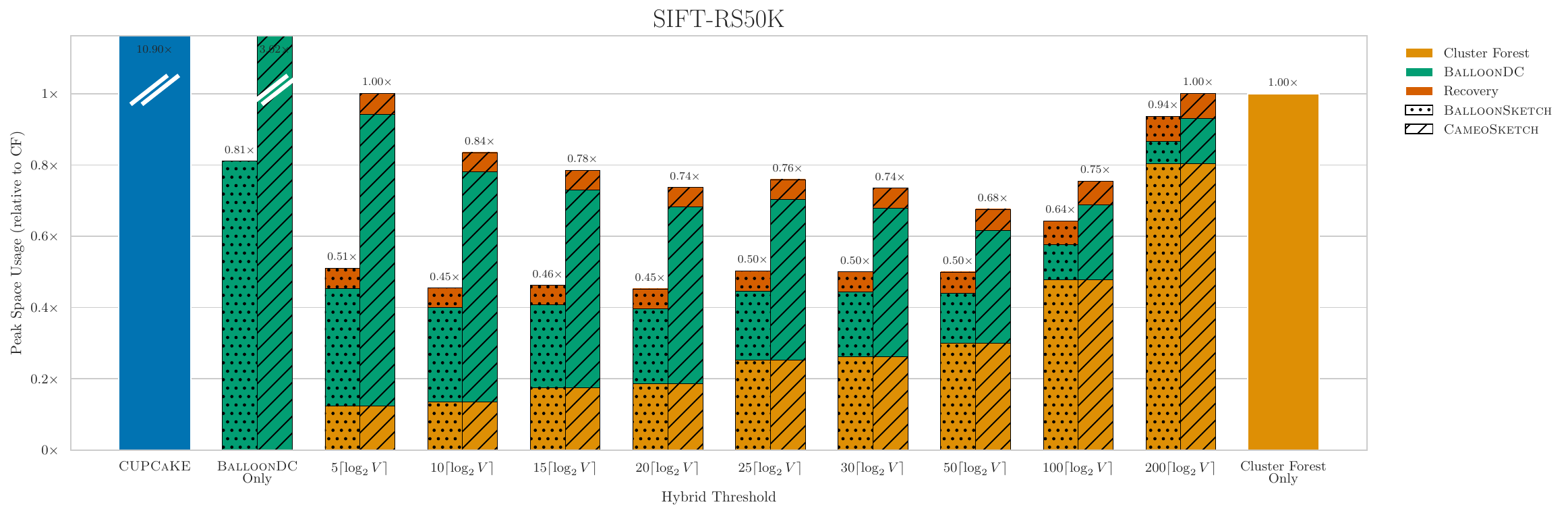}
    \vspace{-1em}
    \caption{\revision{Space usage of \sysname at various $\delta$ on SIFT-RS-50K with \cupcake, \densesysname, and \clusterforest.}}
    \label{fig:threshold_sift_rs50k}
\end{figure*}

\clearpage

%% file: main.bbl

\begin{thebibliography}{65}


\ifx \showCODEN    \undefined \def \showCODEN     #1{\unskip}     \fi
\ifx \showDOI      \undefined \def \showDOI       #1{#1}\fi
\ifx \showISBNx    \undefined \def \showISBNx     #1{\unskip}     \fi
\ifx \showISBNxiii \undefined \def \showISBNxiii  #1{\unskip}     \fi
\ifx \showISSN     \undefined \def \showISSN      #1{\unskip}     \fi
\ifx \showLCCN     \undefined \def \showLCCN      #1{\unskip}     \fi
\ifx \shownote     \undefined \def \shownote      #1{#1}          \fi
\ifx \showarticletitle \undefined \def \showarticletitle #1{#1}   \fi
\ifx \showURL      \undefined \def \showURL       {\relax}        \fi
\providecommand\bibfield[2]{#2}
\providecommand\bibinfo[2]{#2}
\providecommand\natexlab[1]{#1}
\providecommand\showeprint[2][]{arXiv:#2}

\bibitem[Acar et~al\mbox{.}(2019)]%
        {acar2019parallel}
\bibfield{author}{\bibinfo{person}{Umut~A. Acar}, \bibinfo{person}{Daniel Anderson}, \bibinfo{person}{Guy~E. Blelloch}, {and} \bibinfo{person}{Laxman Dhulipala}.} \bibinfo{year}{2019}\natexlab{}.
\newblock \showarticletitle{Parallel Batch-Dynamic Graph Connectivity}. In \bibinfo{booktitle}{\emph{Proceedings of the 31st ACM Symposium on Parallelism in Algorithms and Architectures (SPAA)}}. \bibinfo{pages}{381--392}.
\newblock
\urldef\tempurl%
\url{https://doi.org/10.1145/3323165.3323196}
\showDOI{\tempurl}


\bibitem[Acar et~al\mbox{.}(2020)]%
        {acar2020parallel}
\bibfield{author}{\bibinfo{person}{Umut~A. Acar}, \bibinfo{person}{Daniel Anderson}, \bibinfo{person}{Guy~E. Blelloch}, \bibinfo{person}{Laxman Dhulipala}, {and} \bibinfo{person}{Sam Westrick}.} \bibinfo{year}{2020}\natexlab{}.
\newblock \showarticletitle{Parallel Batch-Dynamic Trees via Change Propagation}. In \bibinfo{booktitle}{\emph{28th Annual European Symposium on Algorithms (ESA)}}. \bibinfo{pages}{2:1--2:23}.
\newblock
\urldef\tempurl%
\url{https://doi.org/10.4230/LIPIcs.ESA.2020.2}
\showDOI{\tempurl}


\bibitem[Ahn et~al\mbox{.}(2012a)]%
        {Ahn2012}
\bibfield{author}{\bibinfo{person}{Kook~Jin Ahn}, \bibinfo{person}{Sudipto Guha}, {and} \bibinfo{person}{Andrew McGregor}.} \bibinfo{year}{2012}\natexlab{a}.
\newblock \showarticletitle{Analyzing graph structure via linear measurements}. In \bibinfo{booktitle}{\emph{Proceedings of the twenty-third annual ACM-SIAM symposium on Discrete Algorithms}}. SIAM, \bibinfo{pages}{459--467}.
\newblock


\bibitem[Ahn et~al\mbox{.}(2012b)]%
        {AhnGM12b}
\bibfield{author}{\bibinfo{person}{Kook~Jin Ahn}, \bibinfo{person}{Sudipto Guha}, {and} \bibinfo{person}{Andrew McGregor}.} \bibinfo{year}{2012}\natexlab{b}.
\newblock \showarticletitle{Graph sketches: sparsification, spanners, and subgraphs}. In \bibinfo{booktitle}{\emph{{PODS}}}. \bibinfo{publisher}{{ACM}}, \bibinfo{pages}{5--14}.
\newblock


\bibitem[Alon et~al\mbox{.}(1999)]%
        {alonFrequencyMoments1999}
\bibfield{author}{\bibinfo{person}{Noga Alon}, \bibinfo{person}{Yossi Matias}, {and} \bibinfo{person}{Mario Szegedy}.} \bibinfo{year}{1999}\natexlab{}.
\newblock \showarticletitle{The Space Complexity of Approximating the Frequency Moments}.
\newblock \bibinfo{journal}{\emph{J. Comput. System Sci.}} \bibinfo{volume}{58}, \bibinfo{number}{1} (\bibinfo{year}{1999}), \bibinfo{pages}{137--147}.
\newblock
\urldef\tempurl%
\url{https://doi.org/10.1006/jcss.1997.1545}
\showDOI{\tempurl}


\bibitem[Alvarez-Hamelin et~al\mbox{.}(2005)]%
        {alvarezhamelin2005kcore}
\bibfield{author}{\bibinfo{person}{J.~Ignacio Alvarez-Hamelin}, \bibinfo{person}{Luca Dall'Asta}, \bibinfo{person}{Alain Barrat}, {and} \bibinfo{person}{Alessandro Vespignani}.} \bibinfo{year}{2005}\natexlab{}.
\newblock \showarticletitle{Large Scale Networks Fingerprinting and Visualization Using the $k$-Core Decomposition}. In \bibinfo{booktitle}{\emph{Advances in Neural Information Processing Systems (NeurIPS)}}. \bibinfo{pages}{41--50}.
\newblock


\bibitem[Ang et~al\mbox{.}(2010)]%
        {krongraphs}
\bibfield{author}{\bibinfo{person}{James~A. Ang}, \bibinfo{person}{Brian~W. Barrett}, \bibinfo{person}{Kyle~B. Wheeler}, {and} \bibinfo{person}{Richard~C. Murphy}.} \bibinfo{year}{2010}\natexlab{}.
\newblock \showarticletitle{Introducing the {G}raph 500}. In \bibinfo{booktitle}{\emph{Cray User Group (CUG) Proceedings}}.
\newblock
\urldef\tempurl%
\url{https://cug.org/5-publications/proceedings_attendee_lists/CUG10CD/pages/1-program/final_program/CUG10_Proceedings/pages/authors/11-15Wednesday/14C-Murphy-paper.pdf}
\showURL{%
\tempurl}


\bibitem[Batagelj and Zaver{\v s}nik(2003)]%
        {batagelj2003om}
\bibfield{author}{\bibinfo{person}{Vladimir Batagelj} {and} \bibinfo{person}{Matja{\v z} Zaver{\v s}nik}.} \bibinfo{year}{2003}\natexlab{}.
\newblock \showarticletitle{An $O(m)$ Algorithm for Cores Decomposition of Networks}.
\newblock \bibinfo{journal}{\emph{arXiv preprint cs/0310049}} (\bibinfo{year}{2003}).
\newblock


\bibitem[Belazzougui et~al\mbox{.}(2024)]%
        {belazzougui2024iblt}
\bibfield{author}{\bibinfo{person}{Djamal Belazzougui}, \bibinfo{person}{Gregory Kucherov}, {and} \bibinfo{person}{Stefan Walzer}.} \bibinfo{year}{2024}\natexlab{}.
\newblock \showarticletitle{{Better Space-Time-Robustness Trade-Offs for Set Reconciliation}}. In \bibinfo{booktitle}{\emph{51st International Colloquium on Automata, Languages, and Programming (ICALP 2024)}} \emph{(\bibinfo{series}{Leibniz International Proceedings in Informatics (LIPIcs)}, Vol.~\bibinfo{volume}{297})}, \bibfield{editor}{\bibinfo{person}{Karl Bringmann}, \bibinfo{person}{Martin Grohe}, \bibinfo{person}{Gabriele Puppis}, {and} \bibinfo{person}{Ola Svensson}} (Eds.). \bibinfo{publisher}{Schloss Dagstuhl -- Leibniz-Zentrum f{\"u}r Informatik}, \bibinfo{address}{Dagstuhl, Germany}, \bibinfo{pages}{20:1--20:19}.
\newblock
\showISBNx{978-3-95977-322-5}
\showISSN{1868-8969}
\urldef\tempurl%
\url{https://doi.org/10.4230/LIPIcs.ICALP.2024.20}
\showDOI{\tempurl}


\bibitem[Boldi and Vigna(2004)]%
        {boldi2004webgraph}
\bibfield{author}{\bibinfo{person}{Paolo Boldi} {and} \bibinfo{person}{Sebastiano Vigna}.} \bibinfo{year}{2004}\natexlab{}.
\newblock \showarticletitle{The {WebGraph} Framework {I}: Compression Techniques}. In \bibinfo{booktitle}{\emph{Proceedings of the 13th International Conference on World Wide Web (WWW)}}. \bibinfo{publisher}{ACM}, \bibinfo{pages}{595--602}.
\newblock
\urldef\tempurl%
\url{https://doi.org/10.1145/988672.988752}
\showDOI{\tempurl}


\bibitem[Borgatti and Everett(2000)]%
        {borgatti2000models}
\bibfield{author}{\bibinfo{person}{Stephen~P. Borgatti} {and} \bibinfo{person}{Martin~G. Everett}.} \bibinfo{year}{2000}\natexlab{}.
\newblock \showarticletitle{Models of Core/Periphery Structures}.
\newblock \bibinfo{journal}{\emph{Social Networks}} \bibinfo{volume}{21}, \bibinfo{number}{4} (\bibinfo{year}{2000}), \bibinfo{pages}{375--395}.
\newblock
\urldef\tempurl%
\url{https://doi.org/10.1016/S0378-8733(99)00019-2}
\showDOI{\tempurl}


\bibitem[Brown(2011)]%
        {brownGeometricConcentration}
\bibfield{author}{\bibinfo{person}{Daniel~G. Brown}.} \bibinfo{year}{2011}\natexlab{}.
\newblock \showarticletitle{How {{I}} Wasted Too Long Finding a Concentration Inequality for Sums of Geometric Variables}.
\newblock  (\bibinfo{year}{2011}).
\newblock


\bibitem[Carey et~al\mbox{.}(2022)]%
        {stars2022}
\bibfield{author}{\bibinfo{person}{CJ Carey}, \bibinfo{person}{Jonathan Halcrow}, \bibinfo{person}{Rajesh Jayaram}, \bibinfo{person}{Vahab Mirrokni}, \bibinfo{person}{Warren Schudy}, {and} \bibinfo{person}{Peilin Zhong}.} \bibinfo{year}{2022}\natexlab{}.
\newblock \showarticletitle{Stars: Tera-Scale Graph Building for Clustering and Learning}. In \bibinfo{booktitle}{\emph{Advances in Neural Information Processing Systems (NeurIPS)}}.
\newblock


\bibitem[Carmi et~al\mbox{.}(2007)]%
        {carmi2007model}
\bibfield{author}{\bibinfo{person}{Shai Carmi}, \bibinfo{person}{Shlomo Havlin}, \bibinfo{person}{Scott Kirkpatrick}, \bibinfo{person}{Yuval Shavitt}, {and} \bibinfo{person}{Eran Shir}.} \bibinfo{year}{2007}\natexlab{}.
\newblock \showarticletitle{A Model of {Internet} Topology Using $k$-Shell Decomposition}.
\newblock \bibinfo{journal}{\emph{Proceedings of the National Academy of Sciences (PNAS)}} \bibinfo{volume}{104}, \bibinfo{number}{27} (\bibinfo{year}{2007}), \bibinfo{pages}{11150--11154}.
\newblock
\urldef\tempurl%
\url{https://doi.org/10.1073/pnas.0701175104}
\showDOI{\tempurl}


\bibitem[Charikar et~al\mbox{.}({[n.\,d.]})]%
        {countsketch}
\bibfield{author}{\bibinfo{person}{Moses Charikar}, \bibinfo{person}{Kevin Chen}, {and} \bibinfo{person}{Martin {Farach-Colton}}.} \bibinfo{year}{[n.\,d.]}\natexlab{}.
\newblock \showarticletitle{Finding {{Frequent Items}} in {{Data Streams}}}.
\newblock  (\bibinfo{year}{[n.\,d.]}).
\newblock


\bibitem[Chuzhoy et~al\mbox{.}(2020)]%
        {chuzhoy2020deterministic}
\bibfield{author}{\bibinfo{person}{Julia Chuzhoy}, \bibinfo{person}{Yu Gao}, \bibinfo{person}{Jason Li}, \bibinfo{person}{Danupon Nanongkai}, \bibinfo{person}{Richard Peng}, {and} \bibinfo{person}{Thatchaphol Saranurak}.} \bibinfo{year}{2020}\natexlab{}.
\newblock \showarticletitle{A Deterministic Algorithm for Balanced Cut with Applications to Dynamic Connectivity, Flows, and Beyond}. In \bibinfo{booktitle}{\emph{61st Annual IEEE Symposium on Foundations of Computer Science (FOCS)}}. \bibinfo{pages}{1158--1167}.
\newblock
\urldef\tempurl%
\url{https://doi.org/10.1109/FOCS46700.2020.00111}
\showDOI{\tempurl}


\bibitem[Collet(2016)]%
        {xxhash}
\bibfield{author}{\bibinfo{person}{Yann Collet}.} \bibinfo{year}{2016}\natexlab{}.
\newblock \showarticletitle{xxHash-Extremely fast non-cryptographic hash algorithm}.
\newblock \bibinfo{journal}{\emph{URL https://github. com/Cyan4973/xxHash}} (\bibinfo{year}{2016}).
\newblock


\bibitem[Cormode and Firmani(2014)]%
        {cormode2014unifying}
\bibfield{author}{\bibinfo{person}{Graham Cormode} {and} \bibinfo{person}{Donatella Firmani}.} \bibinfo{year}{2014}\natexlab{}.
\newblock \showarticletitle{A unifying framework for l0-sampling algorithms}.
\newblock \bibinfo{journal}{\emph{Distrib. Parallel Databases}} \bibinfo{volume}{32}, \bibinfo{number}{3} (\bibinfo{date}{sep} \bibinfo{year}{2014}), \bibinfo{pages}{315–335}.
\newblock
\showISSN{0926-8782}
\urldef\tempurl%
\url{https://doi.org/10.1007/s10619-013-7131-9}
\showDOI{\tempurl}


\bibitem[Cormode and Muthukrishnan(2005a)]%
        {cordmodCountMinSketch2005}
\bibfield{author}{\bibinfo{person}{Graham Cormode} {and} \bibinfo{person}{Shanmugavelayutham Muthukrishnan}.} \bibinfo{year}{2005}\natexlab{a}.
\newblock \showarticletitle{An Improved Data Stream Summary: The Count-Min Sketch and Its Applications}.
\newblock \bibinfo{journal}{\emph{Journal of Algorithms}} \bibinfo{volume}{55}, \bibinfo{number}{1} (\bibinfo{year}{2005}), \bibinfo{pages}{58--75}.
\newblock
\urldef\tempurl%
\url{https://doi.org/10.1016/j.jalgor.2003.12.001}
\showDOI{\tempurl}


\bibitem[Cormode and Muthukrishnan(2005b)]%
        {cormode2005improved}
\bibfield{author}{\bibinfo{person}{Graham Cormode} {and} \bibinfo{person}{S. Muthukrishnan}.} \bibinfo{year}{2005}\natexlab{b}.
\newblock \showarticletitle{An Improved Data Stream Summary: The {Count-Min} Sketch and Its Applications}.
\newblock \bibinfo{journal}{\emph{Journal of Algorithms}} \bibinfo{volume}{55}, \bibinfo{number}{1} (\bibinfo{year}{2005}), \bibinfo{pages}{58--75}.
\newblock


\bibitem[Davis and Hu(2011)]%
        {suitesparse}
\bibfield{author}{\bibinfo{person}{Timothy~A. Davis} {and} \bibinfo{person}{Yifan Hu}.} \bibinfo{year}{2011}\natexlab{}.
\newblock \showarticletitle{The university of Florida sparse matrix collection}.
\newblock \bibinfo{journal}{\emph{ACM Trans. Math. Softw.}} \bibinfo{volume}{38}, \bibinfo{number}{1}, Article \bibinfo{articleno}{1} (\bibinfo{date}{dec} \bibinfo{year}{2011}), \bibinfo{numpages}{25}~pages.
\newblock
\showISSN{0098-3500}
\urldef\tempurl%
\url{https://doi.org/10.1145/2049662.2049663}
\showDOI{\tempurl}


\bibitem[De~Man et~al\mbox{.}(2024)]%
        {deman2024towards}
\bibfield{author}{\bibinfo{person}{Quinten De~Man}, \bibinfo{person}{Laxman Dhulipala}, \bibinfo{person}{Adam Karczmarz}, \bibinfo{person}{Jakub \L{}\k{a}cki}, \bibinfo{person}{Julian Shun}, {and} \bibinfo{person}{Zhongqi Wang}.} \bibinfo{year}{2024}\natexlab{}.
\newblock \showarticletitle{Towards Scalable and Practical Batch-Dynamic Connectivity}.
\newblock \bibinfo{journal}{\emph{Proc. VLDB Endow.}} \bibinfo{volume}{18}, \bibinfo{number}{3} (\bibinfo{date}{Nov.} \bibinfo{year}{2024}), \bibinfo{pages}{889–901}.
\newblock
\showISSN{2150-8097}
\urldef\tempurl%
\url{https://doi.org/10.14778/3712221.3712250}
\showDOI{\tempurl}


\bibitem[De~Man et~al\mbox{.}(2025)]%
        {deman2025fast}
\bibfield{author}{\bibinfo{person}{Quinten De~Man}, \bibinfo{person}{Qamber Jafri}, \bibinfo{person}{Daniel Delayo}, \bibinfo{person}{Evan~T. West}, \bibinfo{person}{Michael~A. Bender}, {and} \bibinfo{person}{David Tench}.} \bibinfo{year}{2025}\natexlab{}.
\newblock \bibinfo{title}{Fast and Compact Sketch-Based Dynamic Connectivity}.
\newblock
\newblock
\showeprint[arxiv]{2509.14433}~[cs.DS]
\urldef\tempurl%
\url{https://arxiv.org/abs/2509.14433}
\showURL{%
\tempurl}


\bibitem[De~Man et~al\mbox{.}(2026)]%
        {deman2026ufo}
\bibfield{author}{\bibinfo{person}{Quinten De~Man}, \bibinfo{person}{Atharva Sharma}, \bibinfo{person}{Kishen~N Gowda}, {and} \bibinfo{person}{Laxman Dhulipala}.} \bibinfo{year}{2026}\natexlab{}.
\newblock \showarticletitle{UFO Trees: Practical and Provably-Efficient Parallel Batch-Dynamic Trees}. In \bibinfo{booktitle}{\emph{Proceedings of the 31st ACM SIGPLAN Annual Symposium on Principles and Practice of Parallel Programming}} (Sydney, NSW, Australia) \emph{(\bibinfo{series}{PPoPP '26})}. \bibinfo{publisher}{Association for Computing Machinery}, \bibinfo{address}{New York, NY, USA}, \bibinfo{pages}{109–122}.
\newblock
\showISBNx{9798400723100}
\urldef\tempurl%
\url{https://doi.org/10.1145/3774934.3786431}
\showDOI{\tempurl}


\bibitem[Dong et~al\mbox{.}(2011)]%
        {nndescent2011}
\bibfield{author}{\bibinfo{person}{Wei Dong}, \bibinfo{person}{Moses Charikar}, {and} \bibinfo{person}{Kai Li}.} \bibinfo{year}{2011}\natexlab{}.
\newblock \showarticletitle{Efficient k-nearest neighbor graph construction for generic similarity measures}. In \bibinfo{booktitle}{\emph{Proceedings of the 20th International Conference on World Wide Web (WWW)}}.
\newblock


\bibitem[Efraimidis and Spirakis(2006)]%
        {efraimidisWeightedSampling2006}
\bibfield{author}{\bibinfo{person}{Pavlos~S. Efraimidis} {and} \bibinfo{person}{Paul~G. Spirakis}.} \bibinfo{year}{2006}\natexlab{}.
\newblock \showarticletitle{Weighted Random Sampling with a Reservoir}.
\newblock \bibinfo{journal}{\emph{Inform. Process. Lett.}} \bibinfo{volume}{97}, \bibinfo{number}{5} (\bibinfo{year}{2006}), \bibinfo{pages}{181--185}.
\newblock
\urldef\tempurl%
\url{https://doi.org/10.1016/j.ipl.2005.11.003}
\showDOI{\tempurl}


\bibitem[Eppstein et~al\mbox{.}(1997)]%
        {eppstein1997sparsification}
\bibfield{author}{\bibinfo{person}{David Eppstein}, \bibinfo{person}{Zvi Galil}, \bibinfo{person}{Giuseppe~F. Italiano}, {and} \bibinfo{person}{Amnon Nissenzweig}.} \bibinfo{year}{1997}\natexlab{}.
\newblock \showarticletitle{Sparsification---A Technique for Speeding Up Dynamic Graph Algorithms}.
\newblock \bibinfo{journal}{\emph{J. ACM}} \bibinfo{volume}{44}, \bibinfo{number}{5} (\bibinfo{year}{1997}), \bibinfo{pages}{669--696}.
\newblock
\urldef\tempurl%
\url{https://doi.org/10.1145/265910.265914}
\showDOI{\tempurl}


\bibitem[Faloutsos et~al\mbox{.}(1999)]%
        {faloutsos1999power}
\bibfield{author}{\bibinfo{person}{Michalis Faloutsos}, \bibinfo{person}{Petros Faloutsos}, {and} \bibinfo{person}{Christos Faloutsos}.} \bibinfo{year}{1999}\natexlab{}.
\newblock \showarticletitle{On Power-Law Relationships of the Internet Topology}.
\newblock \bibinfo{journal}{\emph{ACM SIGCOMM Computer Communication Review}} \bibinfo{volume}{29}, \bibinfo{number}{4} (\bibinfo{year}{1999}), \bibinfo{pages}{251--262}.
\newblock
\urldef\tempurl%
\url{https://doi.org/10.1145/316194.316229}
\showDOI{\tempurl}


\bibitem[Fan et~al\mbox{.}(2014)]%
        {fan2014cuckoo}
\bibfield{author}{\bibinfo{person}{Bin Fan}, \bibinfo{person}{David~G. Andersen}, \bibinfo{person}{Michael Kaminsky}, {and} \bibinfo{person}{Michael~D. Mitzenmacher}.} \bibinfo{year}{2014}\natexlab{}.
\newblock \showarticletitle{Cuckoo Filter: Practically Better Than {Bloom}}. In \bibinfo{booktitle}{\emph{Proceedings of the 10th ACM International on Conference on Emerging Networking Experiments and Technologies (CoNEXT '14)}}. \bibinfo{pages}{75--88}.
\newblock


\bibitem[Feigenbaum et~al\mbox{.}(2005)]%
        {semistreaming1}
\bibfield{author}{\bibinfo{person}{Joan Feigenbaum}, \bibinfo{person}{Sampath Kannan}, \bibinfo{person}{Andrew McGregor}, \bibinfo{person}{Siddharth Suri}, {and} \bibinfo{person}{Jian Zhang}.} \bibinfo{year}{2005}\natexlab{}.
\newblock \showarticletitle{On Graph Problems in a Semi-Streaming Model}.
\newblock \bibinfo{journal}{\emph{Theor. Comput. Sci.}} \bibinfo{volume}{348}, \bibinfo{number}{2} (\bibinfo{date}{Dec.} \bibinfo{year}{2005}), \bibinfo{pages}{207--216}.
\newblock
\showISSN{0304-3975}
\urldef\tempurl%
\url{https://doi.org/10.1016/j.tcs.2005.09.013}
\showDOI{\tempurl}


\bibitem[Frederickson(1985)]%
        {frederickson1985data}
\bibfield{author}{\bibinfo{person}{Greg~N. Frederickson}.} \bibinfo{year}{1985}\natexlab{}.
\newblock \showarticletitle{Data Structures for On-Line Updating of Minimum Spanning Trees, with Applications}.
\newblock \bibinfo{journal}{\emph{SIAM J. Comput.}} \bibinfo{volume}{14}, \bibinfo{number}{4} (\bibinfo{year}{1985}), \bibinfo{pages}{781--798}.
\newblock
\urldef\tempurl%
\url{https://doi.org/10.1137/0214055}
\showDOI{\tempurl}


\bibitem[Gibb et~al\mbox{.}(2015)]%
        {gibb2015dynamic}
\bibfield{author}{\bibinfo{person}{David Gibb}, \bibinfo{person}{Bruce Kapron}, \bibinfo{person}{Valerie King}, {and} \bibinfo{person}{Nolan Thorn}.} \bibinfo{year}{2015}\natexlab{}.
\newblock \bibinfo{title}{Dynamic graph connectivity with improved worst case update time and sublinear space}.
\newblock
\newblock
\showeprint[arxiv]{1509.06464}~[cs.DS]


\bibitem[Goodrich and Mitzenmacher(2011)]%
        {goodrich2011iblt}
\bibfield{author}{\bibinfo{person}{Michael~T Goodrich} {and} \bibinfo{person}{Michael Mitzenmacher}.} \bibinfo{year}{2011}\natexlab{}.
\newblock \showarticletitle{Invertible bloom lookup tables}. In \bibinfo{booktitle}{\emph{2011 49th Annual Allerton Conference on Communication, Control, and Computing (Allerton)}}. IEEE, \bibinfo{pages}{792--799}.
\newblock


\bibitem[Henzinger and King(1999)]%
        {henzinger1999randomized}
\bibfield{author}{\bibinfo{person}{Monika~R. Henzinger} {and} \bibinfo{person}{Valerie King}.} \bibinfo{year}{1999}\natexlab{}.
\newblock \showarticletitle{Randomized Fully Dynamic Graph Algorithms with Polylogarithmic Time per Operation}.
\newblock \bibinfo{journal}{\emph{J. ACM}} \bibinfo{volume}{46}, \bibinfo{number}{4} (\bibinfo{year}{1999}), \bibinfo{pages}{502--516}.
\newblock
\urldef\tempurl%
\url{https://doi.org/10.1145/320211.320215}
\showDOI{\tempurl}


\bibitem[Henzinger and Thorup(1997)]%
        {henzinger1997sampling}
\bibfield{author}{\bibinfo{person}{Monika~R. Henzinger} {and} \bibinfo{person}{Mikkel Thorup}.} \bibinfo{year}{1997}\natexlab{}.
\newblock \showarticletitle{Sampling to Provide or to Bound: With Applications to Fully Dynamic Graph Algorithms}.
\newblock \bibinfo{journal}{\emph{Random Structures \& Algorithms}} \bibinfo{volume}{11}, \bibinfo{number}{4} (\bibinfo{year}{1997}), \bibinfo{pages}{369--379}.
\newblock
\urldef\tempurl%
\url{https://doi.org/10.1002/(SICI)1098-2418(199712)11:4<369::AID-RSA5>3.0.CO;2-V}
\showDOI{\tempurl}


\bibitem[Holm et~al\mbox{.}(2001)]%
        {holm2001polylog}
\bibfield{author}{\bibinfo{person}{Jacob Holm}, \bibinfo{person}{Kristian de Lichtenberg}, {and} \bibinfo{person}{Mikkel Thorup}.} \bibinfo{year}{2001}\natexlab{}.
\newblock \showarticletitle{Poly-Logarithmic Deterministic Fully-Dynamic Algorithms for Connectivity, Minimum Spanning Tree, 2-Edge, and Biconnectivity}.
\newblock \bibinfo{journal}{\emph{J. ACM}} \bibinfo{volume}{48}, \bibinfo{number}{4} (\bibinfo{year}{2001}), \bibinfo{pages}{723--760}.
\newblock
\urldef\tempurl%
\url{https://doi.org/10.1145/502090.502095}
\showDOI{\tempurl}


\bibitem[Holm et~al\mbox{.}(2018)]%
        {holm2018dynamic}
\bibfield{author}{\bibinfo{person}{Jacob Holm}, \bibinfo{person}{Eva Rotenberg}, {and} \bibinfo{person}{Mikkel Thorup}.} \bibinfo{year}{2018}\natexlab{}.
\newblock \showarticletitle{Dynamic Bridge-Finding in $\tilde{O}(\log^2 n)$ Amortized Time}. In \bibinfo{booktitle}{\emph{Proceedings of the 29th Annual ACM-SIAM Symposium on Discrete Algorithms (SODA)}}. \bibinfo{pages}{35--52}.
\newblock
\urldef\tempurl%
\url{https://doi.org/10.1137/1.9781611975031.3}
\showDOI{\tempurl}


\bibitem[Huang et~al\mbox{.}(2023)]%
        {huang2023fully}
\bibfield{author}{\bibinfo{person}{Shang-En Huang}, \bibinfo{person}{Dawei Huang}, \bibinfo{person}{Tsvi Kopelowitz}, \bibinfo{person}{Seth Pettie}, {and} \bibinfo{person}{Mikkel Thorup}.} \bibinfo{year}{2023}\natexlab{}.
\newblock \showarticletitle{Fully Dynamic Connectivity in $O(\log n (\log\log n)^2)$ Amortized Expected Time}.
\newblock \bibinfo{journal}{\emph{Theoretics}}  \bibinfo{volume}{2} (\bibinfo{year}{2023}), \bibinfo{pages}{6:1--6:56}.
\newblock
\urldef\tempurl%
\url{https://doi.org/10.46298/theoretics.23.6}
\showDOI{\tempurl}


\bibitem[Kapron et~al\mbox{.}(2013)]%
        {kapron2013dynamic}
\bibfield{author}{\bibinfo{person}{Bruce~M. Kapron}, \bibinfo{person}{Valerie King}, {and} \bibinfo{person}{Ben Mountjoy}.} \bibinfo{year}{2013}\natexlab{}.
\newblock \showarticletitle{Dynamic graph connectivity in polylogarithmic worst case time}. In \bibinfo{booktitle}{\emph{Proceedings of the Twenty-Fourth Annual ACM-SIAM Symposium on Discrete Algorithms (SODA)}}. SIAM, \bibinfo{pages}{1131--1142}.
\newblock


\bibitem[Karp et~al\mbox{.}(2003)]%
        {karpStreamingAlgorithms2003}
\bibfield{author}{\bibinfo{person}{Richard~M. Karp}, \bibinfo{person}{Scott Shenker}, {and} \bibinfo{person}{Christos~H. Papadimitriou}.} \bibinfo{year}{2003}\natexlab{}.
\newblock \showarticletitle{A Simple Algorithm for Finding Frequent Elements in Streams and Its Application to Load Balancing}. In \bibinfo{booktitle}{\emph{Proceedings of the 15th Annual ACM Symposium on Parallelism in Algorithms and Architectures (SPAA)}}. \bibinfo{pages}{1--10}.
\newblock
\urldef\tempurl%
\url{https://doi.org/10.1145/777412.777415}
\showDOI{\tempurl}


\bibitem[Kitsak et~al\mbox{.}(2010)]%
        {kitsak2010identification}
\bibfield{author}{\bibinfo{person}{Maksim Kitsak}, \bibinfo{person}{Lazaros~K. Gallos}, \bibinfo{person}{Shlomo Havlin}, \bibinfo{person}{Fredrik Liljeros}, \bibinfo{person}{Lev Muchnik}, \bibinfo{person}{H.~Eugene Stanley}, {and} \bibinfo{person}{Hern{\'a}n~A. Makse}.} \bibinfo{year}{2010}\natexlab{}.
\newblock \showarticletitle{Identification of Influential Spreaders in Complex Networks}.
\newblock \bibinfo{journal}{\emph{Nature Physics}} \bibinfo{volume}{6}, \bibinfo{number}{11} (\bibinfo{year}{2010}), \bibinfo{pages}{888--893}.
\newblock
\urldef\tempurl%
\url{https://doi.org/10.1038/nphys1746}
\showDOI{\tempurl}


\bibitem[Kwak et~al\mbox{.}(2010)]%
        {twittergraph}
\bibfield{author}{\bibinfo{person}{Haewoon Kwak}, \bibinfo{person}{Changhyun Lee}, \bibinfo{person}{Hosung Park}, {and} \bibinfo{person}{Sue Moon}.} \bibinfo{year}{2010}\natexlab{}.
\newblock \showarticletitle{What is {T}witter, a Social Network or a News Media?}. In \bibinfo{booktitle}{\emph{Proceedings of the 19th International Conference on World Wide Web (WWW)}}.
\newblock


\bibitem[McAuley and Leskovec(2012)]%
        {mcauley2012learning}
\bibfield{author}{\bibinfo{person}{Julian~J McAuley} {and} \bibinfo{person}{Jure Leskovec}.} \bibinfo{year}{2012}\natexlab{}.
\newblock \showarticletitle{Learning to discover social circles in ego networks.}. In \bibinfo{booktitle}{\emph{NIPS}}, Vol.~\bibinfo{volume}{2012}. Citeseer, \bibinfo{pages}{548--56}.
\newblock


\bibitem[McGregor(2014)]%
        {mcgregor2014graph}
\bibfield{author}{\bibinfo{person}{Andrew McGregor}.} \bibinfo{year}{2014}\natexlab{}.
\newblock \showarticletitle{Graph stream algorithms: a survey}.
\newblock \bibinfo{journal}{\emph{ACM SIGMOD Record}} \bibinfo{volume}{43}, \bibinfo{number}{1} (\bibinfo{year}{2014}), \bibinfo{pages}{9--20}.
\newblock


\bibitem[Mislove et~al\mbox{.}(2007)]%
        {mislove2007measurement}
\bibfield{author}{\bibinfo{person}{Alan Mislove}, \bibinfo{person}{Massimiliano Marcon}, \bibinfo{person}{Krishna~P. Gummadi}, \bibinfo{person}{Peter Druschel}, {and} \bibinfo{person}{Bobby Bhattacharjee}.} \bibinfo{year}{2007}\natexlab{}.
\newblock \showarticletitle{Measurement and Analysis of Online Social Networks}. In \bibinfo{booktitle}{\emph{Proceedings of the 7th ACM SIGCOMM Conference on Internet Measurement (IMC)}}. \bibinfo{pages}{29--42}.
\newblock
\urldef\tempurl%
\url{https://doi.org/10.1145/1298306.1298311}
\showDOI{\tempurl}


\bibitem[Misra and Gries(1982)]%
        {misraGriesSummary1982}
\bibfield{author}{\bibinfo{person}{Jayadev Misra} {and} \bibinfo{person}{David Gries}.} \bibinfo{year}{1982}\natexlab{}.
\newblock \showarticletitle{Finding Repeated Elements}.
\newblock \bibinfo{journal}{\emph{Science of Computer Programming}} \bibinfo{volume}{2}, \bibinfo{number}{2} (\bibinfo{year}{1982}), \bibinfo{pages}{143--152}.
\newblock
\urldef\tempurl%
\url{https://doi.org/10.1016/0167-6423(82)90012-0}
\showDOI{\tempurl}


\bibitem[Muthukrishnan(2005)]%
        {semistreaming2}
\bibfield{author}{\bibinfo{person}{S. Muthukrishnan}.} \bibinfo{year}{2005}\natexlab{}.
\newblock \showarticletitle{Data Streams: Algorithms and Applications}.
\newblock \bibinfo{journal}{\emph{Foundations and Trends{\textregistered} in Theoretical Computer Science}} \bibinfo{volume}{1}, \bibinfo{number}{2} (\bibinfo{year}{2005}), \bibinfo{pages}{117--236}.
\newblock
\showISSN{1551-305X}
\urldef\tempurl%
\url{https://doi.org/10.1561/0400000002}
\showDOI{\tempurl}


\bibitem[Nanongkai et~al\mbox{.}(2017)]%
        {nanongkai2017dynamic}
\bibfield{author}{\bibinfo{person}{Danupon Nanongkai}, \bibinfo{person}{Thatchaphol Saranurak}, {and} \bibinfo{person}{Christian Wulff-Nilsen}.} \bibinfo{year}{2017}\natexlab{}.
\newblock \showarticletitle{Dynamic Minimum Spanning Forest with Subpolynomial Worst-Case Update Time}. In \bibinfo{booktitle}{\emph{IEEE 58th Annual Symposium on Foundations of Computer Science (FOCS)}}. \bibinfo{pages}{950--961}.
\newblock
\urldef\tempurl%
\url{https://doi.org/10.1109/FOCS.2017.92}
\showDOI{\tempurl}


\bibitem[Nelson and Yu(2019)]%
        {nelson2019optimal}
\bibfield{author}{\bibinfo{person}{Jelani Nelson} {and} \bibinfo{person}{Huacheng Yu}.} \bibinfo{year}{2019}\natexlab{}.
\newblock \showarticletitle{Optimal lower bounds for distributed and streaming spanning forest computation}. In \bibinfo{booktitle}{\emph{Proceedings of the Thirtieth Annual ACM-SIAM Symposium on Discrete Algorithms}}. SIAM, \bibinfo{pages}{1844--1860}.
\newblock


\bibitem[Newman(2003)]%
        {newman2003structure}
\bibfield{author}{\bibinfo{person}{M.~E.~J. Newman}.} \bibinfo{year}{2003}\natexlab{}.
\newblock \showarticletitle{The Structure and Function of Complex Networks}.
\newblock \bibinfo{journal}{\emph{SIAM Rev.}} \bibinfo{volume}{45}, \bibinfo{number}{2} (\bibinfo{year}{2003}), \bibinfo{pages}{167--256}.
\newblock
\urldef\tempurl%
\url{https://doi.org/10.1137/S003614450342480}
\showDOI{\tempurl}


\bibitem[{OpenStreetMap contributors}(2017)]%
        {roadgraphs}
\bibfield{author}{\bibinfo{person}{{OpenStreetMap contributors}}.} \bibinfo{year}{2017}\natexlab{}.
\newblock \bibinfo{title}{{Planet dump retrieved from https://planet.osm.org }}.
\newblock \bibinfo{howpublished}{\url{ https://www.openstreetmap.org }}.
\newblock


\bibitem[Rossi and Ahmed(2015)]%
        {netrepo}
\bibfield{author}{\bibinfo{person}{Ryan~A. Rossi} {and} \bibinfo{person}{Nesreen~K. Ahmed}.} \bibinfo{year}{2015}\natexlab{}.
\newblock \showarticletitle{The Network Data Repository with Interactive Graph Analytics and Visualization}. In \bibinfo{booktitle}{\emph{AAAI}}.
\newblock
\urldef\tempurl%
\url{https://networkrepository.com}
\showURL{%
\tempurl}


\bibitem[Seidman(1983)]%
        {seidman1983network}
\bibfield{author}{\bibinfo{person}{Stephen~B. Seidman}.} \bibinfo{year}{1983}\natexlab{}.
\newblock \showarticletitle{Network Structure and Minimum Degree}.
\newblock \bibinfo{journal}{\emph{Social Networks}} \bibinfo{volume}{5}, \bibinfo{number}{3} (\bibinfo{year}{1983}), \bibinfo{pages}{269--287}.
\newblock
\urldef\tempurl%
\url{https://doi.org/10.1016/0378-8733(83)90028-X}
\showDOI{\tempurl}


\bibitem[Shaked and Shanthikumar(2007)]%
        {shakedStochasticOrders2007}
\bibfield{editor}{\bibinfo{person}{Moshe Shaked} {and} \bibinfo{person}{J.~George Shanthikumar}} (Eds.). \bibinfo{year}{2007}\natexlab{}.
\newblock \bibinfo{booktitle}{\emph{Stochastic {{Orders}}}}.
\newblock \bibinfo{publisher}{Springer}, \bibinfo{address}{New York, NY}.
\newblock
\showISBNx{978-0-387-32915-4}
\urldef\tempurl%
\url{https://doi.org/10.1007/978-0-387-34675-5}
\showDOI{\tempurl}


\bibitem[Simhadri et~al\mbox{.}(2021)]%
        {simhadri2021results}
\bibfield{author}{\bibinfo{person}{Harsha~Vardhan Simhadri}, \bibinfo{person}{George Williams}, \bibinfo{person}{Madlib-Abhimanyu}, {et~al\mbox{.}}} \bibinfo{year}{2021}\natexlab{}.
\newblock \showarticletitle{Results of the {NeurIPS}'21 Challenge on Billion-Scale Approximate Nearest Neighbor Search}. In \bibinfo{booktitle}{\emph{NeurIPS 2021 Competitions and Demonstrations Track}}. \bibinfo{publisher}{PMLR}.
\newblock


\bibitem[Sleator and Tarjan(1983)]%
        {sleator1983data}
\bibfield{author}{\bibinfo{person}{Daniel~D Sleator} {and} \bibinfo{person}{Robert~Endre Tarjan}.} \bibinfo{year}{1983}\natexlab{}.
\newblock \showarticletitle{A data structure for dynamic trees}.
\newblock \bibinfo{journal}{\emph{J. Comput. System Sci.}} \bibinfo{volume}{26}, \bibinfo{number}{3} (\bibinfo{year}{1983}).
\newblock


\bibitem[Sleator and Tarjan(1985)]%
        {sleator1985self}
\bibfield{author}{\bibinfo{person}{Daniel~D. Sleator} {and} \bibinfo{person}{Robert~E. Tarjan}.} \bibinfo{year}{1985}\natexlab{}.
\newblock \showarticletitle{Self-Adjusting Binary Search Trees}.
\newblock \bibinfo{journal}{\emph{J. ACM}} \bibinfo{volume}{32}, \bibinfo{number}{3} (\bibinfo{year}{1985}), \bibinfo{pages}{652--686}.
\newblock
\urldef\tempurl%
\url{https://doi.org/10.1145/3828.3835}
\showDOI{\tempurl}


\bibitem[Tench et~al\mbox{.}(2022)]%
        {graphzeppelin}
\bibfield{author}{\bibinfo{person}{David Tench}, \bibinfo{person}{Evan West}, \bibinfo{person}{Victor Zhang}, \bibinfo{person}{Michael~A. Bender}, \bibinfo{person}{Abiyaz Chowdhury}, \bibinfo{person}{J.~Ahmed Dellas}, \bibinfo{person}{Martin Farach-Colton}, \bibinfo{person}{Tyler Seip}, {and} \bibinfo{person}{Kenny Zhang}.} \bibinfo{year}{2022}\natexlab{}.
\newblock \showarticletitle{GraphZeppelin: Storage-Friendly Sketching for Connected Components on Dynamic Graph Streams}. In \bibinfo{booktitle}{\emph{Proceedings of the 2022 International Conference on Management of Data}} (Philadelphia, PA, USA) \emph{(\bibinfo{series}{SIGMOD '22})}. \bibinfo{publisher}{Association for Computing Machinery}, \bibinfo{address}{New York, NY, USA}, \bibinfo{pages}{325–339}.
\newblock
\showISBNx{9781450392495}
\urldef\tempurl%
\url{https://doi.org/10.1145/3514221.3526146}
\showDOI{\tempurl}


\bibitem[Tench et~al\mbox{.}(2024)]%
        {landscape}
\bibfield{author}{\bibinfo{person}{David Tench}, \bibinfo{person}{Evan~T. West}, \bibinfo{person}{Kenny Zhang}, \bibinfo{person}{Michael Bender}, \bibinfo{person}{Daniel DeLayo}, \bibinfo{person}{Martin Farach-Colton}, \bibinfo{person}{Gilvir Gill}, \bibinfo{person}{Tyler Seip}, {and} \bibinfo{person}{Victor Zhang}.} \bibinfo{year}{2024}\natexlab{}.
\newblock \bibinfo{title}{Exploring the Landscape of Distributed Graph Sketching}.
\newblock
\newblock
\showeprint[arxiv]{2410.07518}~[cs.DC]
\urldef\tempurl%
\url{https://arxiv.org/abs/2410.07518}
\showURL{%
\tempurl}


\bibitem[Thorup(2000)]%
        {thorup2000near}
\bibfield{author}{\bibinfo{person}{Mikkel Thorup}.} \bibinfo{year}{2000}\natexlab{}.
\newblock \showarticletitle{Near-Optimal Fully-Dynamic Graph Connectivity}. In \bibinfo{booktitle}{\emph{Proceedings of the 32nd Annual ACM Symposium on Theory of Computing (STOC)}}. \bibinfo{pages}{343--350}.
\newblock
\urldef\tempurl%
\url{https://doi.org/10.1145/335305.335345}
\showDOI{\tempurl}


\bibitem[Woodruff(2014)]%
        {woodruffSketchingToolNumerical2014}
\bibfield{author}{\bibinfo{person}{David~P. Woodruff}.} \bibinfo{year}{2014}\natexlab{}.
\newblock \showarticletitle{Sketching as a {{Tool}} for {{Numerical Linear Algebra}}}.
\newblock \bibinfo{journal}{\emph{Foundations and Trends{\textregistered} in Theoretical Computer Science}} \bibinfo{volume}{10}, \bibinfo{number}{1-2} (\bibinfo{year}{2014}), \bibinfo{pages}{1--157}.
\newblock
\showISSN{1551-305X, 1551-3068}
\urldef\tempurl%
\url{https://doi.org/10.1561/0400000060}
\showDOI{\tempurl}
\showeprint[arxiv]{1411.4357}~[cs]


\bibitem[Wulff-Nilsen(2013)]%
        {wulff2013faster}
\bibfield{author}{\bibinfo{person}{Christian Wulff-Nilsen}.} \bibinfo{year}{2013}\natexlab{}.
\newblock \showarticletitle{Faster deterministic fully-dynamic graph connectivity}. In \bibinfo{booktitle}{\emph{Proceedings of the twenty-fourth Annual ACM-SIAM Symposium on Discrete Algorithms}}. SIAM, \bibinfo{pages}{1757--1769}.
\newblock


\bibitem[Xu et~al\mbox{.}(2024)]%
        {IDTree}
\bibfield{author}{\bibinfo{person}{Lantian Xu}, \bibinfo{person}{Dong Wen}, \bibinfo{person}{Lu Qin}, \bibinfo{person}{Ronghua Li}, \bibinfo{person}{Ying Zhang}, {and} \bibinfo{person}{Xuemin Lin}.} \bibinfo{year}{2024}\natexlab{}.
\newblock \showarticletitle{Constant-time Connectivity Querying in Dynamic Graphs}.
\newblock \bibinfo{journal}{\emph{Proc. ACM Manag. Data}} \bibinfo{volume}{2}, \bibinfo{number}{6}, Article \bibinfo{articleno}{230} (\bibinfo{date}{Dec.} \bibinfo{year}{2024}), \bibinfo{numpages}{23}~pages.
\newblock
\urldef\tempurl%
\url{https://doi.org/10.1145/3698805}
\showDOI{\tempurl}


\bibitem[Yang and Leskovec(2012)]%
        {yang2012communities}
\bibfield{author}{\bibinfo{person}{Jaewon Yang} {and} \bibinfo{person}{Jure Leskovec}.} \bibinfo{year}{2012}\natexlab{}.
\newblock \showarticletitle{Defining and Evaluating Network Communities based on Ground-truth}. In \bibinfo{booktitle}{\emph{IEEE International Conference on Data Mining (ICDM)}}.
\newblock


\bibitem[Yu et~al\mbox{.}(2023)]%
        {parlayann2023}
\bibfield{author}{\bibinfo{person}{Shangdi Yu}, \bibinfo{person}{Yan Gu}, {and} \bibinfo{person}{Julian Shun}.} \bibinfo{year}{2023}\natexlab{}.
\newblock \showarticletitle{Parlay{ANN}: Scalable and Deterministic Algorithms for Approximate Nearest Neighbor Search}. In \bibinfo{booktitle}{\emph{Proceedings of the 28th ACM SIGPLAN Annual Symposium on Principles and Practice of Parallel Programming (PPoPP)}}.
\newblock


\end{thebibliography}
